\documentclass[12pt,draftcls,peerreviewca,onecolumn]{IEEEtran}

\usepackage{geometry}
\usepackage{amsfonts}
\usepackage{cite,graphicx,amsmath,amssymb}
\usepackage{subfigure}
\usepackage{fancyhdr}
\usepackage{dsfont}

\usepackage{array,color}
\usepackage{amsmath}
\newtheorem{remark}{Remark}

\newtheorem{challenge}{Challenge}

\newcommand{\siml}{\stackrel{\beta\to\infty}{\sim}}
\newcommand{\sims}{\stackrel{\beta\to0}{\sim}}

\newcommand{\dis}{\stackrel{d}{\sim}}
\newcommand{\eqla}{\stackrel{(a)}{=}}
\newcommand{\eqlb}{\stackrel{(b)}{=}}
\newcommand{\eqlc}{\stackrel{(c)}{=}}
\newcommand{\eqld}{\stackrel{(d)}{=}}
\newcommand{\eqle}{\stackrel{(e)}{=}}
\newcommand{\eqlf}{\stackrel{(f)}{=}}

\newtheorem{theorem}{Theorem}

\newtheorem{lemma}{Lemma}

\newtheorem{corollary}{Corollary}

\newtheorem{proposition}{Proposition}

\newcommand{\define}{\stackrel{\Delta}{=}}

\newcommand{\papertitle}{User-Centric Interference Nulling in Downlink Multi-Antenna Heterogeneous Networks}

\IEEEoverridecommandlockouts

\begin{document}
\title{\papertitle}
\author{Ying Cui, \hspace{5mm}Yueping Wu, \hspace{5mm}Dongdong Jiang,
\hspace{5mm}Bruno Clerckx\\
\thanks{
Y. Cui and D. Jiang are with the Department of  Electronic Engineering, Shanghai Jiao Tong University, China.
Y.~Wu and B.~Clerckx are with Department of Electrical and Electronic Engineering, Imperial College London, United Kingdom. B.~Clerckx is also with School of Electrical Engineering, Korea University, Korea. This work was presented in part in IEEE ISIT 2015.}
}
\maketitle


\begin{abstract}
In heterogeneous networks (HetNets), strong interference due to spectrum reuse affects each user's  signal-to-interference ratio (SIR), and hence is one limiting factor of network performance.
In this paper, we propose a user-centric interference nulling (IN) scheme in a downlink large-scale HetNet to improve coverage/outage probability by   improving each user's SIR. This IN scheme 
utilizes at most  {\em maximum IN degree of freedom (DoF)} at each macro-BS  to avoid interference to uniformly selected macro (pico) users with signal-to-individual-interference ratio (SIIR) below a {\em macro (pico) IN threshold}, where the maximum IN DoF and the two IN thresholds  are three design parameters. 
Using tools from stochastic geometry, we first obtain a tractable expression of the coverage (equivalently outage) probability. Then, we analyze the asymptotic coverage/outage probability in the low and high SIR threshold regimes. The  analytical results indicate 
that the maximum IN DoF can affect the order gain of the  outage probability in the low SIR threshold regime, but cannot affect the order gain of the   coverage probability in the high SIR threshold regime. 
Moreover, we characterize the optimal maximum IN DoF which optimizes the asymptotic coverage/outage probability. The optimization results reveal that 
the IN scheme can linearly improve the  outage probability  in the low SIR threshold regime, but cannot improve the  coverage  probability in the high SIR threshold regime. 
Finally, numerical results show that the proposed scheme can achieve good gains  in coverage/outage probability  over a maximum ratio beamforming scheme  and a user-centric almost blank subframes (ABS) scheme.
\end{abstract}


\newpage

\section{Introduction}

Heterogenous wireless networks (HetNets), i.e., the deployment of low power small cell  base stations (BSs)  overlaid with conventional large power macro-BSs, provide a powerful approach to meet the massive growth in traffic demands by aggressively reusing existing spectrum assets \cite{zhang11,Ghosh12}. 
However, spectrum reuse in HetNets causes strong interference. This affects the signal-to-interference ratio (SIR) of each user, and hence is one of the limiting factors of network performance. Interference management techniques are thus desirable in HetNets\cite{Leemagazine12}. One such technique is interference cooperation. For example,  in \cite{nigam14,nie14,sakr14}, different interference cooperation strategies are considered and their performances are analyzed for large-scale HetNets under random models  in which the locations of  BSs and users  are spatially distributed as independent homogeneous Poisson point processes (PPPs)\cite{andrews11,andrews12}.
However, in \cite{nigam14,nie14,sakr14}, the cooperation clusters are formed to favor  a typical user located at the origin of the network (referred to as the typical user)  only, and hence, the analytical performance of the typical user is better than the actual network performance (of all the users). In addition, \cite{nigam14,nie14,sakr14} only consider single-antenna BSs. Orthogonalizing the time or frequency resource allocated to macro cells and small cells can also mitigate interference in HetNets.  One such technique is almost blank subframes (ABS) in 3GPP LTE \cite{singh13}. In ABS, the time or frequency resource is partitioned, whereby  offloaded users and the other users are served using different portions of the resource  in HetNets with offloading. The performance of ABS in large-scale HetNets with offloading is analyzed in \cite{singh13} using tools from stochastic geometry. Note that ABS focuses on mitigating the interference of offloaded users, and
\cite{singh13} only considers single-antenna BSs.

Deploying multiple antennas at each BS in HetNets can further improve network performance. With multiple antennas, besides boosting   signals to desired users, more effective interference management techniques can be implemented \cite{Adhikary14,Hosseini13,Kountouris13,NguyenTWC13,wu14,xia13}.
For example,  in \cite{Adhikary14,Hosseini13,Kountouris13,NguyenTWC13}, the authors consider HetNets with a single multi-antenna macro-BS and multiple multi-antenna small-BSs, where the multiple antennas at the macro-BS are used for serving its scheduled users as well as mitigating the interference to  some small cell users using different interference coordination schemes. These schemes are analyzed and shown to improve the network performance.  In particular,  \cite{NguyenTWC13} also considers multiple antennas at each user, and proposes an opportunistic interference alignment scheme to design the transmit and receive beamformers to mitigate interference. Each small BS is assumed to have a different nearest victim small user, and  victim user selection is not considered.   Note that since only one macro-BS is considered in \cite{Adhikary14,Hosseini13,Kountouris13,NguyenTWC13}, the analytical results obtained in \cite{Adhikary14,Hosseini13,Kountouris13,NguyenTWC13} cannot reflect the macro-tier interference, and thus may not offer accurate insights for practical HetNets.  In \cite{wu14,xia13}, large-scale multi-antenna HetNets are considered.
 Specifically, \cite{wu14} considers  offloading, and proposes an interference nulling  (IN) scheme where some degree of freedom at each macro-BS can be used for avoiding its interference to some of its offloaded users.  The rate coverage probability is analyzed and optimized by optimizing the amount of degree of freedom (DoF) for interference nulling. However, the IN scheme proposed in \cite{wu14}  only improves the performance of  scheduled  offloaded users,   and scheduled offloaded users are selected by the corresponding macro-BS for interference nulling with equal probability. Hence,  the IN scheme proposed in \cite{wu14} may not effectively improve the overall  rate coverage probability. In \cite{xia13}, a fixed number of BSs which provide the strongest average received power for the typical user form a cluster, and adopt an interference coordination scheme where the BSs in each cluster mitigate interference to users in this cluster.  The coverage probability is analyzed based on the assumption that the BSs in each cluster are the strongest BSs of all the users in this cluster.

The investigation of interference management techniques  in large-scale single-tier  multi-antenna cellular networks is less involved than that in large-scale multi-antenna HetNets, and hence has been more extensively conducted  (see \cite{Akoum13,huang13,Ying11,li14IC}  and the references therein).
 In \cite{Akoum13,huang13,Ying11}, all the BSs are grouped into disjoint clusters. Coordination \cite{Akoum13,huang13} and cooperation \cite{Ying11} are performed among  the BSs within each cluster to mitigate intra-cluster interference. Specifically, \cite{Akoum13} and \cite{huang13} design  disjoint BS clustering from a transmitter's point of view and fail to consider each user's interference situation. The dynamic clustering proposed in \cite{Ying11} considers all the users' signal and interference situations  to optimize the network performance. However, it requires centralized control and may not be suitable for large networks.
Recently, a novel distributed user-centric IN scheme, which takes account of each user's desired signal strength  and interference level, is proposed and analyzed for  (single-tier) multi-antenna small cell networks  in \cite{li14IC}. However, in \cite{li14IC}, the maximum DoF for IN (i.e., maximum IN DoF) at each BS is not adjustable, and  thus cannot properly utilize   resource in small cell networks. Moreover, directly applying the scheme in \cite{li14IC} to HetNets cannot fully exploit different properties of  macro and pico users in HetNets. 

In this paper, we consider a downlink large-scale two-tier multi-antenna HetNet and propose a user-centric IN scheme to improve  the coverage probability  by improving each user's SIR.   This scheme has three design parameters: the maximum IN DoF $U$, and the IN thresholds for  macro and pico users, respectively.  In this scheme, each scheduled  macro (pico) user first sends an IN request to a macro-BS\footnote{Note that, compared to a pico-BS, a macro-BS usually causes stronger interference due to larger transmit power, and has a better capability of performing spatial cancellation due to  a larger number of transmit antennas. Thus, it is more advisable to perform IN at macro-BSs.} if the power ratio of its desired signal and the interference from the macro-BS, referred to as the signal-to-individual-interference ratio (SIIR), is below the IN threshold for macro (pico) users. Then, each macro-BS utilizes zero-forcing beamforming (ZFBF) precoder to avoid interference to at most $U$ scheduled users which send IN requests to it as well as boost the desired signal to its scheduled user.  In general, the performance analysis and optimization of interference management techniques in large-scale multi-antenna HetNets are very challenging, mainly due to i) the statistical dependence among macro-BSs and pico-BSs \cite{Adhikary14}, ii) the complex distribution of a desired signal using multi-antenna communication schemes, and iii) the complicated interference distribution caused by interference management techniques (e.g., beamforming).  Our main contributions are summarized below. The analytical and numerical  results obtained in this paper provide valuable design insights for practical HetNets.
\begin{itemize}
\item We obtain a tractable expression of the coverage (equivalently outage) probability, by adopting appropriate approximations and utilizing tools from stochastic geometry.
\item We obtain the asymptotic expressions of the coverage/outage probability in the  low and high SIR threshold regimes, using series  expansions of special functions. The  analytical results indicate that 
the maximum IN DoF can affect the order gain of the outage probability in the low SIR threshold regime, but cannot affect the order gain of the  coverage probability in the high SIR threshold regime; the IN thresholds only affect the coefficients of the  coverage/outage probability  in the low and high SIR threshold regimes.
\item  We consider the optimizations of the maximum IN DoF for given IN thresholds in the two asymptotic  regimes, which are challenging  integer programming problems with very complicated objective functions.  By exploiting the structure of each objective function, we characterize the optimal maximum IN DoF.
The optimization results reveal that the IN scheme can linearly improve the  outage probability in the low SIR threshold regime, but cannot improve the  coverage probability in the high SIR threshold regime.
\item We  show that  the  IN scheme can achieve good gains in coverage/outage probability over a maximum ratio beamforming scheme and a user-centric  ABS scheme, using numerical results.
\end{itemize}

The key notations used in the paper are listed in Table \ref{tab:para_system}. 
\begin{table}[t]
\caption{Key notations.}\label{tab:para_system}
\begin{center}
\vspace{-10mm}
\begin{scriptsize}
\begin{tabular}{|c!{\vrule width 1.5pt}c|}
\hline
Notation&Description\rule{0pt}{3mm}\\
\hline
$\Phi_{j}$, $\Phi_{u}$ & PPP of BSs in the $j$th tier, PPP of users\\
\hline
$\lambda_{j}$, $\lambda_{u}$& Density of PPP $\Phi_{j}$, density of PPP $\Phi_{u}$\\
\hline
$P_{j}$, $N_{j}$& Transmit power at each BS in the $j$th tier, number of transmit antennas at each BS in the $j$th tier\\
\hline
$\alpha_{j}$& Path loss exponent in the $j$th tier\\
\hline
$\mathcal{U}_j$ &Set of macro-users ($j=1$), set of pico-users ($j=2$)\\
\hline
$Y_j$&Distance between the typical user and its serving BS in the $j$th tier\\
\hline
$\mathcal{A}_j$ & Association probability of the typical user to $\mathcal{U}_j$\\
\hline
$K_0$ & Number of the potential IN users of an arbitrary macro-BS\\
\hline
$\mathcal{S}$, $\beta$ & SIR coverage probability, SIR threshold\\
\hline
$U$, $T_j$ & Maximum IN DoF, IN threshold for the $j$th tier  in the IN scheme\\
\hline
\end{tabular}
\end{scriptsize}
\vspace{-6mm}
\end{center}
\end{table}


\section{Network Model}
We consider a two-tier HetNet where a macro-cell tier is overlaid with a pico-cell tier, as shown in Fig.\ \ref{fig:model}. The locations of  macro-BSs and pico-BSs are spatially distributed as two independent homogeneous Poisson point processes (PPPs) $\Phi_{1}$ and $\Phi_{2}$ with densities $\lambda_{1}$ and $\lambda_{2}$, respectively. The locations of users are also distributed as an independent homogeneous PPP $\Phi_{u}$ with density $\lambda_{u}$. Without loss of generality, denote the macro-cell tier as the $1$st tier and the pico-cell tier as the $2$nd tier. We focus on the downlink scenario. The macro-BSs and the pico-BSs share the same spectrum concurrently.  Each macro-BS has $N_{1}$ antennas with total transmission power $P_{1}$, each pico-BS has $N_{2}$ antennas with total transmission power $P_{2}$, and each user has a single antenna. Assume $N_{1}>N_{2}$. We consider both large-scale fading and small-scale fading. Specifically, due to large-scale fading, transmitted signals from the $j$th tier with distance $r$ are attenuated by a factor $r^{-\alpha_{j}}$, where $\alpha_{j}>2$ is the path loss exponent of the $j$th tier and $j=1,2$. For small-scale fading, we assume Rayleigh fading channels.

\subsection{User Association}\label{subset:user_asso}
We assume open access \cite{nigam14}. User $i$ (denoted as $u_{i}$) is associated with the BS which provides the maximum \emph{long-term average}  (over small-scale fading) received power among all the macro-BSs and pico-BSs. This associated BS is called the \emph{serving BS} of user $i$. Note that within each tier, the nearest BS to user $i$ provides the strongest long-term average  received power in this tier. User $i$ is thus associated with (the nearest BS in) the $j^{*}_{i}$th tier, if\footnote{In the user association procedure, the first antenna is normally used to transmit signal (using the total transmission power of each BS) for received power determination according to LTE standards.} $j_{i}^{*}=  {\arg\:\max}_{j\in\{1,2\}}P_{j}Z_{i,j}^{-\alpha_{j}}$,
where $Z_{i,j}$ is the distance between user $i$ and its nearest BS in the $j$th tier. We refer to the users associated with the macro-cell tier as the \emph{macro-users}, denoted as $\mathcal{U}_{1}\triangleq\left\{u_{i}|P_{1}Z_{i,1}^{-\alpha_{1}}\ge P_{2}Z_{i,2}^{-\alpha_{2}}\right\}$, and the users associated with the pico-cell tier as the \emph{pico-users}, denoted as $\mathcal{U}_{2}\triangleq\left\{u_{i}|P_{2}Z_{i,2}^{-\alpha_{2}}>P_{1}Z_{i,1}^{-\alpha_{1}}\right\}$. All the users can be partitioned into two disjoint user sets: $\mathcal{U}_{1}$ and $\mathcal{U}_{2}$. After the user association, each BS schedules its associated users according to TDMA, i.e., scheduling  one user in each time slot. Hence, there is no intra-cell interference.


\subsection{Performance Metric}
In this paper, we study the performance of the typical user denoted as $u_{0}$, which is a scheduled user located at the origin \cite{Tanbourgi14}. Since HetNets are interference-limited, we ignore the thermal noise in the analysis of this paper. Note that the analytical results with thermal noise can be obtained in a similar way\cite{Hunter08TWC}. The \emph{coverage probability} of $u_{0}$  is defined as the probability that the SIR of $u_{0}$ is larger than a threshold \cite{nigam14}, i.e., 
\small{\begin{align}\label{eq:CP_def}
\mathcal{S}(\beta)&\define{\rm Pr}\left({\rm SIR}_{0}>\beta\right)
\end{align}}\normalsize where $\beta$ is the SIR threshold.  The \emph{outage probability} of $u_{0}$  is defined as the probability that the SIR of $u_{0}$ is smaller than or equal to a threshold, i.e., $1-\mathcal{S}(\beta)$.  The coverage/outage probability provides the cumulative probability function (c.d.f.) of the random SIR over the entire
network\cite{andrews11}. In Sections~\ref{sec:CP_general}, \ref{sec:CP_low} and \ref{sec:CP_high}, we shall analyze the coverage/outage probability in the general, low and high SIR threshold regimes, separately. 

\section{User-centric Interference Nulling Scheme}
In this section, we first elaborate on a user-centric IN scheme. 
Then, we obtain some distributions related to this scheme.

\begin{figure}[t] \centering
\includegraphics[width=0.35\columnwidth]{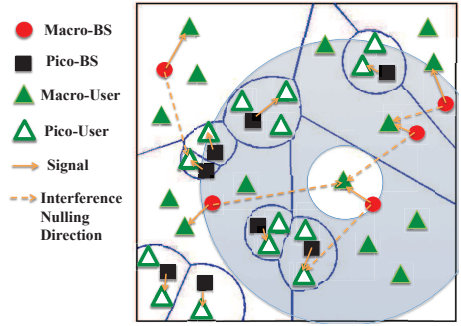}
\caption{\small{System Model ($U=1$).}}\label{fig:model}
\end{figure}

\subsection{Scheme Description}\label{subsec:IN_descp}
First, we refer to an interfering macro-BS which causes the SIIR at  scheduled user $i$ in the $j$th tier ($u_{i}\in\mathcal{U}_{j}$) below threshold $T_{j}\geq1$ as a \emph{potential IN macro-BS} of $u_{i}$, where $j=1,2$.  We refer to $T_{j}$ as the {\em IN threshold} for the $j$th tier. Mathematically, interfering macro-BS $\ell$ is a potential IN macro-BS of scheduled user $u_{i}\in\mathcal{U}_{j}$ if $\frac{P_{j}Z_{i,j}^{-\alpha_{j}}}{P_{1}D_{1,\ell i}^{-\alpha_{1}}}<T_{j}$, where $D_{1,\ell i}$ is the distance from  macro-BS $\ell$ to $u_{i}$. Note that $T_{1}$ and $T_{2}$ are two design parameters of the IN scheme. In each time slot, each scheduled user sends IN requests to all of its potential IN macro-BSs. We refer to the scheduled users which send IN requests to interfering macro-BS $\ell$ as the \emph{potential IN users} of  interfering macro-BS $\ell$ (in this time slot). We introduce  another design parameter $U\in\{0,1,\cdots, N_1-1\}$ of this IN scheme, referred to as the {\em maximum IN DoF}.
Consider a particular time slot. Let $K_{\ell}$ denote the number of the potential IN users of  interfering macro-BS $\ell$.  Note that $T_1=T_2=1$ implies $K_{\ell}=0$.  Consider two cases in the following. i) If $K_{\ell}>0$ and $U>0$, macro-BS $\ell$ makes use of at most $U$   DoF to perform IN to some of its potential IN users.  In particular, if $0<K_{\ell}\le U$, macro-BS $\ell$ can perform IN to all of its $K_{\ell}$ potential IN users using $K_{\ell}$ DoF; if $K_{\ell}>U$, macro-BS $\ell$ randomly selects $U$ out of its $K_{\ell}$ potential IN users according to the uniform distribution, and perform IN to the selected $U$ users using $U$ DoF. Hence, in this case, macro-BS $\ell$ performs IN to $u_{{\rm IN},\ell}\define\min\left(U,K_{\ell}\right)$ potential IN users (referred to as the \emph{IN users} of macro-BS $\ell$) using $u_{{\rm IN},\ell}$ DoF (referred to as the \emph{IN DoF} of macro-BS $\ell$). ii) If $K_{\ell}=0$ or $U=0$, macro-BS $\ell$ does not perform IN. In this case, we let $u_{{\rm IN},\ell}=0$. In both cases, $N_{1}-u_{{\rm IN},\ell}$ DoF at macro-BS $\ell$ is used for boosting the desired signal to its scheduled user.

Now, we introduce the precoding vectors at macro-BSs in the IN scheme. Consider two cases in the following. i) If $K_{\ell}>0$ and $U>0$, macro-BS $\ell$ utilizes the low-complexity ZFBF precoder to serve its scheduled user and simultaneously perform IN to its $u_{{\rm IN},\ell}$ IN users. Specifically,  denote $\mathbf{H}_{1,\ell}=\left[\mathbf{h}_{1,\ell}\; \mathbf{g}_{1,\ell1}\;\ldots\;\mathbf{g}_{1,\ell u_{{\rm IN},\ell}}\right]^{\dagger}$, where $\mathbf{h}_{1,\ell}\dis\mathcal{CN}_{N_{1},1}\left(\mathbf{0}_{N_{1}\times 1},\mathbf{I}_{N_{1}}\right)$ denotes the channel vector between macro-BS $\ell$ and its scheduled user, and $\mathbf{g}_{1,\ell i}\dis \mathcal{CN}_{N_{1},1}\left(\mathbf{0}_{N_{1}\times 1},\mathbf{I}_{N_{1}}\right)$ denotes the channel vector between macro-BS $\ell$ and its $i$th IN user $(i=1,\ldots,u_{{\rm IN},\ell})$. The ZFBF precoding matrix at macro-BS $\ell$ is designed to be  the pseudo-inverse of $\mathbf{H}_{1,\ell}$, i.e.,  $\mathbf{W}_{1,\ell}=\mathbf{H}_{1,\ell}^{\dagger}\left(\mathbf{H}_{1,\ell}\mathbf{H}_{1,\ell}^{\dagger}\right)^{-1}$ and the ZFBF vector at macro-BS $\ell$ is designed to be  $\mathbf{f}_{1,\ell}=\frac{\mathbf{w}_{1,\ell}}{\|\mathbf{w}_{1,\ell}\|}$, where $\mathbf{w}_{1,\ell}$ is the first column of $\mathbf{W}_{1,\ell}$ \cite{WieselTSP08}. ii) If $K_{\ell}=0$ or $U=0$, macro-BS $\ell$ uses the maximal ratio transmission (MRT) precoder to serve its scheduled user, which is a special case of the ZFBF precoder introduced for $K_{\ell}>0$ and $U>0$, and can be readily obtained from it by letting $u_{{\rm IN},\ell}=0$, i.e., $\mathbf{H}_{1,\ell}=\mathbf{h}_{1,\ell}^{\dagger}$. Next, we introduce the precoding vectors at pico-BSs. Each pico-BS utilizes the MRT precoder to serve its scheduled user. Specifically, the beamforming vector at pico-BS $\ell$ is $\mathbf{f}_{2,\ell}=\frac{\mathbf{h}_{2,\ell}}{\left\|\mathbf{h}_{2,\ell}\right\|}$,
where $\mathbf{h}_{2,\ell}\dis\mathcal{CN}_{N_{2},1}\left(\mathbf{0}_{N_{2}\times 1},\mathbf{I}_{N_{2}}\right)$ denotes the channel vector between pico-BS $\ell$ and its scheduled user. Note that the simple beamforming scheme (without interference management) can be included in the IN scheme as a special case by letting $T_{1}=T_{2}=1$ and/or $U=0$.  Note that all the analytical results in this paper hold for $T_{1}=T_{2}=1$ and/or $U=0$.

 Let $\mathbf{h}_{j,00}$ denote the channel vector between $u_{0}\in\mathcal{U}_{j}$ and its serving BS $B_{j,0}$, $D_{j,\ell0}$   denote the distance between $u_{0}$ and BS $\ell$ in the $j$th tier, $Y_{j}$ denote the distance between $u_{0}$ and $B_{j,0}$, and $\mathbf{f}_{j,0}$ denote the beamforming vector at $B_{j,0}$, with $\left|\mathbf{h}_{j,00}^{\dagger}\mathbf{f}_{j,0}\right|^{2}\dis{\rm Gamma}\left(M_{j},1\right)$  (i.e., $\chi_{2M_j}^2$), $M_{1}=N_{1}-u_{{\rm IN},0}$ and $M_{2}=N_{2}$ \cite[Lemma 1]{JindalTCOM11}.   Here, the notation $X \dis Y$ means that $X$ \emph{is distributed as} $Y$.   Let $\mathbf{h}_{j,\ell0}$ denote the channel vector between $u_0$ and BS $\ell$ in the $j$th tier, and  $\mathbf{f}_{j,\ell}$ denote the beamforming vector at BS $\ell$ in the $j$th tier, with $\left|\mathbf{h}_{j,\ell0}^{\dagger}\mathbf{f}_{j,\ell}\right|^{2}\dis{\rm Gamma}(1,1)$ (i.e., $\chi_2^2$)   \cite[Lemma 1]{JindalTCOM11}.  Let $x_{j,\ell}$ denote the symbol sent from BS $\ell$ in the $j$th tier to its scheduled user satisfying ${\rm E}\left[x_{j,\ell}x_{j,\ell}^{*}\right]=P_{j}$.
 Let $\Phi_{j,1C}$  denote the potential IN macro-BSs  of $u_0\in \mathcal U_j$ which do not select it for IN. Let
$\Phi_{j,1O}$ denote the  interfering macro-BSs of $u_0\in \mathcal U_j$ which are not its potential IN macro-BSs.  Let $\Phi_{j,2}$ denote the interfering pico-BSs of $u_0\in \mathcal U_j$.  
 As in \cite{singh13,andrews11}, we assume that all macro-BSs and pico-BSs are   active.
 We now discuss the received signal of  $u_0$.
  \begin{enumerate}
  \item Macro-User: The received signal of the typical user $u_{0}\in \mathcal{U}_{1}$ is
\small{\begin{align}
&y_{1,0}=Y_{1}^{-\frac{\alpha_{1}}{2}}\mathbf{h}_{1,00}^{\dagger}\mathbf{f}_{1,0}x_{1,0}+\sum_{\ell\in\Phi_{1,1C}}D_{1,\ell0}^{-\frac{\alpha_{1}}{2}}\mathbf{h}_{1,\ell0}^{\dagger}\mathbf{f}_{1,\ell}x_{1,\ell}+\sum_{\ell\in\Phi_{1,1O}}D_{1,\ell0}^{-\frac{\alpha_{1}}{2}}\mathbf{h}_{1,\ell0}^{\dagger}\mathbf{f}_{1,\ell}x_{1,\ell}+\sum_{\ell\in\Phi_{1,2}}D_{2,\ell0}^{-\frac{\alpha_{2}}{2}}\mathbf{h}^{\dagger}_{2,\ell0}\mathbf{f}_{2,\ell}x_{2,\ell}.\label{eq:y0_macro}
\end{align}}\normalsize
Note that $\Phi_{1,1C}\cup\Phi_{1,1O}\cup\{B_{1,0}\}\subseteq\Phi_1$ and $\Phi_{1,2}=\Phi_2$.
\item Pico-User:
The received signal of the typical user $u_{0}\in\mathcal{U}_{2}$ is
\small{\begin{align}
&y_{2,0}=Y_{2}^{-\frac{\alpha_{2}}{2}}\mathbf{h}_{2,00}^{\dagger}\mathbf{f}_{2,0}x_{2,0}+\sum_{\ell\in\Phi_{2,1C}}D_{1,\ell0}^{-\frac{\alpha_{1}}{2}}\mathbf{h}_{1,\ell0}^{\dagger}\mathbf{f}_{1,\ell}x_{1,\ell}+\sum_{\ell\in\Phi_{2,1O}}D_{1,\ell0}^{-\frac{\alpha_{1}}{2}}\mathbf{h}_{1,\ell0}^{\dagger}\mathbf{f}_{1,\ell}x_{1,\ell}+\sum_{\ell\in\Phi_{2,2}}D_{2,\ell0}^{-\frac{\alpha_{2}}{2}}\mathbf{h}_{2,\ell0}^{\dagger}\mathbf{f}_{2,\ell}x_{2,\ell}.\label{eq:y0_pico}
\end{align}}\normalsize
Note that $\Phi_{2,1C}\cup\Phi_{2,1O}\subseteq\Phi_1$ and $\Phi_{2,2}\cup\{B_{2,0}\}=\Phi_{2}$.
  \end{enumerate}
\color{black}

We now obtain the SIR of the typical user. Under the above IN scheme, $u_{0}\in \mathcal U_j$ experiences three types of interference: 1) residual aggregated interference $I_{j,1C}$ from its potential IN macro-BSs $\Phi_{j,1C}$ which do not select $u_{0}$ for IN, 2) aggregated interference $I_{j,1O}$ from interfering macro-BSs $\Phi_{j,1O}$ which are not its potential IN macro-BSs, and 3) aggregated interference $I_{j,2}$ from all interfering pico-BSs $\Phi_{j,2}$.  Specifically, the SIR of the typical user $u_{0}\in\mathcal{U}_{j}$  is given by
\small{\begin{align}\label{eq:SIRj0}
{\rm SIR}_{j,0}&=\frac{P_{j}Y_{j}^{-\alpha_{j}}\left|\mathbf{h}_{j,00}^{\dagger}\mathbf{f}_{j,0}\right|^{2}}{P_{1}I_{j,1C}+P_{1}I_{j,1O}+P_{2}I_{j,2}}
\end{align}}\normalsize
where
\small{$$I_{j,1C}=\sum_{\ell\in\Phi_{j,1C}}D_{1,\ell0}^{-\alpha_{1}}\left|\mathbf{h}_{1,\ell0}^{\dagger}\mathbf{f}_{1,\ell}\right|^{2}, I_{j,1O}=\sum_{\ell\in\Phi_{j,1O}}D_{1,\ell0}^{-\alpha_{1}}\left|\mathbf{h}_{1,\ell0}^{\dagger}\mathbf{f}_{1,\ell}\right|^{2},I_{j,2}=\sum_{\ell\in\Phi_{j,2}}D_{2,\ell0}^{-\alpha_{2}}\left|\mathbf{h}_{2,\ell0}^{\dagger}\mathbf{f}_{2,\ell}\right|^{2}.$$}\normalsize

\subsection{Preliminary Results}\label{subsec:prelim}
In this part, we evaluate some distributions related to the IN scheme, which will be used to calculate the coverage probability in (\ref{eq:CP_def}). These distributions are based on approximations, the accuracy of which will be verified  in Section \ref{sec:CP_general}. We first calculate the probability mass function (p.m.f.) of the number of the potential IN users of 
 an arbitrary (chosen uniformly at random) macro-BS, denoted as $K_{0}$. The p.m.f. of $K_{0}$ depends on the point processes formed by the scheduled macro and pico users, which are related to but not PPPs \cite{bai13}. For analytical tractability, we approximate the scheduled macro and pico users as two independent PPPs with densities $\lambda_{1}$ and $\lambda_{2}$, respectively.  Note that approximating the scheduled users as a homogeneous PPP has been considered in existing papers (see e.g., \cite{bai13}). Then, we have the p.m.f. of $K_{0}$ as follows.
\begin{lemma}\label{lem:num_req_pmf}
The p.m.f. of $K_{0}$ is given by
\small{\begin{align}\label{eq:K_pmf}
{\rm Pr}\left(K_{0}=k\right)\approx\frac{\bar{L}(T_1,T_2)^{k}}{k!}\exp\left(-\bar{L}(T_1,T_2)\right), \quad k=0,1,\cdots
\end{align}}\normalsize
where $\bar{L}(T_1,T_2)=\bar{L}_{1}(T_1)+\bar{L}_{2}(T_2)$ with
\small{\begin{align}
\bar{L}_{j}(T_j)=&2\pi\lambda_{j}\int_{0}^{\infty}r\int_{\left(\frac{P_{j}}{P_{1}T_{j}}\right)^{\frac{1}{\alpha_{j}}}r^{\frac{\alpha_{1}}{\alpha_{j}}}}^{\left(\frac{P_{j}}{P_{1}}\right)^{\frac{1}{\alpha_{j}}}r^{\frac{\alpha_{1}}{\alpha_{j}}}}f_{Y_{j}}(y){\rm d}y{\rm d}r\;,\;j=1,2.\label{eqn:def-Lj}
\end{align}}\normalsize
Here,   the p.d.f.s of $Y_{j}$ (the  distance between $u_{0}$ and its serving BS $B_{j,0}$) $f_{Y_{j}}(y)$ ($j=1,2$) are given as follows \cite[Lemma 4]{SinghTWC13}:
\small{\begin{align}
f_{Y_{1}}(y)&=\frac{2\pi\lambda_{1}}{\mathcal{A}_{1}}y\exp\left(-\pi\left(\lambda_{1}y^{2}+\lambda_{2}\left(\frac{P_{2}}{P_{1}}\right)^{\frac{2}{\alpha_{2}}}y^{\frac{2\alpha_{1}}{\alpha_{2}}}\right)\right)\;\label{eq:pdfY1}\\
f_{Y_{2}}(y)&=\frac{2\pi\lambda_{2}}{\mathcal{A}_{2}}y\exp\left(-\pi\left(\lambda_{1}\left(\frac{P_{1}}{P_{2}}\right)^{\frac{2}{\alpha_{1}}}y^{\frac{2\alpha_{2}}{\alpha_{1}}}+\lambda_{2}y^{2}\right)\right)\label{eq:pdfY2}\;
\end{align}}\normalsize
where $\mathcal{A}_{j}\define {\rm Pr}\left(u_{0}\in\mathcal{U}_{j}\right)$ ($j=1,2$) are given by
\small{\begin{align}
&\mathcal{A}_{1}=2\pi\lambda_{1}\int_{0}^{\infty}z\exp\left(-\pi\lambda_{1}z^{2}\right)\exp\left(-\pi\lambda_{2}\left(\frac{P_{2}}{P_{1}}\right)^{\frac{2}{\alpha_{2}}}z^{\frac{2\alpha_{1}}{\alpha_{2}}}\right){\rm d}z\label{eq:A1}\\
&\mathcal{A}_{2}=2\pi\lambda_{2}\int_{0}^{\infty}z\exp\left(-\pi\lambda_{2}z^{2}\right)\exp\left(-\pi\lambda_{1}\left(\frac{P_{1}}{P_{2}}\right)^{\frac{2}{\alpha_{1}}}z^{\frac{2\alpha_{2}}{\alpha_{1}}}\right){\rm d}z.\label{eq:A2}
\end{align}}\normalsize
\end{lemma}
\begin{proof}
See Appendix \ref{proof:pmf_Lbar}.
\end{proof}

Note that $\bar{L}(T_1,T_2)$ represents the average number of IN requests of the scheduled users received by an arbitrary macro-BS, and $\bar{L}_{j}(T_j)$  represents the average number of IN requests of the scheduled users in the $j$th tier received by  an arbitrary macro-BS.  From \eqref{eqn:def-Lj}, we can easily see that $\bar{L}_{j}(T_j)$ and $\bar{L}(T_1,T_2)$ increase with $T_1$ and $T_2$. From \eqref{eq:K_pmf}, we know that $K_0$ approximately follows the Poisson distribution with mean  $\bar{L}(T_1,T_2)$.  


Next, we calculate the p.m.f. of the number of the  IN users  of
 an arbitrary (chosen uniformly at random) macro-BS 
$u_{{\rm IN},0}=\min\left(U,K_{0}\right)$ based on \emph{Lemma \ref{lem:num_req_pmf}}.
\begin{lemma}\label{lem:pmf_uIN0}
The p.m.f. of $u_{{\rm IN},0}$ is given by
\small{\begin{align}
{\rm Pr}\left(u_{{\rm IN},0}=u\right)=
\begin{cases}
&{\rm Pr}\left(K_{0}=u\right)\;,\hspace{12mm} {\rm for}\;0\le u<U \\
&\sum_{k=U}^{\infty}{\rm Pr}\left(K_{0}=k\right)\;,\hspace{2mm}{\rm for}\;u=U
\end{cases}
\;.\notag
\end{align}}\normalsize
\end{lemma}

Now, we calculate  the probability that an arbitrary  (chosen uniformly at random) potential IN macro-BS of $u_{0}$ selects $u_{0}$ for IN, referred to as the {\em IN probability} and denoted as $p_{c}\left(U,T_{1},T_{2}\right)$, based on \emph{Lemma \ref{lem:num_req_pmf}}.

\begin{lemma}\label{lem:INprob}
The IN probability is given by
\small{\begin{align}
p_{c}\left(U,T_{1},T_{2}\right)
\approx&\exp\left(-\bar{L}(T_1,T_2)\right)\left(\sum_{k=0}^{U-1}\frac{\bar{L}(T_1,T_2)^{k}}{k!}+U\sum_{k=U}^{\infty}\frac{\bar{L}(T_1,T_2)^{k}}{(k+1)!}\right)\;.\notag
\end{align}}\normalsize
\end{lemma}
\begin{proof}
See Appendix \ref{proof:lem_INprob}. 
\end{proof}

Note that different potential IN macro-BSs of $u_{0}$ selects $u_{0}$ for IN dependently (as the numbers of the potential IN users of these macro-BSs are dependent). For analytical tractability, we assume that different potential IN macro-BSs of $u_{0}$ select $u_{0}$ for IN independently.
Using \emph{independent thinning}, $u_0$'s potential IN  macro-BSs which do not select $u_0$ for IN   can be \emph{approximated} by a homogeneous PPP with density $p_{\bar{c}}\left(U,T_{1},T_{2}\right)\lambda_{1}$, where $p_{\bar{c}}\left(U,T_{1},T_{2}\right)\define1-p_{c}\left(U,T_{1},T_{2}\right)$.

\section{Coverage Probability--General SIR Threshold Regime}\label{sec:CP_general}
In this section, we investigate the coverage probability in the general SIR threshold regime.
By total probability theorem and the preliminary results obtained in Section \ref{subsec:prelim} (under some approximations), we have  the following theorem.

\begin{figure*}[!t]
\small{\begin{align}
&\mathcal{S}_{1}(\beta,U,T_{1},T_{2})=\sum_{u=0}^{U}{\rm Pr}\left(u_{{\rm IN},0}=u\right)\int_{0}^{\infty}\sum_{n=0}^{N_{1}-u-1}\frac{1}{n!}\sum_{\left(n_{a}\right)_{a=1}^{3}\in\mathcal{N}_{n}}\binom{n}{n_{1},n_{2},n_{3}}\mathcal{\tilde{L}}^{(n_{1})}_{I_{1,1C}}\left(U, \beta y^{\alpha_{1}},y,T^{\frac{1}{\alpha_{1}}}y\right)\notag\\
&\hspace{6cm}\times\mathcal{\tilde{L}}^{(n_{2})}_{I_{1,1O}}\left(\beta y^{\alpha_{1}},T_{1}^{\frac{1}{\alpha_{1}}}y\right)\mathcal{\tilde{L}}^{(n_{3})}_{I_{1,2}}\left(\beta\frac{P_{2}}{P_{1}}y^{\alpha_{1}},\left(\frac{P_{2}}{P_{1}}\right)^{\frac{1}{\alpha_{2}}}y^{\frac{\alpha_{1}}{\alpha_{2}}}\right)f_{Y_{1}}(y){\rm d}y\label{eq:CP1}\\
&\mathcal{S}_{2}(\beta,U,T_{1},T_{2})=\int_{0}^{\infty}\sum_{n=0}^{N_{2}-1}\frac{1}{n!}\sum_{(n_{a})_{a=1}^{3}\in\mathcal{N}_{n}}\binom{n}{n_{1},n_{2},n_{3}}\mathcal{\tilde{L}}^{(n_{1})}_{I_{2,1C}}\left(U,\beta\frac{P_{1}}{P_{2}}y^{\alpha_{2}},\left(\frac{P_{1}}{P_{2}}\right)^{\frac{1}{\alpha_{1}}}y^{\frac{\alpha_{2}}{\alpha_{1}}},\left(\frac{P_{1}T_{2}}{P_{2}}\right)^{\frac{1}{\alpha_{1}}}y^{\frac{\alpha_{2}}{\alpha_{1}}}\right)\notag\\
&\hspace{6cm}\times\mathcal{\tilde{L}}^{(n_{2})}_{I_{2,1O}}\left(\beta\frac{P_{1}}{P_{2}}y^{\alpha_{2}},\left(\frac{P_{1}T_{2}}{P_{2}}\right)^{\frac{1}{\alpha_{1}}}y^{\frac{\alpha_{2}}{\alpha_{1}}}\right)\mathcal{\tilde{L}}^{(n_{3})}_{I_{2,2}}\left(\beta y^{\alpha_{2}},y\right)f_{Y_{2}}(y){\rm d}y\label{eq:CP2}\\
&\tilde{\mathcal{L}}_{I_{j,1C}}^{(n)}(U,s,r_{j,1C},r_{j,1O})=\exp\Bigg(-\left(B^{'}\left(\frac{2}{\alpha_{1}},1-\frac{2}{\alpha_{1}},\frac{1}{1+sr_{j,1C}^{-\alpha_{1}}}\right)-B^{'}\left(\frac{2}{\alpha_{1}},1-\frac{2}{\alpha_{1}},\frac{1}{1+sr_{j,1O}^{-\alpha_{1}}}\right)\right)\notag\\
&\hspace{3cm}\times\frac{2\pi}{\alpha_{1}}p_{\bar{c}}\left(U,T_{1},T_{2}\right)\lambda_{1}s^{\frac{2}{\alpha_{1}}}\Bigg)\sum_{(m_{a})_{a=1}^{n}\in\mathcal{M}_{n}}\frac{n!}{\prod_{a=1}^{n}m_{a}!}\prod_{a=1}^{n}\Bigg(\frac{2\pi}{\alpha_1}p_{\bar{c}}\left(U,T_{1},T_{2}\right)\lambda_{1}s^{\frac{2}{\alpha_{1}}}\notag\\
&\hspace{3cm}\times\left(B^{'}\left(1+\frac{2}{\alpha_{1}},a-\frac{2}{\alpha_{1}},\frac{1}{1+sr_{j,1C}^{-\alpha_{1}}}\right)-B^{'}\left(1+\frac{2}{\alpha_{1}},a-\frac{2}{\alpha_{1}},\frac{1}{1+sr_{j,1O}^{-\alpha_{1}}}\right)\right)\Bigg)^{m_{a}}\label{eq:LTdiff_I1in}\\
&\mathcal{\tilde{L}}^{(n)}_{I_{j,k}}(s,r_{j,k})=\exp\left(-\frac{2\pi\lambda_{J(k)}}{\alpha_{J(k)}}s^{\frac{2}{\alpha_{J(k)}}}B^{'}\left(\frac{2}{\alpha_{J(k)}},1-\frac{2}{\alpha_{J(k)}},\frac{1}{1+\frac{s}{r_{j,k}^{\alpha_{J(k)}}}}\right)\right)\sum_{(m_{a})_{a=1}^{n}\in\mathcal{M}_{n}}\frac{n!}{\prod_{a=1}^{n}m_{a}!}\notag\\
&\hspace{22mm}\times \prod_{a=1}^{n}\Bigg(\frac{2\pi\lambda_{J(k)}}{\alpha_{J(k)}}s^{\frac{2}{\alpha_{J(k)}}}B^{'}\bigg(1+\frac{2}{\alpha_{J(k)}},a-\frac{2}{\alpha_{J(k)}},\frac{1}{1+\frac{s}{r_{j,k}^{\alpha_{J(k)}}}}\bigg) \Bigg)^{m_{a}}\label{eq:LTdiff_1O2}
\end{align}}
\normalsize \hrulefill
\end{figure*}

\begin{theorem}[Coverage Probability]\label{thm:overall_CP}
Under design parameters $U$, $T_{1}$ and $T_{2}$, we have:
1) coverage probability of a macro-user: $\mathcal{S}_{1}\left(\beta,U,T_{1},T_{2}\right)\define{\rm Pr}\left({\rm SIR}_{0}>\beta|u_{0}\in\mathcal{U}_{1}\right)$, given in (\ref{eq:CP1});
2) coverage probability of a pico-user: $\mathcal{S}_{2}\left(\beta,U,T_{1},T_{2}\right)\define{\rm Pr}\left({\rm SIR}_{0}>\beta|u_{0}\in\mathcal{U}_{2}\right)$, given in (\ref{eq:CP2});
3) overall coverage probability $\mathcal{S}\left(\beta,U,T_{1},T_{2}\right)=\mathcal{A}_{1}\mathcal{S}_{1}(\beta,U,T_{1},T_{2})+\mathcal{A}_{2}\mathcal{S}_{2}(\beta,U,T_{1},T_{2})$, where $\mathcal{A}_{j}\define {\rm Pr}\left(u_{0}\in\mathcal{U}_{j}\right)$ ($j=1,2$) are given in (\ref{eq:A1}) and (\ref{eq:A2}).
Here, $\mathcal{\tilde{L}}^{(n)}_{I_{j,1C}}\left(U,s,r_{j,1C},r_{j,1O}\right)$ and $\mathcal{\tilde{L}}^{(n)}_{I_{j,k}}(s,r_{j,k})$ ($k\in\{1O,2\}$) are given in (\ref{eq:LTdiff_I1in}) and (\ref{eq:LTdiff_1O2}) (with $J(1O)=1$ and $J(2)=2$), respectively.
Moreover, $B^{'}(a,b,z)\define\int_{z}^{1}u^{a-1}(1-u)^{b-1}{\rm d}u$ ($0<z<1$)  denotes the complementary incomplete beta function, $\mathcal{N}_{n}\define\{(n_{a})_{a=1}^{3}|n_{a}\in\mathbb{N}^{0},\sum_{a=1}^{3}n_{a}=n\}$, and $\mathcal{M}_{n}\define\{(m_{a})_{a=1}^{n}|m_{a}\in\mathbb{N}^{0},\sum_{a=1}^{3}a\cdot m_{a}=n\}$, where $\mathbb{N}^{0}$ denotes the set of nonnegative integers.
\end{theorem}
\begin{proof}
See Appendix \ref{proof:thm1}.
\end{proof}

Fig.~\ref{fig:CP_OP} plots the coverage probability versus the IN DoF $U$ and the SIR threshold $\beta$. We see from Fig.~\ref{fig:CP_OP} that the ``Analytical" curves, which are plotted using $\mathcal{S}\left(\beta,U,T_{1},T_{2}\right)$ in \emph{Theorem \ref{thm:overall_CP}}, are reasonably close to the ``Monte Carlo" curves, although \emph{Theorem \ref{thm:overall_CP}} is obtained based on some approximations (cf. Section \ref{subsec:prelim}). 
Please note that  the approximation error shown in Fig.~\ref{fig:CP_OP}  is less than 0.022. 
Later, we shall consider the optimization of $U$ for given $T_1$ and $T_2$.\footnote{The coverage probability can  be further improved by jointly adjusting $T_{1}$ and $T_{2}$. We shall consider the optimization of $T_1$ and $T_2$ in the future work.}

\begin{figure}[t]
\centering
\subfigure[Maximum IN DoF at $\beta=10$ dB]{
\includegraphics[height=0.28\columnwidth,width=0.38\columnwidth]
{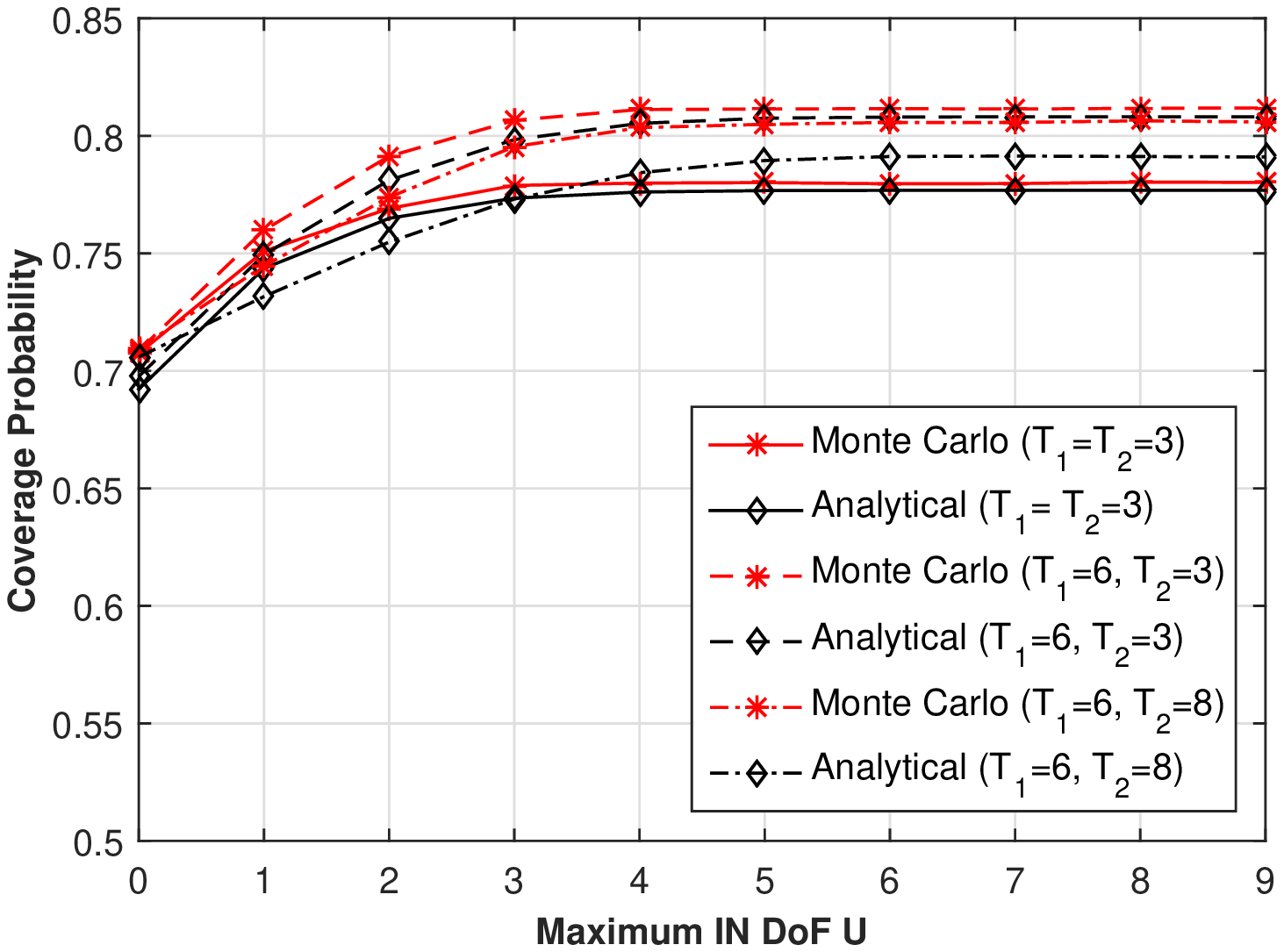}
\label{fig:U_T}
}
\subfigure[SIR threshold at $U=5$]{
\includegraphics[height=0.28\columnwidth,width=0.38\columnwidth]
{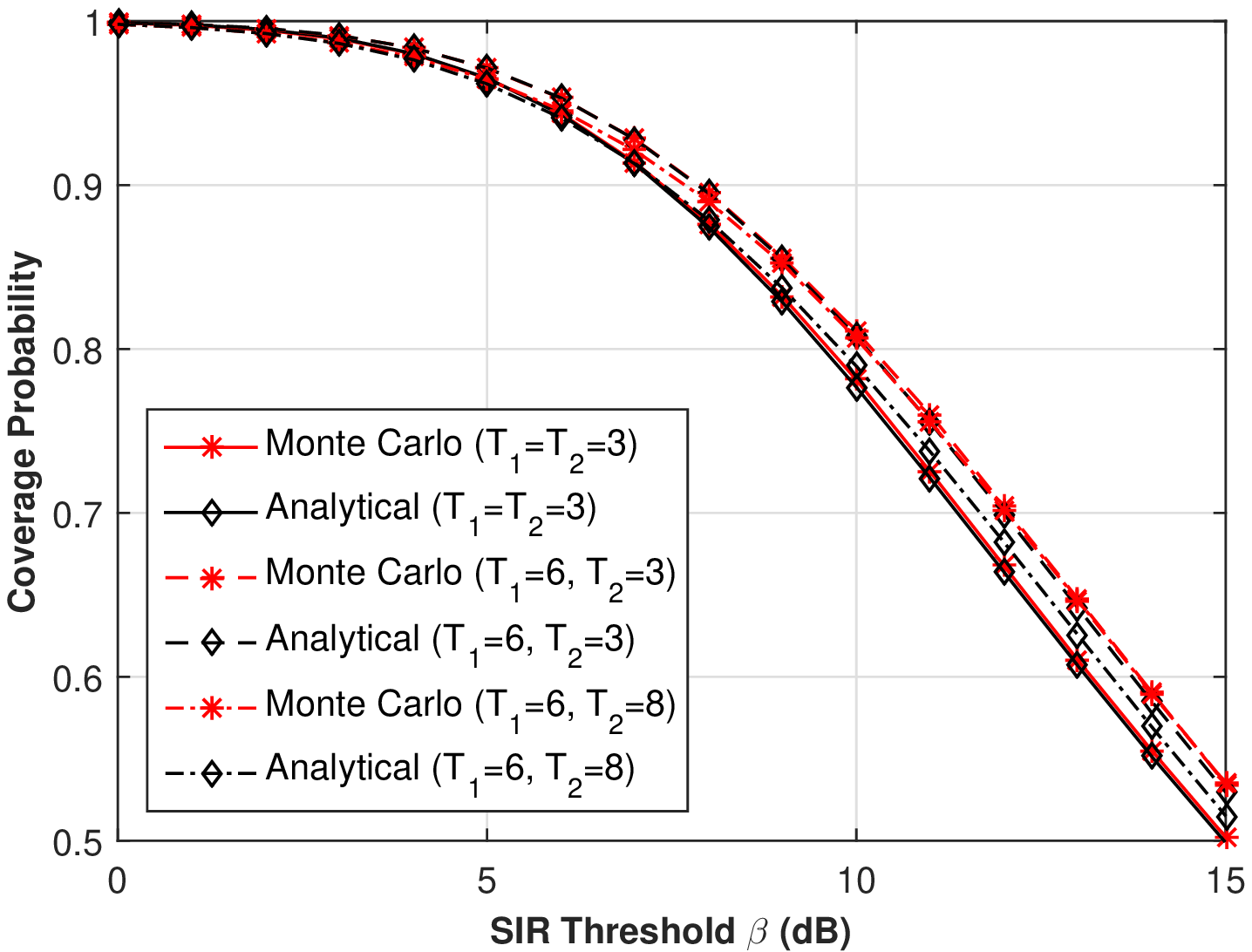}
\label{fig:verification_vs_beta}
}
\caption{\small{Coverage probability versus maximum IN DoF  and SIR threshold. $N_{1}=10$, $N_{2}=8$, $\alpha_{1}=4.5$, $\alpha_{2}=4.7$, $\frac{P_{1}}{P_{2}}=15$ dB, $\lambda_{1}=0.0005$ nodes/m$^{2}$, and $\lambda_{2}=0.001$ nodes/m$^{2}$.  For the Monte Carlo results, we use a two dimensional square of area $240^2$ m$^2$ to simulate the large-scale HetNet  and $10^6$   random realizations to obtain the average coverage probability. The computation time using Monte-Carlo simulations is about  $250$  times of that using the analytical results, demonstrating that the analytical results are more tractable than Monte-Carlo simulations.  
}}\label{fig:CP_OP}
\vspace{-5mm}
\end{figure}


\section{Asymptotic Outage Probability Analysis--Low SIR Threshold Regime}\label{sec:CP_low}
In this section, we   analyze and optimize  the complement of the coverage probability, i.e., the outage probability of the IN scheme in the low   SIR threshold regime, i.e., $\beta\to0$. The asymptotic analysis and optimization   offer important design insights for practical HetNets.


\subsection{Asymptotic Outage Probability Analysis}
In this part, we analyze the asymptotic outage probability ${\rm Pr}\left({\rm SIR}_{0}<\beta\right)$ of the IN scheme when $\beta\to0$. First,  as in \cite{zhang14}, we define the \emph{order gain} of the outage probability (in interference-limited systems), i.e., the exponent of the outage
probability as the SIR threshold decreases to 0:\footnote{Note that this definition is analogous to the standard diversity order gain of error
probability in noise-limited systems, i.e., the exponent of error probability as the
(mean) signal-to-noise ratio (SNR)  increases to infinity\cite{TIT03ZhengTse}\cite{zhang14}.}
\small{\begin{align}
d\define \lim_{\beta\to0}\frac{\log {\rm Pr}\left({\rm SIR}_{0}<\beta\right)}{\log\beta}.
\end{align}}\normalsize
Then, we define the \emph{coefficient} of the asymptotic outage probability: $\lim_{\beta\to0}\frac{{\rm Pr}\left({\rm SIR_{0}}<\beta\right)}{\beta^{d}}$.  Leveraging the order gain and the coefficient of the outage probability, we shall characterize the key behavior of the complex outage probability in the low SIR threshold regime.

Recently, a tractable approach has been proposed in \cite{haenggi14} to characterize the order gain for a class of communication schemes in wireless networks which satisfy certain conditions. However, this approach does not provide tractable analytical expressions for the  coefficient of the asymptotic outage probability for most of the schemes using multiple antennas in this class. By utilizing series expansion of some special functions and dominated convergence theorem, we characterize both the order gain and the coefficient of the asymptotic outage probability of the IN scheme in multi-antenna HetNets, which are presented as follows.
\begin{theorem}[Asymptotic Outage Probability]\label{them:CPj_lowbeta}
Under design parameters $U$, $T_{1}$ and $T_{2}$, when $\beta\to0$, we have:\footnote{$f(\beta)\sims g(\beta)$ means $\lim_{\beta\to0}\frac{f(\beta)}{g(\beta)}=1$.}
1) outage probability of a macro-user: $1-\mathcal{S}_{1}(\beta,U,T_{1},T_{2})\sims b_{1}\left(U,T_{1},T_{2}\right)\beta^{N_{1}-U}$;
2) outage probability of a pico-user: $1-\mathcal{S}_{2}(\beta,U,T_{1},T_{2})\sims b_{2}\left(U,T_{1},T_{2}\right)\beta^{N_{2}}$;
3) overall outage probability: $1-\mathcal{S}\left(\beta,U,T_{1},T_{2}\right)\sims b\left(U,T_{1},T_{2}\right)\beta^{\min\{N_{1}-U,N_{2}\}}$, where
\small{\begin{align}
b\left(U,T_{1},T_{2}\right)=
\begin{cases}
\mathcal A_2b_{2}\left(U,T_{1},T_{2}\right),  & U<N_{1}-N_{2}\\
\mathcal A_1 b_{1}\left(U,T_{1},T_{2}\right)+ \mathcal A_2 b_{2}\left(U,T_{1},T_{2}\right), &U=N_{1}-N_{2}\\
\mathcal A_1 b_{1}\left(U,T_{1},T_{2}\right),  & U>N_{1}-N_{2}
\end{cases}
.\notag
\end{align}}\normalsize
Here, $b_{j}\left(U,T_{1},T_{2}\right)$ is given in (\ref{eq:coeff_smallbeta}) with $U_{j}=U$ and $\mathcal{P}_{j}={\rm Pr}\left(u_{{\rm IN},0}=U\right)$ if $j=1$ and $T_{1},T_{2}>1$; $U_{j}=0$ and $\mathcal{P}_{j}=1$, otherwise. Moreover, $b_{2}\left(U,T_{1},T_{2}\right)$ decreases with $U$.
\begin{figure*}[!t]
\small{\begin{align}\label{eq:coeff_smallbeta}
&b_{j}\left(U,T_{1},T_{2}\right)\notag\\
=&\sum_{(n_{a})_{a=1}^{3}\in\mathcal{N}_{N_{j}-U_{j}}}\sum_{(m_{a})_{a=1}^{n_{1}}\in\mathcal{M}_{n_{1}}}\sum_{(p_{a})_{a=1}^{n_{2}}\in\mathcal{M}_{n_{2}}}\sum_{(q_{a})_{a=1}^{n_{3}}\in\mathcal{M}_{n_{3}}}\left(\int_{0}^{\infty}y^{\frac{2\alpha_{j}}{\alpha_{1}}\left(\sum_{a=1}^{n_{1}}m_{a}+\sum_{a=1}^{n_{2}}p_{a}\right)+\frac{2\alpha_{j}}{\alpha_{2}}\sum_{a=1}^{n_{3}}q_{a}}f_{Y_{j}}(y){\rm d}y\right)\notag\\
&\hspace{-2mm}\times \left(\prod_{a=1}^{n_{2}}\left(\left(\frac{1}{T_{j}}\right)^{a-\frac{2}{\alpha_{1}}}\right)^{p_{a}}\right)\left(\prod_{a=1}^{n_{1}}\frac{1}{m_{a}!}\left(\frac{\frac{2}{\alpha_{1}}\pi \lambda_{1}}{a-\frac{2}{\alpha_{1}}}\left(\frac{P_{1}}{P_{j}}\right)^{\frac{2}{\alpha_{1}}}\right)^{m_{a}}\right)\left(\prod_{a=1}^{n_{2}}\frac{1}{p_{a}!}\left(\frac{\frac{2}{\alpha_{1}}\pi\lambda_{1}}{a-\frac{2}{\alpha_{1}}}\left(\frac{P_{1}}{P_{j}}\right)^{\frac{2}{\alpha_{1}}}\right)^{p_{a}}\right)\notag\\
&\hspace{-2mm}\times \left(\prod_{a=1}^{n_{3}}\frac{1}{q_{a}!}\left(\frac{\frac{2}{\alpha_{2}}\pi\lambda_{2}}{a-\frac{2}{\alpha_{2}}}\left(\frac{P_{2}}{P_{j}}\right)^{\frac{2}{\alpha_{2}}}\right)^{q_{a}}\right)\left(\prod_{a=1}^{n_{1}}\left(p_{\bar{c}}\left(U,T_{1},T_{2}\right)\left(1-\left(\frac{1}{T_{j}}\right)^{a-\frac{2}{\alpha_{1}}}\right)\right)^{m_{a}}\right)\mathcal{P}_{j}
\end{align}}\normalsize
\normalsize \hrulefill
\end{figure*}
\end{theorem}
\begin{proof}
See Appendix \ref{proof:thm2}.
\end{proof}

\begin{figure}[t]
\centering
\subfigure[$U=6$]{
\includegraphics[height=0.28\columnwidth,width=0.38\columnwidth]
{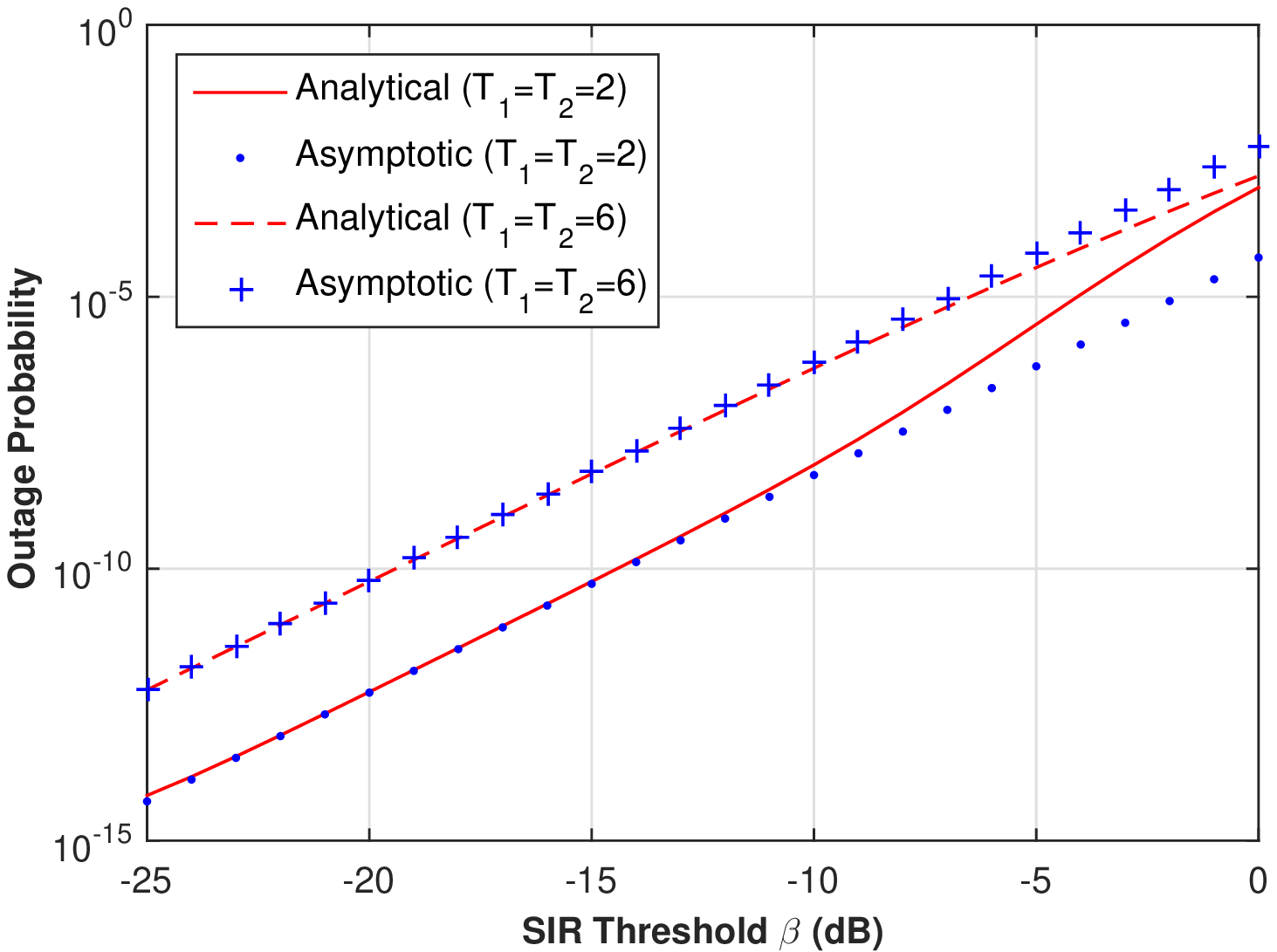}
\label{fig:lowSIR U6}
}
\subfigure[$U=7$]{
\includegraphics[height=0.28\columnwidth,width=0.38\columnwidth]
{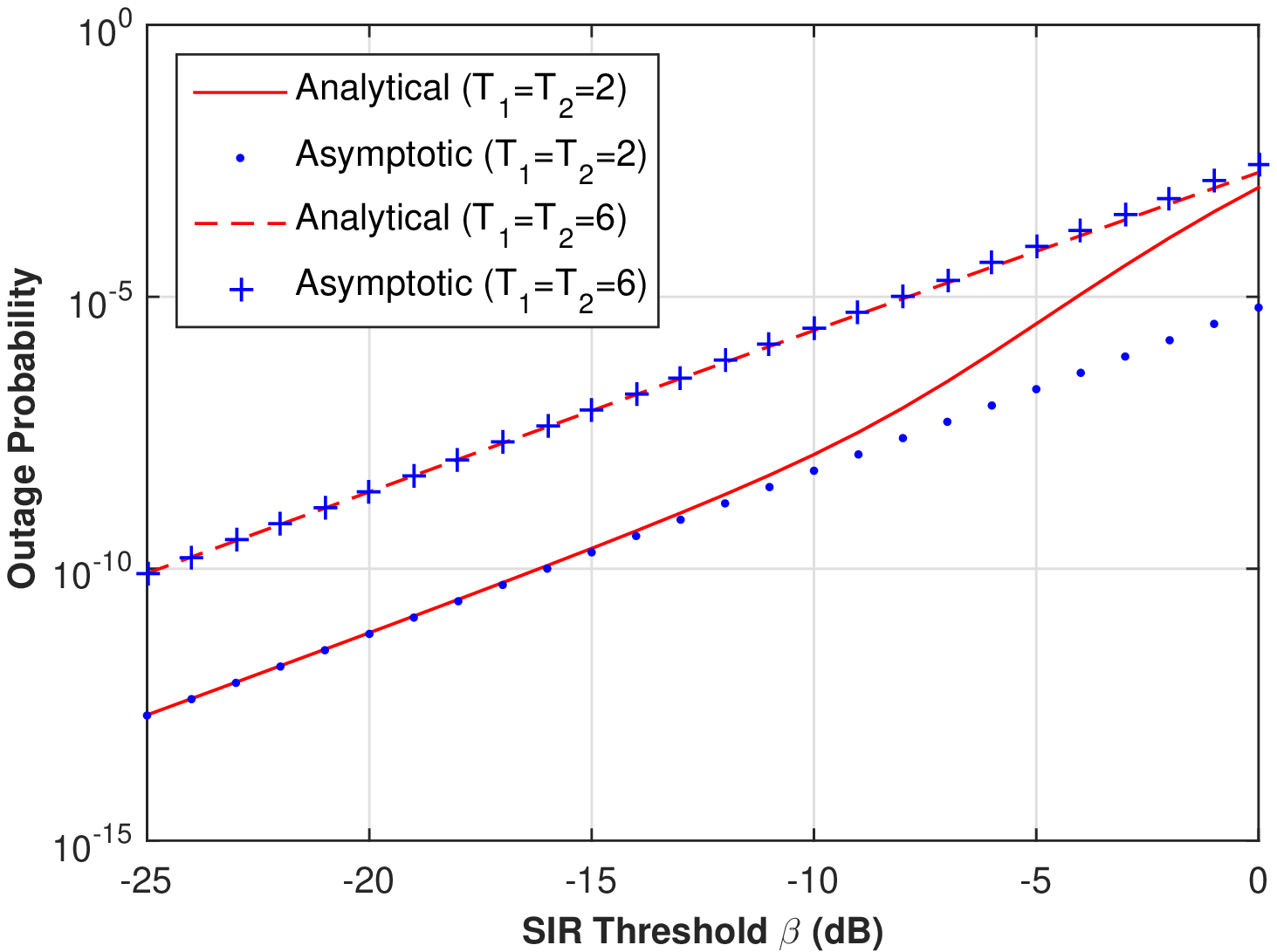}
\label{fig:lowSIR U7}
}
\caption{\small{Outage probability versus SIR threshold in the low SIR threshold regime.  $N_{1}=10$, $N_{2}=8$, $\alpha_{1}=4.5$, $\alpha_{2}=4.7$, $\frac{P_{1}}{P_{2}}=15$ dB, $\lambda_{1}=0.0005$ nodes/m$^{2}$, and $\lambda_{2}=0.001$ nodes/m$^{2}$.}}\label{fig:lowSIR}
\vspace{-5mm}
\end{figure}

From \emph{Theorem \ref{them:CPj_lowbeta}}, we clearly see that the maximum IN DoF $U$ and the IN thresholds $(T_1,T_2)$ affect the asymptotic behavior of the outage probability in dramatically different ways. Specifically,  $U$ can affect the order gain, while  $(T_1,T_2)$ can only affect the coefficient. In addition, we see that $U$ affects the order gain of the asymptotic outage probability through affecting the order gain of the asymptotic macro-user outage probability. On the other hand, in this paper, IN is only performed at macro-BSs, and $U$ is the upper bound on the actual DoF for IN in the ZFBF precoder (which is \emph{random} due to the randomness of the network topology). Therefore, the result of the order gain in \emph{Theorem \ref{them:CPj_lowbeta}} extends the existing order gain result in \emph{single-tier} cellular networks where the DoF for IN in the ZFBF precoder is \emph{deterministic} \cite{huang13}.

 Fig.~\ref{fig:lowSIR} plots the outage probability versus the  SIR threshold in  the low SIR threshold regime.
  We see from  Fig.~\ref{fig:lowSIR} that when the SIR threshold is small, the ``Analytical" curves, which are plotted using  \emph{Theorem \ref{thm:overall_CP}}, are reasonably close to the ``Asymptotic" curves, which are plotted using \emph{Theorem \ref{them:CPj_lowbeta}}.    In addition, from  Fig.~\ref{fig:lowSIR}, we clearly see that the outage probability curves with the same $U$ have the same slope (indicating the same order gain), and  there is a shift between two outage probability curves with the same $U$ but different $(T_{1},T_2)$ (indicating different coefficients). Therefore, Fig.~\ref{fig:lowSIR} verifies \emph{Theorem \ref{them:CPj_lowbeta}}, and shows that the asymptotic outage probability in  the low SIR threshold regime provides a reasonable approximation for the outage probability when the SIR threshold is below -5 dB. 

\subsection{Asymptotic Outage Probability Optimization}
From \emph{Theorem \ref{them:CPj_lowbeta}}, we know that $U$ has a larger impact on the asymptotic outage probability than the IN thresholds. In this part, we characterize the optimal maximum IN DoF $U^{*}(\beta,T_{1},T_{2})$ which minimizes the asymptotic outage probability  given in \emph{Theorem \ref{them:CPj_lowbeta}} (maximizes the asymptotic coverage probability) for given thresholds $T_1$ and $T_2$, i.e.,
\small{\begin{align}\label{eq:optU_asymOP}
U^{*}(\beta,T_{1},T_{2})\define\arg\:\min_{U\in\{0,1,\ldots,N_{1}-1\}}b\left(U,T_{1},T_{2}\right)\beta^{\min\{N_{1}-U,N_{2}\}}\;.
\end{align}}\normalsize

\begin{lemma}[Optimality Property of $U^{*}(\beta,T_{1},T_{2})$]\label{lem:opt_U}
$\exists\bar{\beta}>0$ such that for all  $\beta<\bar{\beta}$, we have\footnote{\emph{Lemma \ref{lem:opt_U}}  is similar to  \emph{Theorem 3} of our previous work \cite{wu14}. The reason is that the two interference management schemes in this paper and \cite{wu14} are both based on IN. One difference is that the proposed  scheme in this paper aims to improve the performance of all users with low SIIR, while the scheme in  \cite{wu14} only improves the performance of offloaded users.\label{ft:relation}}
\small{\begin{align}U^{*}(\beta, T_{1},T_{2})=
\begin{cases}
&\hspace{-2mm}N_{1}-N_{2}-1,\;{\rm if}\:  \mathcal A_2 b_{2}\left(N_{1}-N_{2}-1,T_{1},T_{2}\right)<\mathcal A_1 b_{1}\left(N_{1}-N_{2},T_{1},T_{2}\right)+ \mathcal A_2 b_{2}\left(N_{1}-N_{2},T_{1},T_{2}\right)\\
&\hspace{-2mm}N_{1}-N_{2},\;\hspace{6mm}{\rm otherwise}
\end{cases}
.\notag
\end{align}}\normalsize
\end{lemma}
\begin{proof}
See Appendix \ref{proof:lem_optU}.
\end{proof}

\begin{figure}[t]
\centering
\subfigure[$U^*=1$]{
\includegraphics[height=0.28\columnwidth,width=0.38\columnwidth]
{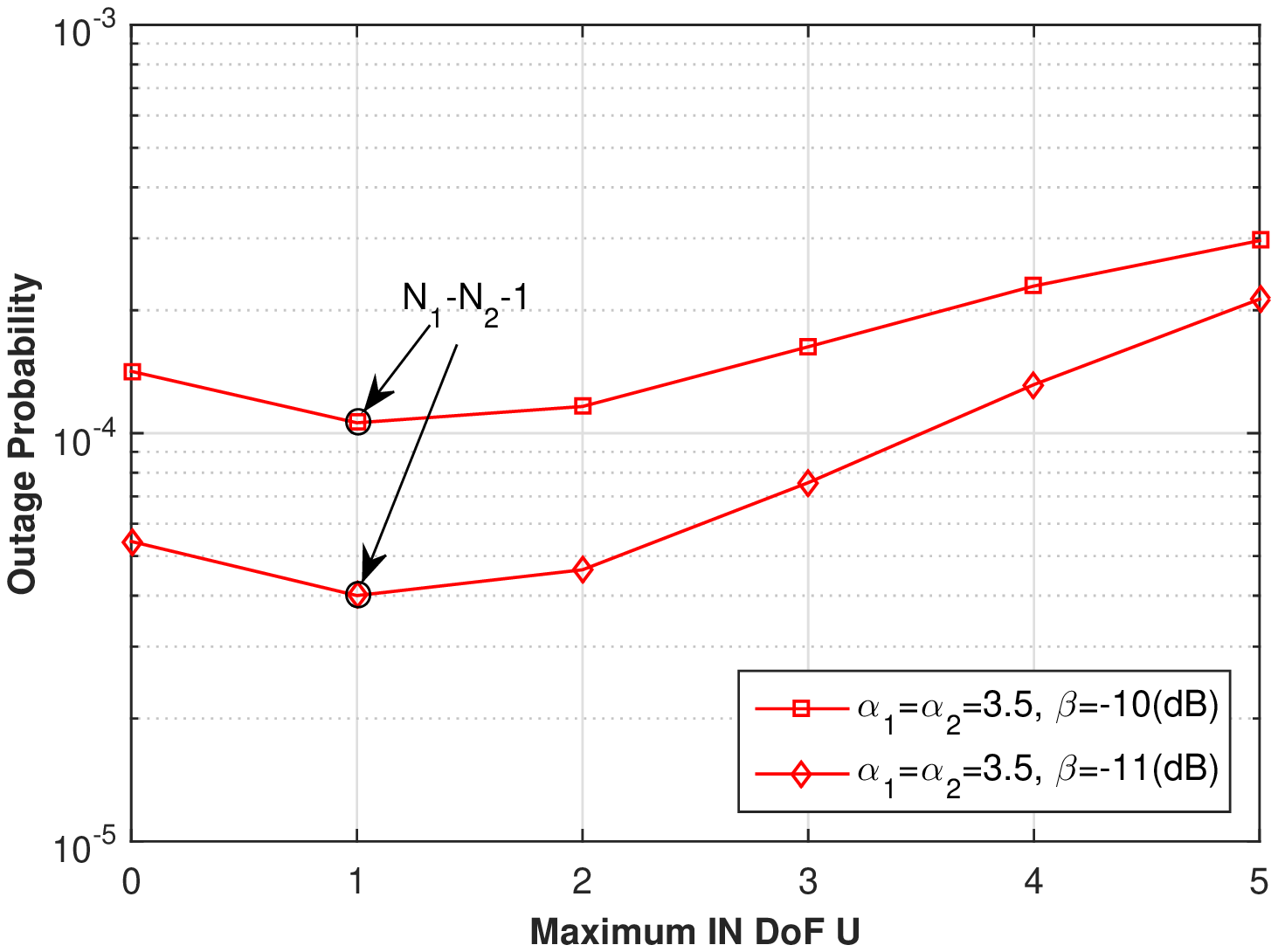}
\label{fig:lowSIR opt U1}
}
\subfigure[$U^*=2$]{
\includegraphics[height=0.29\columnwidth,width=0.38\columnwidth]
{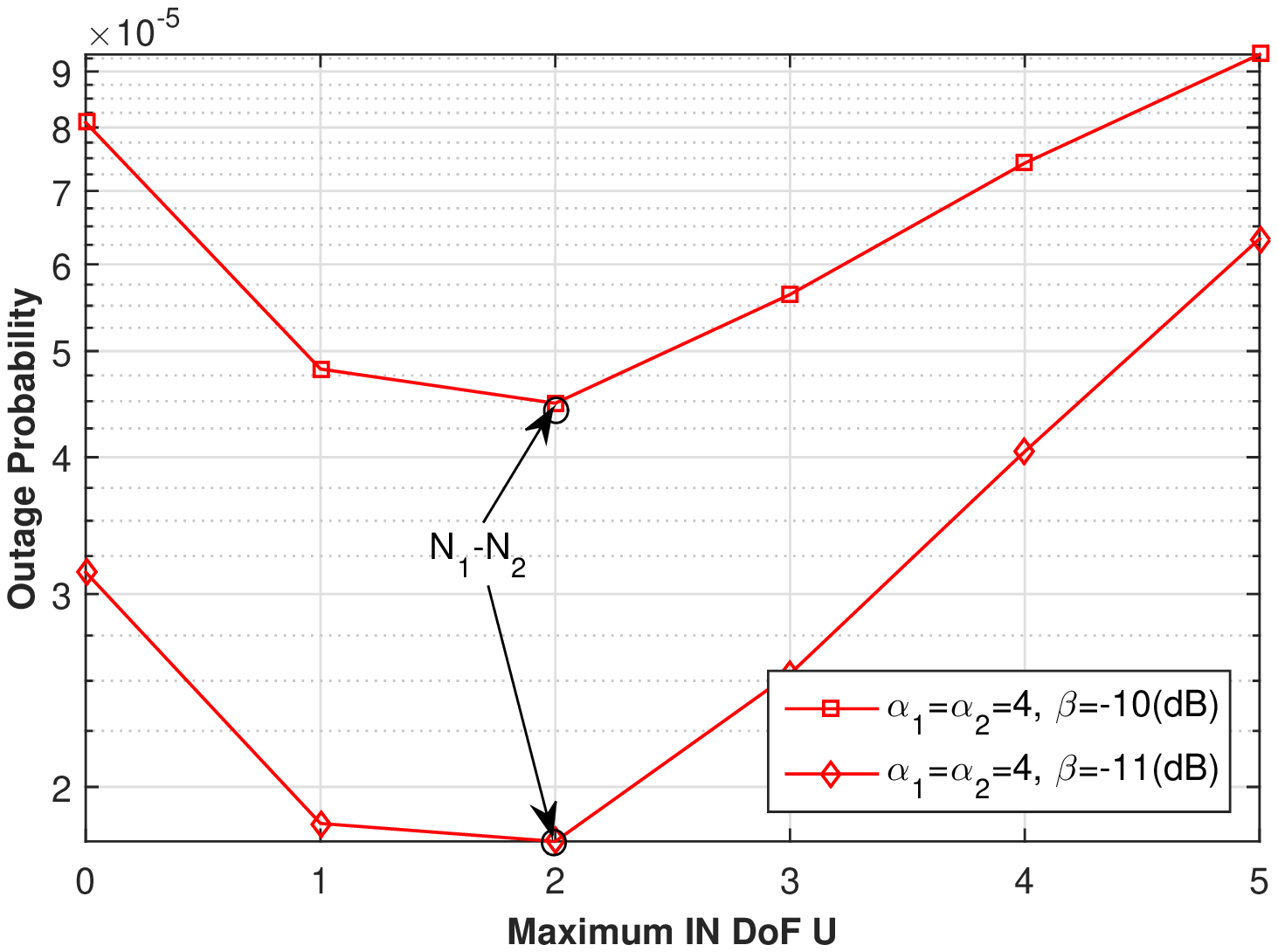}
\label{fig:lowSIR opt U2}
}
\caption{\small{Outage probability versus maximum IN DoF in the low SIR threshold regime.  $N_{1}=6$, $N_{2}=4$, $T_1=T_2=1.8$ $\frac{P_{1}}{P_{2}}=15$ dB, $\lambda_{1}=0.0005$ nodes/m$^{2}$, and $\lambda_{2}=0.001$ nodes/m$^{2}$.}}\label{fig:lowSIR opt}
\vspace{-5mm}
\end{figure}

\begin{figure}[t] \centering
\includegraphics[width=0.4\columnwidth]{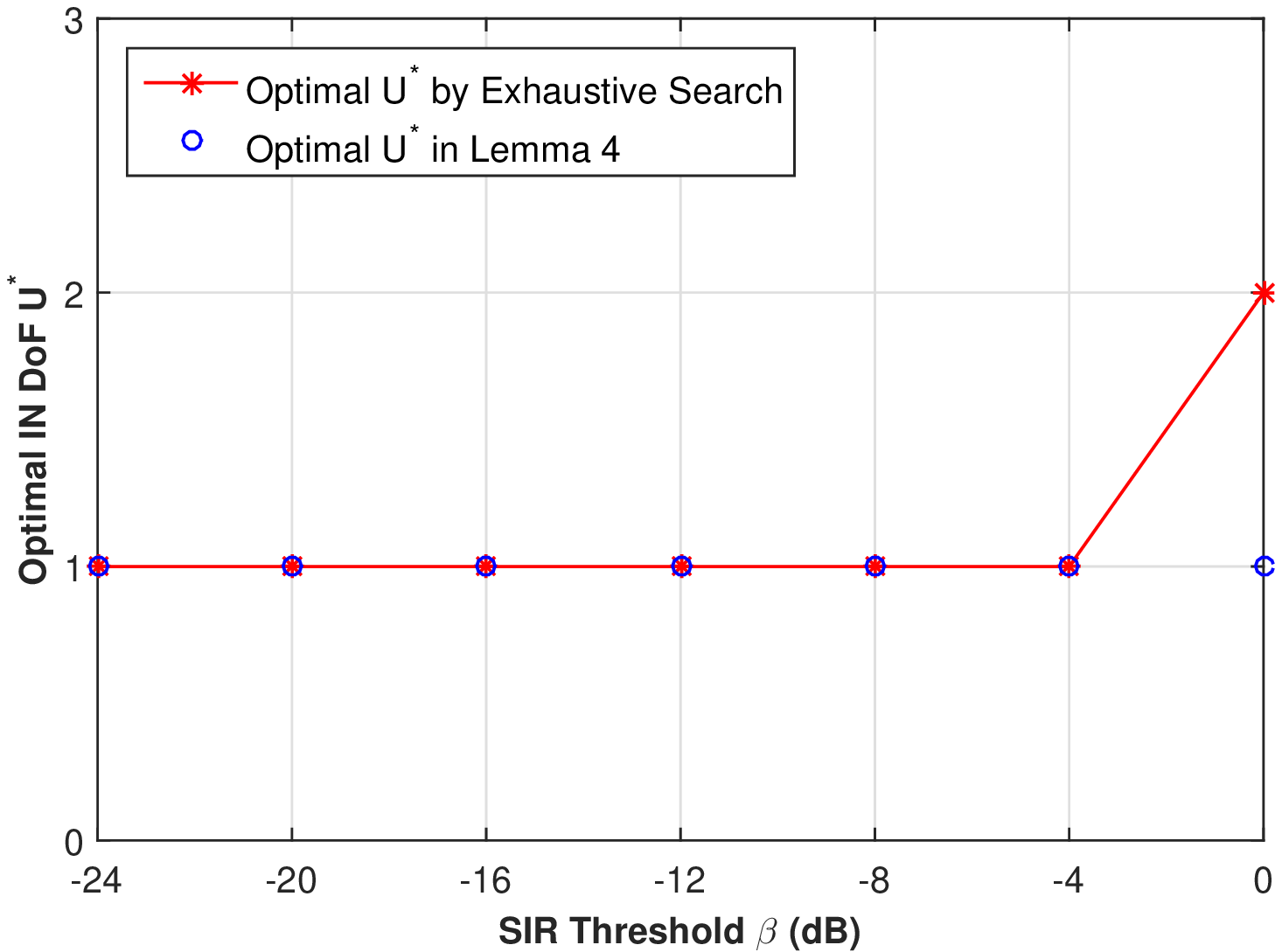}
\caption{\small{Optimal maximum IN DoF  versus SIR threshold  in the low SIR threshold regime. $N_{1}=6$, $N_{2}=4$, $T_1=T_2=2$, $\frac{P_{1}}{P_{2}}=15$ dB, $\alpha_{1}=\alpha_{2}=4$, $\lambda_{1}=0.0005$ nodes/m$^{2}$, and $\lambda_{2}=0.001$ nodes/m$^{2}$.}}\label{application_low_SIR}
\vspace{-5mm}
\end{figure}

\emph{Lemma \ref{lem:opt_U}} indicates that in the low threshold  regime,  the IN scheme achieves  the optimal asymptotic outage probability when reserving  $N_2$ or $N_2+1$ DoF at each macro-BS to boost the desired signal to its scheduled user, which is comparable to the $N_2$ DoF used at each pico-BS  to boost the desired signal to its scheduled user. The reason is that in the low threshold regime, the network performance is mainly limited by the worst users. Balancing the DoF for boosting signals to all the users effectively improves the performance of the worst users.

Fig.~\ref{fig:lowSIR opt} plots the outage probability versus the maximum IN DoF in the  small SIR threshold regime. From Fig.~\ref{fig:lowSIR opt}, we can see that $U^{*}(\beta, T_{1},T_{2})=N_{1}-N_{2}-1$ or $N_{1}-N_{2}$ at small $\beta$. This verifies \emph{Lemma \ref{lem:opt_U}}. Fig.~\ref{application_low_SIR} shows that the asymptotically optimal solution in \emph{Lemma \ref{lem:opt_U}} for the low SIR threshold regime is optimal when the SIR threshold is below -4 dB, as it is the same as the optimal solution optimizing the coverage probability  in \emph{Theorem \ref{thm:overall_CP}} for the general SIR threshold regime. Therefore, the asymptotically optimal solution  in \emph{Lemma \ref{lem:opt_U}} provides good guidance on choosing effective maximum IN DoF when the SIR threshold is relatively small.

\section{Asymptotic Coverage Probability Analysis--High SIR Threshold Regime}\label{sec:CP_high}
In this section, we   analyze and optimize  the coverage probability of the IN scheme in the high   SIR threshold regime, i.e., $\beta\to\infty$.  The asymptotic analysis and optimization   offer important design insights for practical HetNets. 


\subsection{Asymptotic Coverage Probability Analysis}


In this part, we analyze the asymptotic coverage probability of the IN scheme when $\beta\to\infty$.
First, we define the \emph{order gain} of the coverage probability  (in interference-limited systems), i.e.,
the exponent of coverage probability as the SIR threshold increases to infinity:
\begin{align}
d_{c}\define \lim_{\beta\to\infty}\frac{{\rm Pr}\left({\rm SIR}_{0}>\beta\right)}{\log\beta}\;.
\end{align}
Then, we define the \emph{coefficient} of the asymptotic coverage probability: $\lim_{\beta\to\infty}\frac{{\rm Pr}\left({\rm SIR}_{0}>\beta\right)}{\beta^{d_c}}$.  Similarly, leveraging the order gain and the coefficient of the coverage probability, we shall characterize the key behavior of the complex coverage probability in the high SIR threshold regime. 
In the following, we analyze the asymptotic coverage probability in  two scenarios, i.e., $\alpha_1\neq \alpha_2$ and $\alpha_1=\alpha_2$.

When $\alpha_1\neq \alpha_2$, it turns out to be difficult to obtain the expression of the asymptotic coverage probability. Thus, we  derive lower and upper bounds on the asymptotic coverage probability, which are given in the following theorem.
\begin{theorem}[Asymptotic Coverage Probability When $\alpha_{1}\neq\alpha_{2}$]\label{them:CP_highbeta_bound}
Under design parameters $U$, $T_{1}$ and $T_{2}$, when $\alpha_1\neq\alpha_2$ and $\beta\to\infty$, we have:\footnote{$f(\beta)\siml g(\beta)$ means $\lim_{\beta\to\infty}\frac{f(\beta)}{g(\beta)}=1$.} 1) coverage probability of a macro-user: $\mathcal{S}_1\left(\beta,U,T_{1},T_{2}\right)$ $\siml\mathcal{\tilde{S}}_{1}\left(\beta,U,T_{1},T_{2}\right)$, where $\xi_{1}\beta^{-\frac{2}{\alpha_{1}}\frac{\alpha_{\max}}{\alpha_{\min}}}<\mathcal{\tilde{S}}_{1}\left(\beta,U,T_{1},T_{2}\right)<\eta_{1}\left(U,T_1,T_2\right)\beta^{-\frac{2}{\alpha_{1}}}$;
2) coverage probability of a pico-user: $\mathcal{S}_{2}\left(\beta,U,T_{1},T_{2}\right)\siml\mathcal{\tilde{S}}_{2}\left(\beta,U,T_{1},T_{2}\right)$, where $\xi_{2}\beta^{-\frac{2}{\alpha_{2}}\frac{\alpha_{\max}}{\alpha_{\min}}}<\mathcal{\tilde{S}}_{2}\left(\beta,U,T_{1},T_{2}\right)$ $<\eta_{2}\beta^{-\frac{2}{\alpha_{2}}}$;
3) overall coverage probability: $\mathcal{S}\left(\beta,U,T_{1},T_{2}\right)\siml\mathcal{\tilde{S}}\left(\beta,U,T_{1},T_{2}\right)$, where $c^{\rm lb}\beta^{-\frac{2}{\alpha_{\min}}}<\mathcal{\tilde{S}}\left(\beta,U,T_{1},T_{2}\right)<c^{\rm ub}\left(U,T_{1},T_{2}\right)\beta^{-\frac{2}{\alpha_{\max}}}$.
Here, $\alpha_{\min}=\min\left\{\alpha_{1},\alpha_{2}\right\}$, $\alpha_{\max}=\max\left\{\alpha_{1},\alpha_{2}\right\}$, $B(a,b)\triangleq\int_{0}^{1}t^{a-1}(1-t)^{b-1}{\rm d}t$ is the beta function,
$\eta_{1}\left(U,T_{1},T_{2}\right)$ and $\eta_{2}$ are given in (\ref{eq:c1_arb}) and (\ref{eq:c2_arb}), respectively,  $\xi_{j}$ ($j=1,2$) are given in (\ref{eq:c1_arb_lb}) and (\ref{eq:c2_arb_lb}), respectively,
and
\small{\begin{align}
c^{\rm ub}\left(U,T_{1},T_{2}\right)=
\begin{cases}
\eta_{1}\left(U,T_{1},T_{2}\right), & \alpha_1>\alpha_2\\
\eta_{2}, & \alpha_1<\alpha_2
\end{cases},\
c^{\rm lb}=
\begin{cases}
\xi_{1}, & \alpha_1>\alpha_2\\
\xi_{2}, & \alpha_1<\alpha_2
\end{cases}.\nonumber
\end{align}}\normalsize
\begin{figure*}[!t]
\small{\begin{align}
&\eta_{1}\left(U,T_1,T_2\right)=\frac{\pi\lambda_{1}}{\mathcal{A}_{1}}\sum_{u=0}^{U}{\rm Pr}\left(u_{{\rm IN},0}=u\right)\sum_{n=0}^{N_{1}-u-1}\frac{1}{n!}\sum_{n_{2}=0}^{n}\binom{n}{n_{2}}\sum_{(p_{a})_{a=1}^{n_{2}}\in\mathcal{M}_{n_{2}}}\sum_{(q_{a})_{a=1}^{n-n_{2}}\in\mathcal{M}_{n-n_{2}}}\frac{n_{2}!}{\prod_{a=1}^{n_{2}}p_{a}!}\notag\\
&\hspace{10mm}\times \prod_{a=1}^{n_{2}}\left(\frac{2\pi}{\alpha_{1}}\lambda_{1}{ B}\left(1+\frac{2}{\alpha_{1}},a-\frac{2}{\alpha_{1}}\right)\right)^{p_{a}}\frac{\left(n-n_{2}\right)!}{\prod_{a=1}^{n-n_{2}}q_{a}!}\prod_{a=1}^{n-n_{2}}\left(\frac{2\pi}{\alpha_{2}}\lambda_{2}\left(\frac{P_{2}}{P_{1}}\right)^{\frac{2}{\alpha_{2}}}{ B}\left(1+\frac{2}{\alpha_{2}},a-\frac{2}{\alpha_{2}}\right)\right)^{q_{a}}\notag\\
&\hspace{10mm}\times\left(\frac{2\pi\lambda_{1}}{\alpha_{1}}{ B}\left(\frac{2}{\alpha_{1}},1-\frac{2}{\alpha_{1}}\right)\right)^{-\sum_{a=1}^{n_{2}}p_{a}-\frac{\alpha_{1}}{\alpha_{2}}\sum_{a=1}^{n-n_{2}}q_{a}-1}\Gamma\left(\sum_{a=1}^{n_{2}}p_{a}+\frac{\alpha_{1}}{\alpha_{2}}\sum_{a=1}^{n-n_{2}}q_{a}+1\right)
\label{eq:c1_arb}\\
&\eta_{2}=\frac{\pi\lambda_{2}}{\mathcal{A}_{2}}\sum_{n=0}^{N_{2}-1}\frac{1}{n!}\sum_{n_{2}=0}^{n}\binom{n}{n_{2}}\sum_{(p_{a})_{a=1}^{n_{2}}\in\mathcal{M}_{n_{2}}}\sum_{(q_{a})_{a=1}^{n-n_{2}}\in\mathcal{M}_{n-n_{2}}}\frac{n_{2}!}{\prod_{a=1}^{n_{2}}p_{a}!}\frac{\left(n-n_{2}\right)!}{\prod_{a=1}^{n-n_{2}}q_{a}!}\notag\\
&\hspace{10mm}\times \prod_{a=1}^{n_{2}}\left(\frac{2\pi}{\alpha_{1}}\lambda_{1}\left(\frac{P_{1}}{P_{2}}\right)^{\frac{2}{\alpha_{1}}}{ B}\left(1+\frac{2}{\alpha_{1}},a-\frac{2}{\alpha_{1}}\right)\right)^{p_{a}}\prod_{a=1}^{n-n_{2}}\left(\frac{2\pi}{\alpha_{2}}\lambda_{2}{ B}\left(1+\frac{2}{\alpha_{2}},a-\frac{2}{\alpha_{2}}\right)\right)^{q_{a}}\notag\\
&\hspace{10mm}\times\left(\frac{2\pi\lambda_{1}}{\alpha_{2}}{ B}\left(\frac{2}{\alpha_{2}},1-\frac{2}{\alpha_{2}}\right)\right)^{-\frac{\alpha_{2}}{\alpha_{1}}\sum_{a=1}^{n_{2}}p_{a}-\sum_{a=1}^{n-n_{2}}q_{a}-1}\Gamma\left(\frac{\alpha_{2}}{\alpha_{1}}\sum_{a=1}^{n_{2}}p_{a}+\sum_{a=1}^{n-n_{2}}q_{a}+1\right)
\label{eq:c2_arb}\\
&\xi_{1}=\frac{\pi\lambda_{1}\alpha_{\max}}{\mathcal{A}_{1}\alpha_{1}}\left(\frac{2\pi\lambda_{1}}{\alpha_{1}}{ B}\left(\frac{2}{\alpha_{1}},1-\frac{2}{\alpha_{1}}\right)+\frac{2\pi\lambda_{2}}{\alpha_{2}}\left(\frac{P_{2}}{P_{1}}\right)^{\frac{2}{\alpha_{2}}}{ B}\left(\frac{2}{\alpha_{2}},1-\frac{2}{\alpha_{2}}\right)\right)^{-\frac{\alpha_{\max}}{\alpha_{1}}}\Gamma\left(\frac{\alpha_{\max}}{\alpha_{1}}\right)\label{eq:c1_arb_lb}\\
&\xi_{2}=\frac{\pi\lambda_{2}\alpha_{\max}}{\mathcal{A}_{2}\alpha_{2}}\left(\frac{2\pi\lambda_{1}}{\alpha_{1}}\left(\frac{P_{1}}{P_{2}}\right)^{\frac{2}{\alpha_{1}}}{ B}\left(\frac{2}{\alpha_{1}},1-\frac{2}{\alpha_{1}}\right)+\frac{2\pi\lambda_{2}}{\alpha_{2}}{ B}\left(\frac{2}{\alpha_{2}},1-\frac{2}{\alpha_{2}}\right)\right)^{-\frac{\alpha_{\max}}{\alpha_{2}}}\Gamma\left(\frac{\alpha_{\max}}{\alpha_{2}}\right)\label{eq:c2_arb_lb}
\end{align}}\normalsize
\normalsize \hrulefill
\end{figure*}
\end{theorem}
\begin{proof}
See Appendix \ref{proof:CP_highbeta_uneql}.
\end{proof}

\begin{figure}[t]
\centering
\subfigure[$\alpha_{1}=4$, $\alpha_{2}=3.5$]{
\includegraphics[height=0.28\columnwidth,width=0.38\columnwidth]
{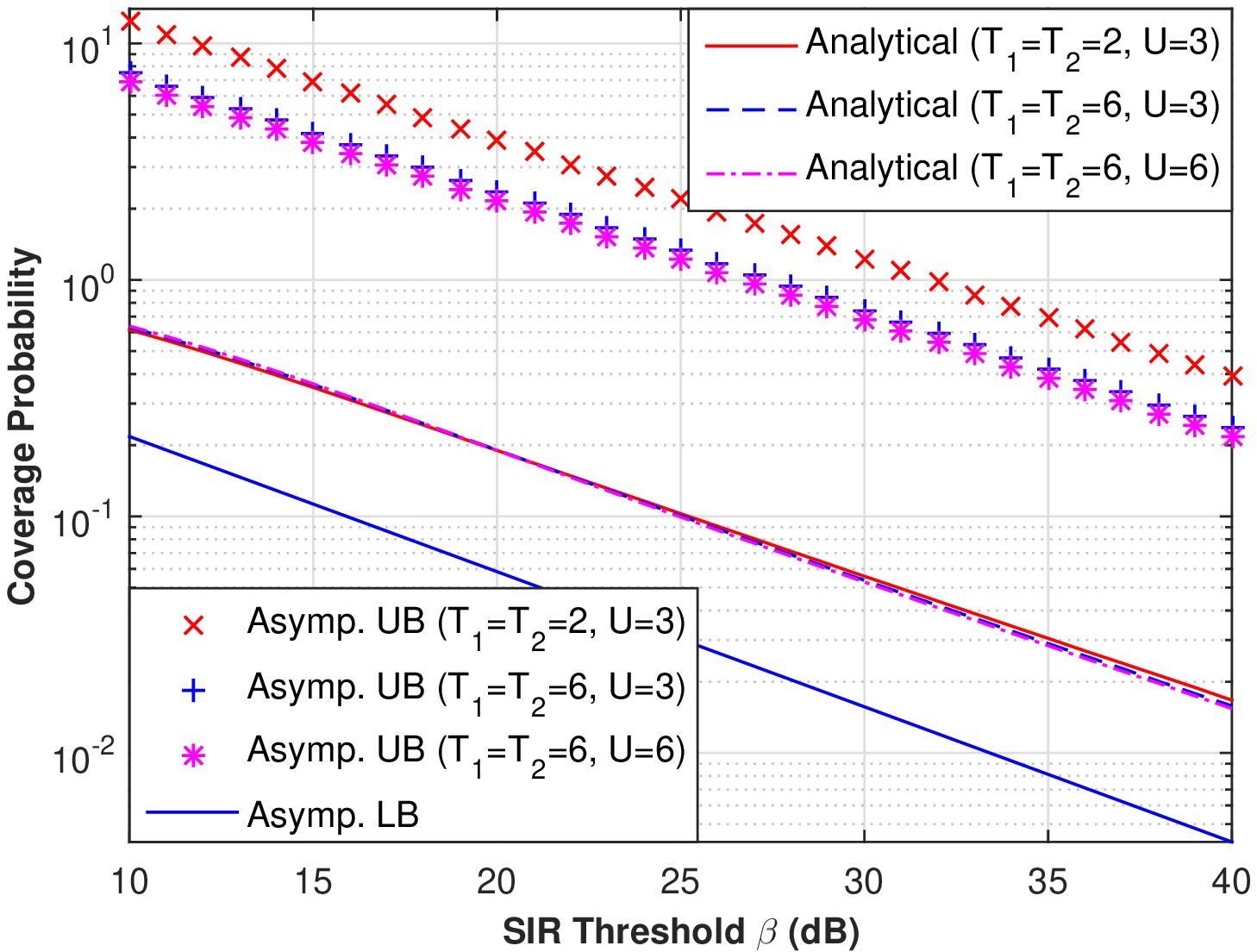}
\label{fig:highSIR_n}
}
\subfigure[$\alpha_{1}=\alpha_{2}=4$]{
\includegraphics[height=0.28\columnwidth,width=0.38\columnwidth]
{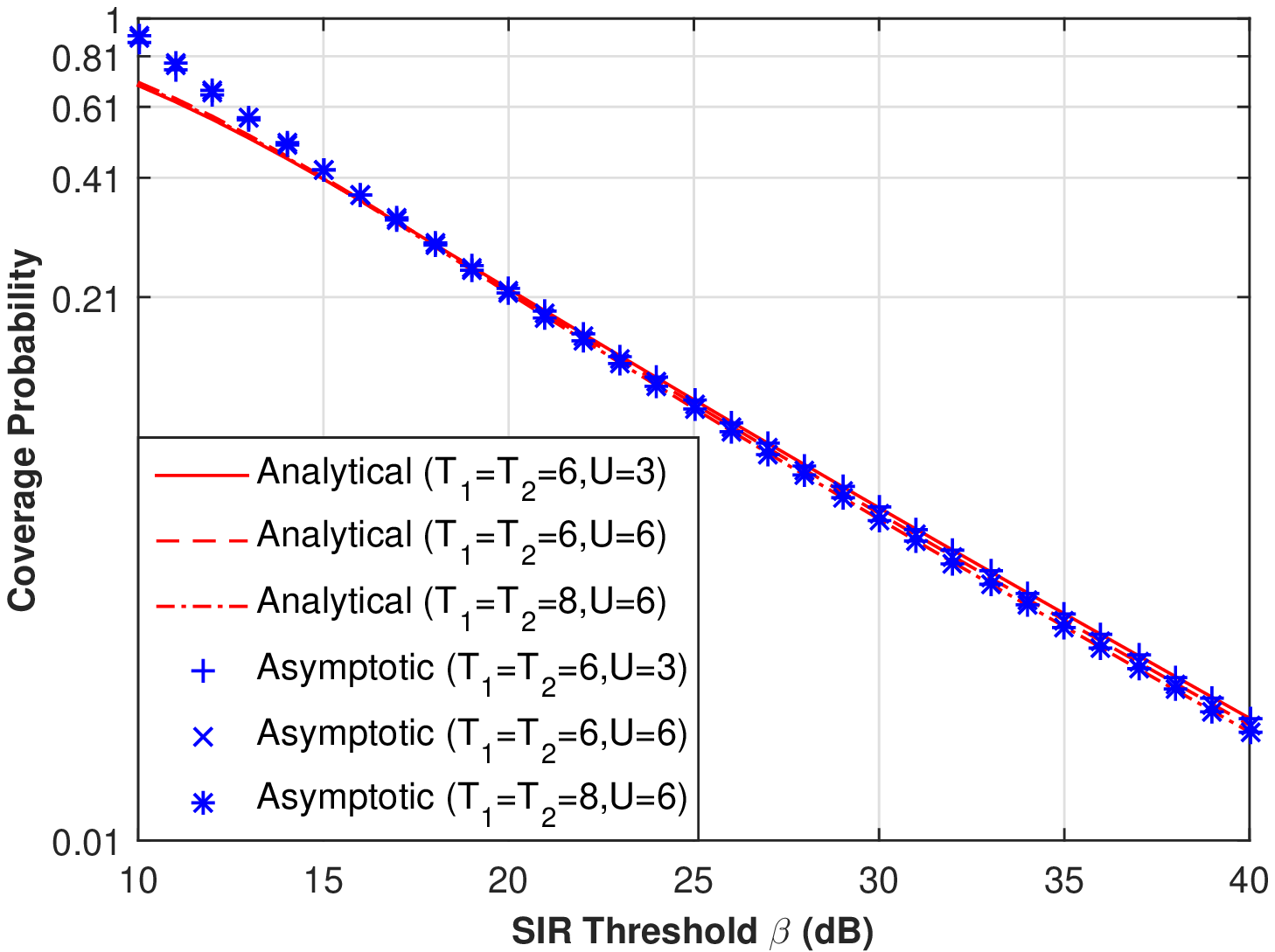}
\label{fig:highSIR_e}
}
\caption{\small{Coverage probability versus SIR threshold in the high SIR threshold regime. $N_{1}=10$, $N_{2}=8$, $\frac{P_{1}}{P_{2}}=15$ dB,  $\lambda_{1}=0.0005$ nodes/m$^{2}$, and $\lambda_{2}=0.001$ nodes/m$^{2}$.}}\label{fig:highSIR}
\vspace{-5mm}
\end{figure}

When $\alpha_1= \alpha_2$, we    derive the asymptotic coverage probability, which is given below.

\begin{theorem}[Asymptotic Coverage Probability When $\alpha_{1}=\alpha_{2}$]\label{them:CP_highbeta_eql}
Under design parameters $U$, $T_{1}$ and $T_{2}$, when $\alpha_{1}=\alpha_{2}=\alpha$ and $\beta\to\infty$, we have:
\begin{figure*}[!t]
\small{\begin{align}
&c_{1}\left(U,T_{1},T_{2}\right)=\frac{\pi\lambda_{1}}{\mathcal{A}_{1}}\sum_{u=0}^{U}{\rm Pr}\left(u_{{\rm IN},0}=u\right)\sum_{n=0}^{N_{1}-u-1}\frac{1}{n!}\sum_{n_{2}=0}^{n}\binom{n}{n_{2}}\sum_{(p_{a})_{a=1}^{n_{2}}\in\mathcal{M}_{n_{2}}}\sum_{(q_{a})_{a=1}^{n-n_{2}}\in\mathcal{M}_{n-n_{2}}}\frac{n_{2}!}{\prod_{a=1}^{n_{2}}p_{a}!}\notag\\
&\hspace{10mm}\times \prod_{a=1}^{n_{2}}\left(\frac{2\pi}{\alpha}\lambda_{1}{ B}\left(1+\frac{2}{\alpha},a-\frac{2}{\alpha}\right)\right)^{p_{a}}\frac{\left(n-n_{2}\right)!}{\prod_{a=1}^{n-n_{2}}q_{a}!}\prod_{a=1}^{n-n_{2}}\left(\frac{2\pi}{\alpha}\lambda_{2}\left(\frac{P_{2}}{P_{1}}\right)^{\frac{2}{\alpha}}{ B}\left(1+\frac{2}{\alpha},a-\frac{2}{\alpha}\right)\right)^{q_{a}}\notag\\
&\hspace{10mm}\times \left(\frac{2\pi}{\alpha}\lambda_{1}{ B}\left(\frac{2}{\alpha},1-\frac{2}{\alpha}\right)+\frac{2\pi}{\alpha}\lambda_{2}\left(\frac{P_{2}}{P_{1}}\right)^{\frac{2}{\alpha}}{ B}\left(\frac{2}{\alpha},1-\frac{2}{\alpha}\right)\right)^{-\frac{\alpha}{\alpha}\sum_{a=1}^{n_{2}}p_{a}-\frac{\alpha}{\alpha}\sum_{a=1}^{n-n_{2}}q_{a}-\frac{\alpha}{\alpha}}\notag\\
&\hspace{10mm}\times\Gamma\left(\sum_{a=1}^{n_{2}}p_{a}+\sum_{a=1}^{n-n_{2}}q_{a}+1\right)\;\label{eq:c1}\\
&c_{2}\left(T_{1},T_{2}\right)=\frac{\pi\lambda_{2}}{\mathcal{A}_{2}}\sum_{n=0}^{N_{2}-1}\frac{1}{n!}\sum_{n_{2}=0}^{n}\binom{n}{n_{2}}\sum_{(p_{a})_{a=1}^{n_{2}}\in\mathcal{M}_{n_{2}}}\sum_{(q_{a})_{a=1}^{n-n_{2}}\in\mathcal{M}_{n-n_{2}}}\frac{n_{2}!}{\prod_{a=1}^{n_{2}}p_{a}!}\frac{\left(n-n_{2}\right)!}{\prod_{a=1}^{n-n_{2}}q_{a}!}\notag\\
&\hspace{10mm}\times \prod_{a=1}^{n_{2}}\left(\frac{2\pi}{\alpha}\lambda_{1}\left(\frac{P_{1}}{P_{2}}\right)^{\frac{2}{\alpha}}{ B}\left(1+\frac{2}{\alpha},a-\frac{2}{\alpha}\right)\right)^{p_{a}}\prod_{a=1}^{n-n_{2}}\left(\frac{2\pi}{\alpha}\lambda_{2}{ B}\left(1+\frac{2}{\alpha},a-\frac{2}{\alpha}\right)\right)^{q_{a}}\notag\\
&\hspace{10mm}\times \left(\frac{2\pi}{\alpha}\lambda_{1}\left(\frac{P_{1}}{P_{2}}\right)^{\frac{2}{\alpha}}{ B}\left(\frac{2}{\alpha},1-\frac{2}{\alpha}\right)+\frac{2\pi}{\alpha}\lambda_{2}{ B}\left(\frac{2}{\alpha},1-\frac{2}{\alpha}\right)\right)^{-\frac{\alpha}{\alpha}\sum_{a=1}^{n_{2}}p_{a}-\frac{\alpha}{\alpha}\sum_{a=1}^{n-n_{2}}q_{a}-\frac{\alpha}{\alpha}}\notag\\
&\hspace{10mm}\times\Gamma\left(\sum_{a=1}^{n_{2}}p_{a}+\sum_{a=1}^{n-n_{2}}q_{a}+1\right)\;\label{eq:c2}
\end{align}}\normalsize
\normalsize \hrulefill
\end{figure*}
1) coverage probability of a macro-user: $\mathcal{S}_{1}\left(\beta,U,T_{1},T_{2}\right)\siml c_{1}\left(U,T_{1},T_{2}\right)\beta^{-\frac{2}{\alpha}}$,
where $c_{1}\left(U,T_{1},T_{2}\right)$ is given in (\ref{eq:c1});
2) coverage probability of a pico-user: $\mathcal{S}_{2}\left(\beta,U,T_{1},T_{2}\right)\siml c_{2}\left(T_{1},T_{2}\right)\beta^{-\frac{2}{\alpha}}$,
where $c_{2}\left(T_{1},T_{2}\right)$ is given in (\ref{eq:c2});
3) overall coverage probability: $\mathcal{S}(\beta,T_{1},T_{2})\siml\left(\mathcal{A}_{1}c_{1}\left(U,T_{1},T_{2}\right)+\mathcal{A}_{2}c_{2}\left(T_{1},T_{2}\right)\right)\beta^{-\frac{2}{\alpha}}$.
\end{theorem}
\begin{proof}
See Appendix \ref{proof:largebeta_eqlalpha}.
\end{proof}


From \emph{Theorem \ref{them:CP_highbeta_bound}} and \emph{Theorem \ref{them:CP_highbeta_eql}}, we clearly see that when $\alpha_1\neq \alpha_2$, the order gains of the lower and upper bounds on   the asymptotic coverage probability do not depend on $U$, $T_{1}$ and $T_{2}$; when $\alpha_{1}=\alpha_{2}$, the order gain of  the asymptotic coverage probability does not depend on $U$, $T_{1}$ and $T_{2}$. Hence, for arbitrary $\alpha_1$ and $\alpha_2$, the design parameters $U$, $T_{1}$ and $T_{2}$ do not affect the order gain of the asymptotic coverage probability in both scenarios. In other words, the IN scheme does not provide order-wise performance improvement compared to the simple beamforming scheme  without interference management when $\beta\to\infty$.
In addition, $T_{1}$ and $T_{2}$ affect the coefficient of  the upper bound on   the asymptotic coverage probability when $\alpha_1\neq \alpha_2$  and the coefficient of the asymptotic coverage probability when $\alpha_1= \alpha_2$. $U$ affects the coefficient of the upper bound on the asymptotic coverage probability when $\alpha_1\neq \alpha_2$  and the coefficient of the asymptotic coverage probability  when $\alpha_1= \alpha_2$,  through affecting the upper bound on  the asymptotic coverage probability of a macro-user when $\alpha_1\neq \alpha_2$  and the asymptotic coverage probability of a macro-user  when $\alpha_1= \alpha_2$, respectively.

%
%

  Fig.~\ref{fig:highSIR} plots the coverage probability versus the SIR threshold in the high SIR threshold regime for  $\alpha_1\neq\alpha_2$ and $\alpha_1=\alpha_2$, respectively.  We see from  Fig.~\ref{fig:highSIR_n} that when   $\alpha_1\neq\alpha_2$, the ``Analytical" curves, which are plotted using   \emph{Theorem \ref{thm:overall_CP}}, are bounded by the corresponding ``Asymptotic" upper bound curves and lower bound curve, which are plotted using \emph{Theorem~\ref{them:CP_highbeta_bound}}.  Note that there is only one ``Asymptotic'' lower bound curve, as  the asymptotic lower bound is independent of $U$ and $(T_1,T_2)$.  In addition, from  Fig.~\ref{fig:highSIR_n}, we clearly see that the coverage probability curves with different  $U$ or $(T_1,T_2)$ have slightly different  slopes (indicating  different order gains), and  there is a small shift between any two coverage probability curves with different $U$ or $(T_{1},T_2)$ (indicating different coefficients).
On the other hand,
we see from  Fig.~\ref{fig:highSIR_e}  that when $\alpha_1=\alpha_2$, the ``Analytical" curves, which are plotted using $\mathcal{S}\left(\beta,U,T_{1},T_{2}\right)$ in \emph{Theorem \ref{thm:overall_CP}}, are reasonably close to the ``Asymptotic" curves, which are plotted using \emph{Theorem \ref{them:CP_highbeta_eql}}.    In addition, from  Fig.~\ref{fig:highSIR_e}, we clearly see that the coverage probability curves with different   $U$ or $(T_1,T_2)$ have the same slope (indicating the same order gain), and  there is a shift between any two coverage probability curves with  different $U$ or $(T_{1},T_2)$ (indicating different coefficients).
Therefore, Fig.~\ref{fig:highSIR}  verifies \emph{Theorem \ref{them:CP_highbeta_bound}} and \emph{Theorem \ref{them:CP_highbeta_eql}}, and shows that the asymptotic coverage probability in  the high SIR threshold regime provides a reasonable approximation for the coverage probability when the SIR threshold is above 13 dB.

\subsection{Asymptotic Coverage Probability Optimization}

In this part, we characterize the optimal maximum IN DoF $U^{*}(\beta,T_{1},T_{2})$ which maximizes the upper bound on  the asymptotic coverage probability given in \emph{Theorem \ref{them:CP_highbeta_bound}}  when $\alpha_1\neq\alpha_2$ and  the asymptotic coverage probability given in \emph{Theorem \ref{them:CP_highbeta_eql}} when $\alpha_1=\alpha_2$, for given thresholds $T_1$ and $T_2$, i.e.,
\small{\begin{align}\label{eq:optU_asymOP-high}
U^{*}(\beta,T_{1},T_{2})\define&
\begin{cases}
\arg\max_{U\in\{0,1,\ldots,N_{1}-1\}}c^{\rm ub}\left(U,T_{1},T_{2}\right)\beta^{-\frac{2}{\alpha_{\max}}}, & \alpha_1\neq\alpha_2\\
\arg\max_{U\in\{0,1,\ldots,N_{1}-1\}}\left(\mathcal A_1 c_{1}\left(U,T_{1},T_{2}\right)+ \mathcal A_2 c_{2}\left(T_{1},T_{2}\right)\right)\beta^{-\frac{2}{\alpha}}, &\alpha_1=\alpha_2
\end{cases}\nonumber\\
=&\begin{cases}
\arg\max_{U\in\{0,1,\ldots,N_{1}-1\}}c^{\rm ub}\left(U,T_{1},T_{2}\right), & \alpha_1\neq\alpha_2\\
\arg\max_{U\in\{0,1,\ldots,N_{1}-1\}}c_{1}\left(U,T_{1},T_{2}\right), &\alpha_1=\alpha_2
\end{cases}.
\end{align}}\normalsize
Note that $U$ does not affect the lower bound on  the asymptotic coverage probability given in \emph{Theorem \ref{them:CP_highbeta_bound}}. 
\begin{lemma}[Optimality Property of $U^{*}(\beta,T_{1},T_{2})$]\label{lem:opt_U-high}
There exists $\underline{\beta}<\infty$ such that for all  $\beta>\underline{\beta}$, we have
$U^{*}(\beta, T_{1},T_{2})=0$ for arbitrary $\alpha_1$ and $\alpha_2$.
\end{lemma}
\begin{proof}
See Appendix \ref{proof:lem_optU-high}.
\end{proof}

\begin{figure}[t]
\centering
\subfigure[$\alpha_1\neq \alpha_2$]{
\includegraphics[height=0.28\columnwidth,width=0.38\columnwidth]
{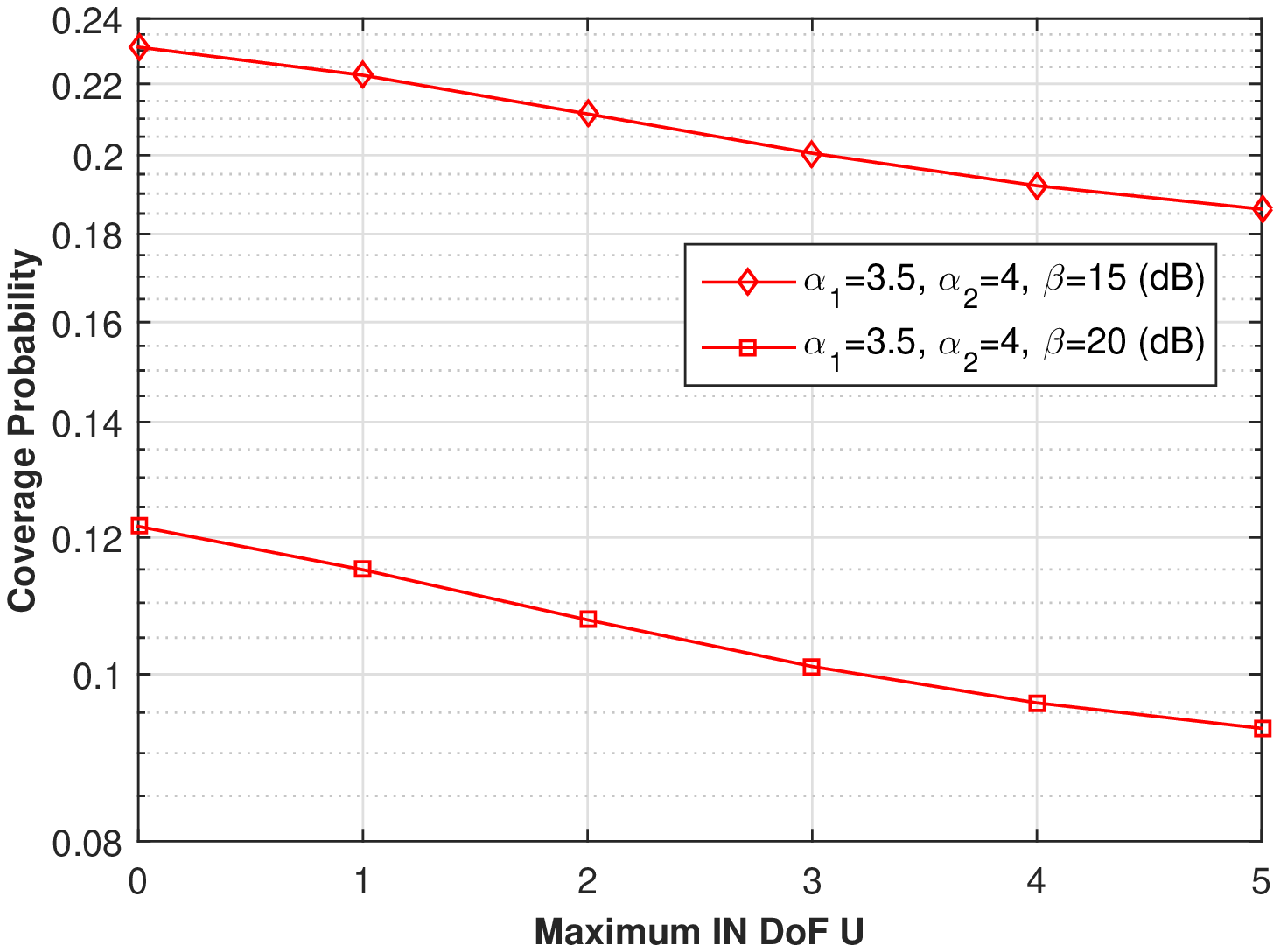}
\label{fig:highSIR_n_opt_U}
}
\subfigure[$\alpha_1=\alpha_2$]{
\includegraphics[height=0.28\columnwidth,width=0.38\columnwidth]
{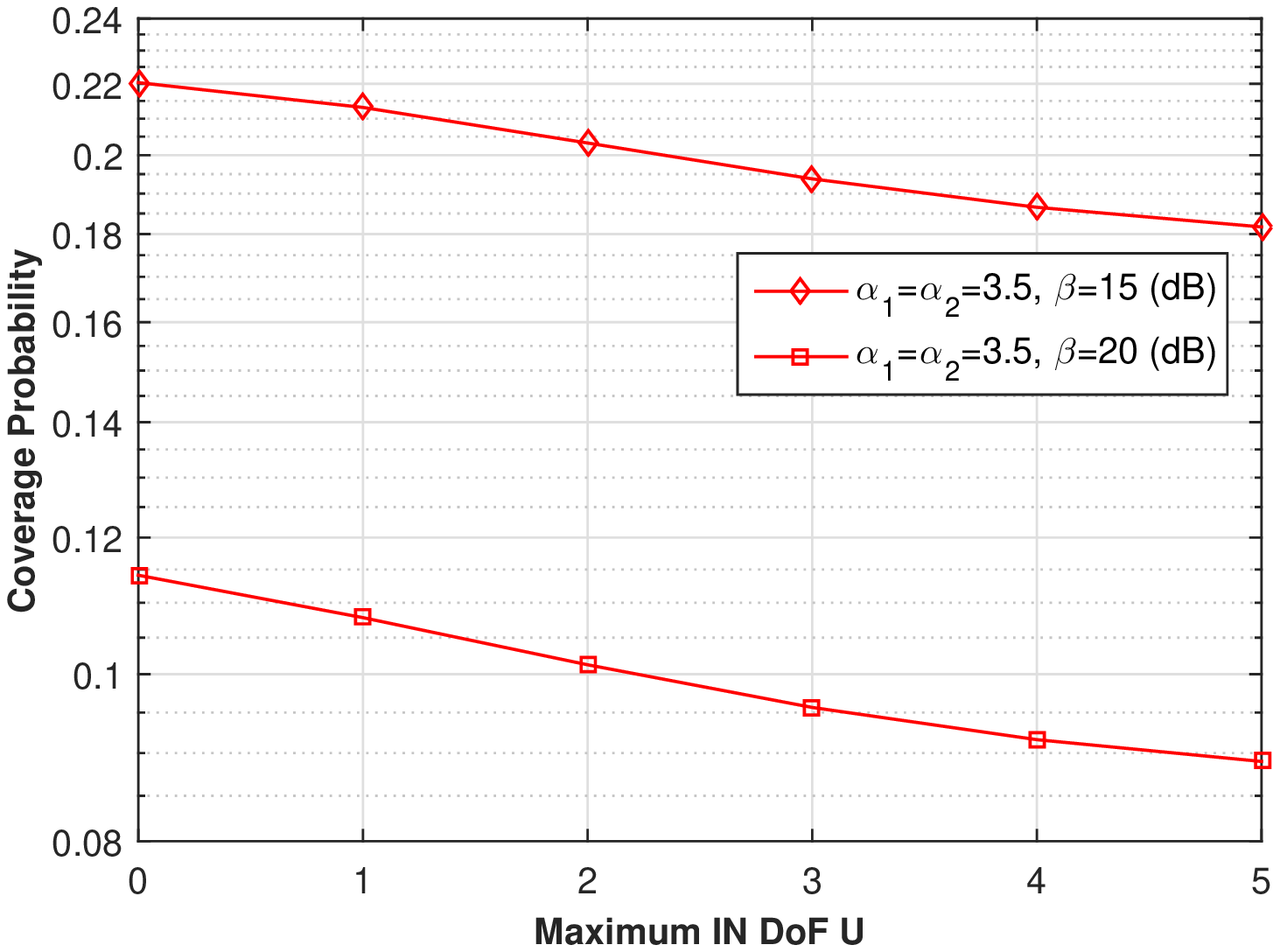}
\label{fig:highSIR_e_opt_U}
}
\caption{\small{Coverage probability versus maximum IN DoF in the high SIR threshold regime.   $N_{1}=6$, $N_{2}=4$, $T_1=T_2=4$ $\frac{P_{1}}{P_{2}}=15$ dB, $\lambda_{1}=0.0005$ nodes/m$^{2}$, and $\lambda_{2}=0.001$ nodes/m$^{2}$.}}\label{fig:highSIR_opt_U}
\vspace{-5mm}
\end{figure}

\begin{figure}[t] \centering
\includegraphics[width=0.4\columnwidth]{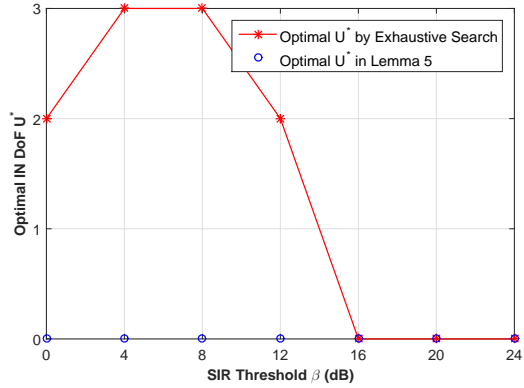}
\caption{\small{Optimal maximum IN DoF  versus SIR threshold  in the high SIR threshold regime. $N_{1}=6$, $N_{2}=4$, $T_1=T_2=2$, $\frac{P_{1}}{P_{2}}=15$ dB, $\alpha_{1}=\alpha_{2}=4$, $\lambda_{1}=0.0005$ nodes/m$^{2}$, and $\lambda_{2}=0.001$ nodes/m$^{2}$.}}\label{application_high_SIR}
\vspace{-5mm}
\end{figure}

\emph{Lemma \ref{lem:opt_U-high}} indicates that performing IN will not improve the asymptotic coverage probability in the high SIR threshold regime.  The reason is that in the high SIR threshold regime, the overall coverage probability is mainly contributed by cell center users, which have much better performance than cell edge users.
 Using all $N_1 $ DoF at each macro-BS to boost the desired signal to its scheduled user can effectively improve the coverage probability of a cell center macro-user, and hence improve the overall coverage probability.

 Fig.~\ref{fig:highSIR_opt_U} plots the coverage probability versus the maximum IN DoF in the high SIR threshold regime. From Fig.~\ref{fig:highSIR_opt_U}, we can see that $U^{*}(\beta,T_{1},T_{2})=0$. This verifies \emph{Lemma \ref{lem:opt_U-high}}. In addition, we can observe that the coverage probability decreases with the  maximum IN DoF. Fig.~\ref{application_high_SIR} shows that the asymptotically optimal solution  in \emph{Lemma \ref{lem:opt_U-high}}  for the high SIR threshold regime is optimal when the SIR threshold is above 16 dB, as it is the same as the optimal solution optimizing the coverage probability  in \emph{Theorem \ref{thm:overall_CP}}  for the general SIR threshold regime. Therefore, the asymptotically optimal solution  in \emph{Lemma \ref{lem:opt_U-high}} provides good guidance on choosing effective maximum IN DoF when the SIR threshold is relatively high.

\section{Numerical Experiments}

In this section, we compare  the proposed user-centric IN scheme with two baseline schemes. One is a simple beamforming scheme (without interference management), which can be treated as a special case of our IN scheme by setting $U=0$ and/or $T_1=T_2=1$.  The other is  a modified version of the existing  ABS scheme in 3GPP-LTE, referred to as the user-centric ABS scheme. The user-centric ABS scheme has three design parameters, i.e., a resource partition  parameter $\eta$ and  two thresholds $T_1$ and $T_2$, where $T_j$ ($j=1,2$) is the threshold for the $j$-th tier. We define a potential ABS macro-BS of a scheduled user in  a similar way to a potential IN macro-BS of a scheduled user  in the user-centric IN scheme. In each slot, each scheduled user sends ABS requests to all of its potential ABS macro-BSs. We define the potential ABS users of a macro-BS in  a similar way to the potential IN users of a  macro-BS in the user-centric IN scheme. $1-\eta$ fraction of (time or frequency) resource is allocated to all the potential ABS macro-BSs to serve their scheduled users, while $\eta$ fraction of resource is allocated to the remaining BSs to serve their own scheduled users. Then, for given $T_1$ and $T_2$, we choose the optimal  $\eta$ to maximize the coverage probability of the user-centric IN scheme.  Under this user-centric ABS scheme, each scheduled potential ABS pico-user or macro-user whose serving macro-BS is not a potential ABS macro-BS  can avoid the interference from all its potential ABS macro-BSs via resource partition in ABS.

Note that the benefit of  the proposed user-centric IN scheme compared to the simple beamforming scheme is that it can optimally allocate  DoF in boosting  desired signals  and managing interference. Thus, the performance of the proposed user-centric IN scheme is always better than that of the simple beamforming scheme. One benefit of the proposed user-centric IN scheme compared to the user-centric ABS
is that it does not have (time or frequency) resource sacrifice. On the other hand,  one loss of the proposed user-centric IN scheme compared to the user-centric  ABS is due to the
DoF  reduction  for boosting desired signals to macro-users.

Fig.~\ref{fig:compare_N1} illustrates the coverage probability versus the number of antennas at each macro-BS $N_1$. From Fig.~\ref{fig:compare_N1}, we can observe that the proposed user-centric IN scheme and the user-centric ABS outperform the simple beamforming scheme, demonstrating the importance of interference management in the parameter region considered in this figure. In addition, the proposed user-centric IN scheme outperforms the user-centric ABS when $N_1$ is relatively large. The reason is as follows. When $N_1$ is relatively large, for serving macro-users, the loss of the user-centric ABS  caused by (time or frequency) resource sacrifice  (due to resource partition) is large, while the loss of the proposed user-centric IN scheme caused by   DoF reduction  (due to  performing  IN) is small. 
Fig.~\ref{fig:compare_alpha1} illustrates  the coverage probability  versus the path loss exponent in the macro-cell tier $\alpha_1$. From Fig.~\ref{fig:compare_alpha1}, we can observe that
the proposed user-centric IN scheme outperforms the user-centric ABS when $\alpha_1$ is relatively large. The reason is as follows. When $\alpha_1$ is large, 
 the loss of the proposed user-centric IN scheme  due to the
DoF  reduction  for boosting desired signals to macro-users is small.\footnote{The observation that the proposed scheme outperforms ABS when $N_1$ or $\alpha_1$ is relatively large is similar to the observation made in \cite{wu14}. The reason can be found in Footnote~\ref{ft:relation}.}  Fig.~\ref{fig:CPvsP_ratio} illustrates the coverage probability versus the power ratio $P_1/P_2$. From Fig.~\ref{fig:CPvsP_ratio}, we can observe that the proposed user-centric IN scheme outperforms the user-centric ABS when $P_1/P_2$ is relatively large. This is because when $P_1/P_2$ is relatively large,   for serving macro-users, the loss of the user-centric ABS  caused by (time or frequency) resource sacrifice  (due to resource partition) is large.



\begin{figure}[t]
\centering
\subfigure[$N_2=2$ and $N_2=4$]{
\includegraphics[height=0.28\columnwidth,width=0.38\columnwidth]
{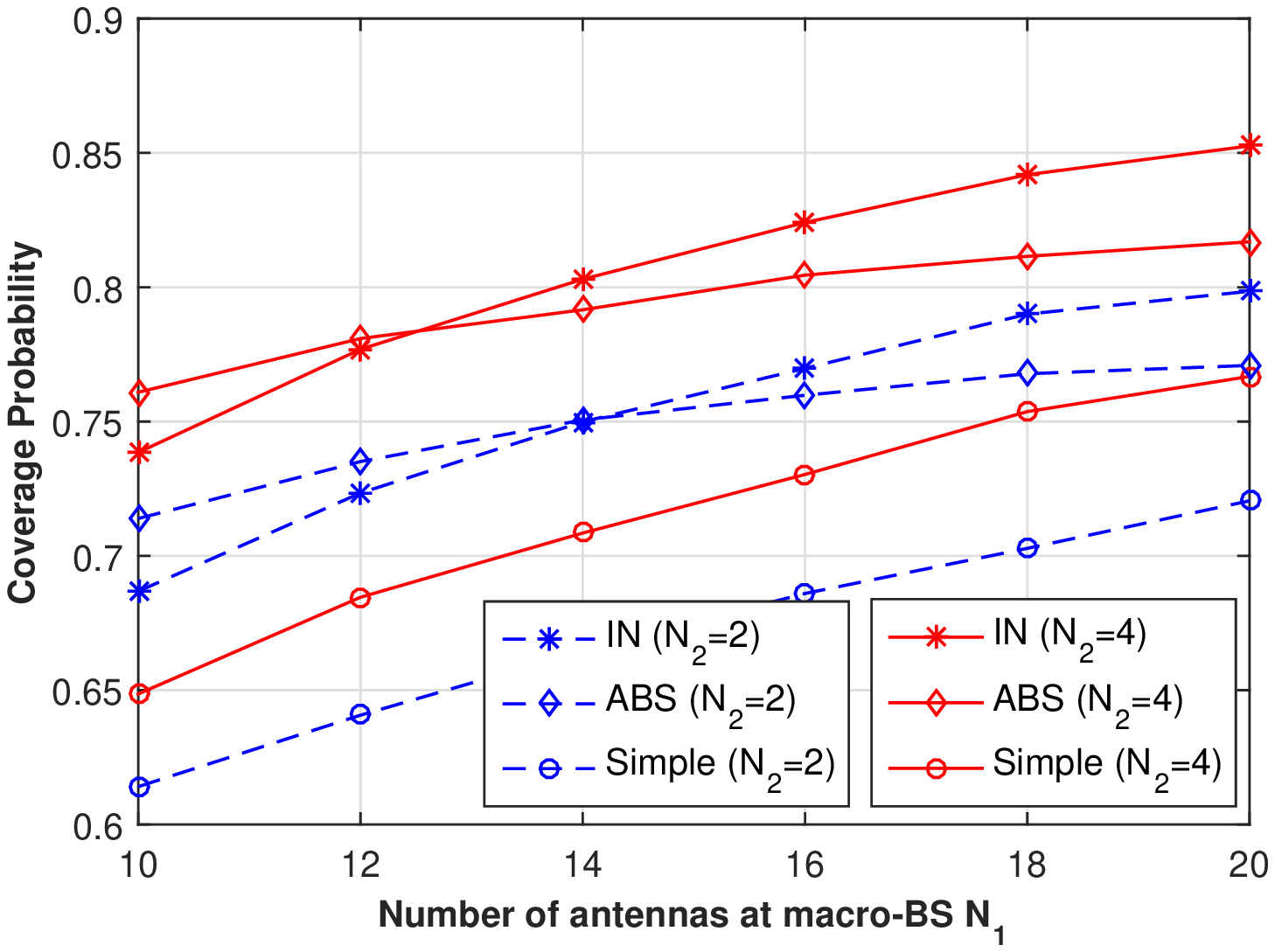}
\label{fig:compare_N1_24}
}
\subfigure[$N_2=6$ and $N_2=8$]{
\includegraphics[height=0.28\columnwidth,width=0.38\columnwidth]
{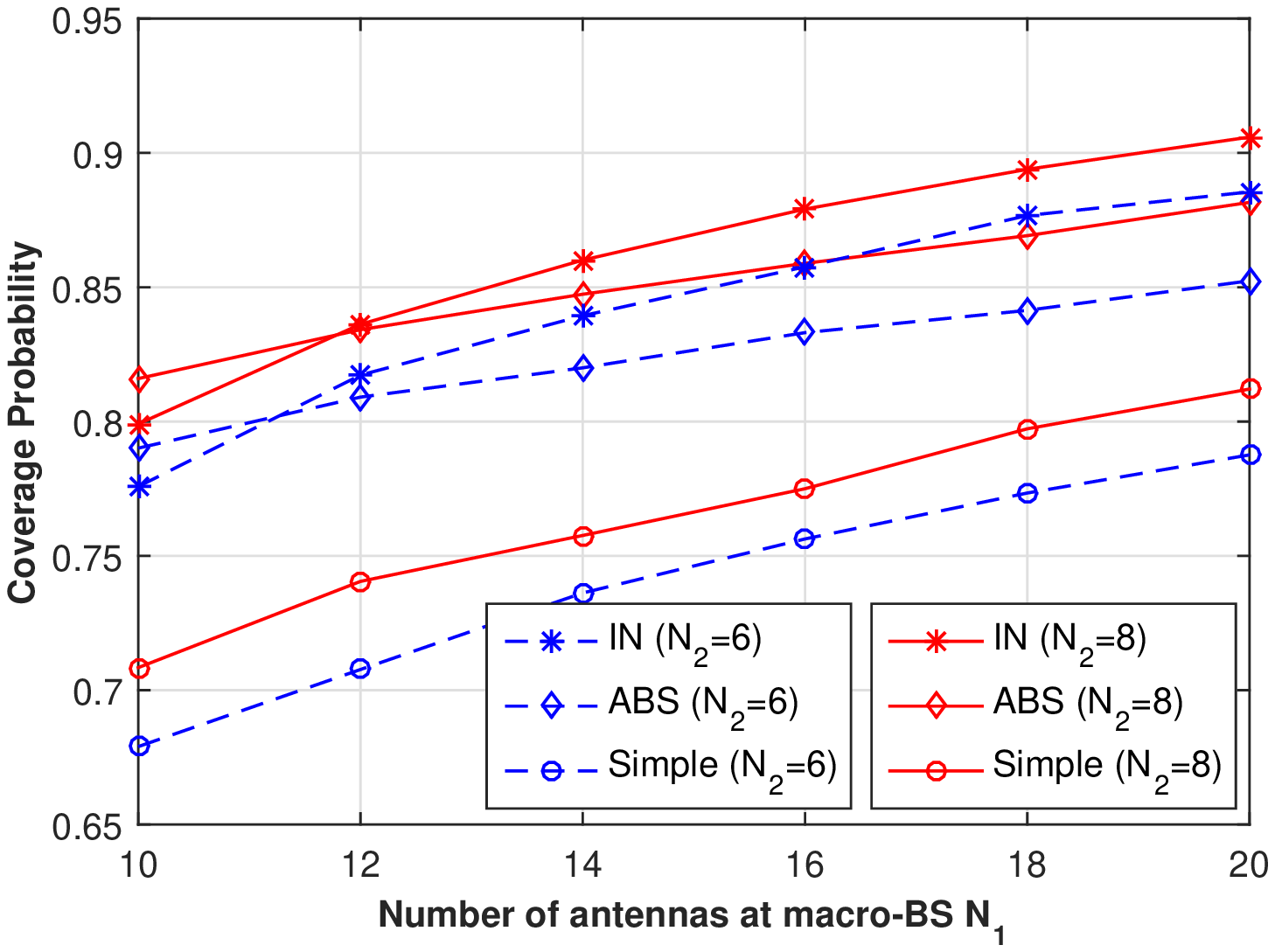}
\label{fig:compare_N1_68}
}
\caption{\small{Coverage probability versus number of antennas at each macro-BS. $\alpha_1=\alpha_2=4.5$, $T_1=T_2=6$ $\frac{P_{1}}{P_{2}}=15$ dB, $\lambda_{1}=0.0005$ nodes/m$^{2}$, and $\lambda_{2}=0.001$ nodes/m$^{2}$.}}\label{fig:compare_N1}
\vspace{-5mm}
\end{figure}

\begin{figure}[t]
\centering
\subfigure[$N_2=2$ and $N_2=4$]{
\includegraphics[height=0.28\columnwidth,width=0.38\columnwidth]
{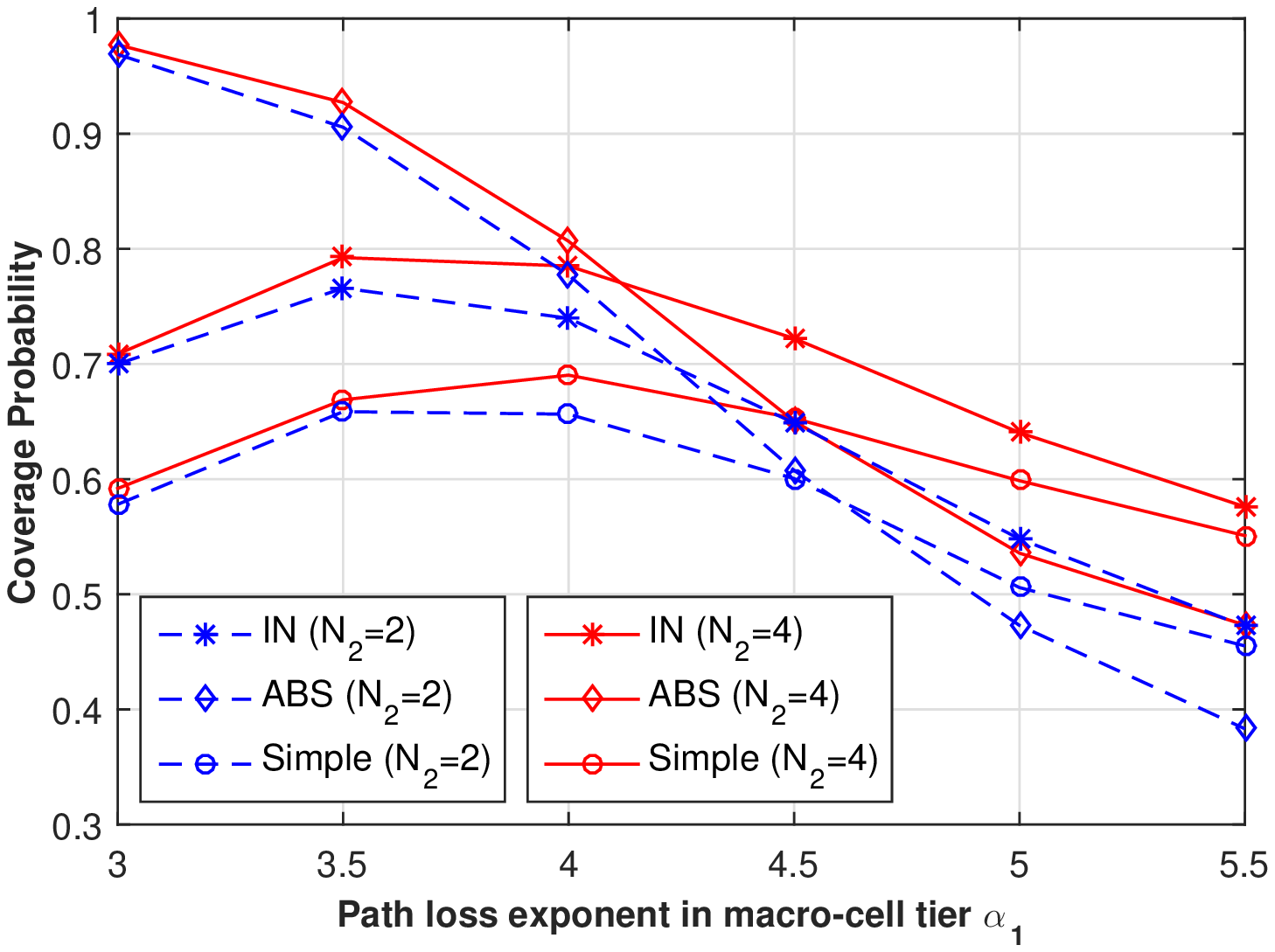}
\label{fig:compare_alpha1_24}
}
\subfigure[$N_2=6$ and $N_2=8$]{
\includegraphics[height=0.28\columnwidth,width=0.38\columnwidth]
{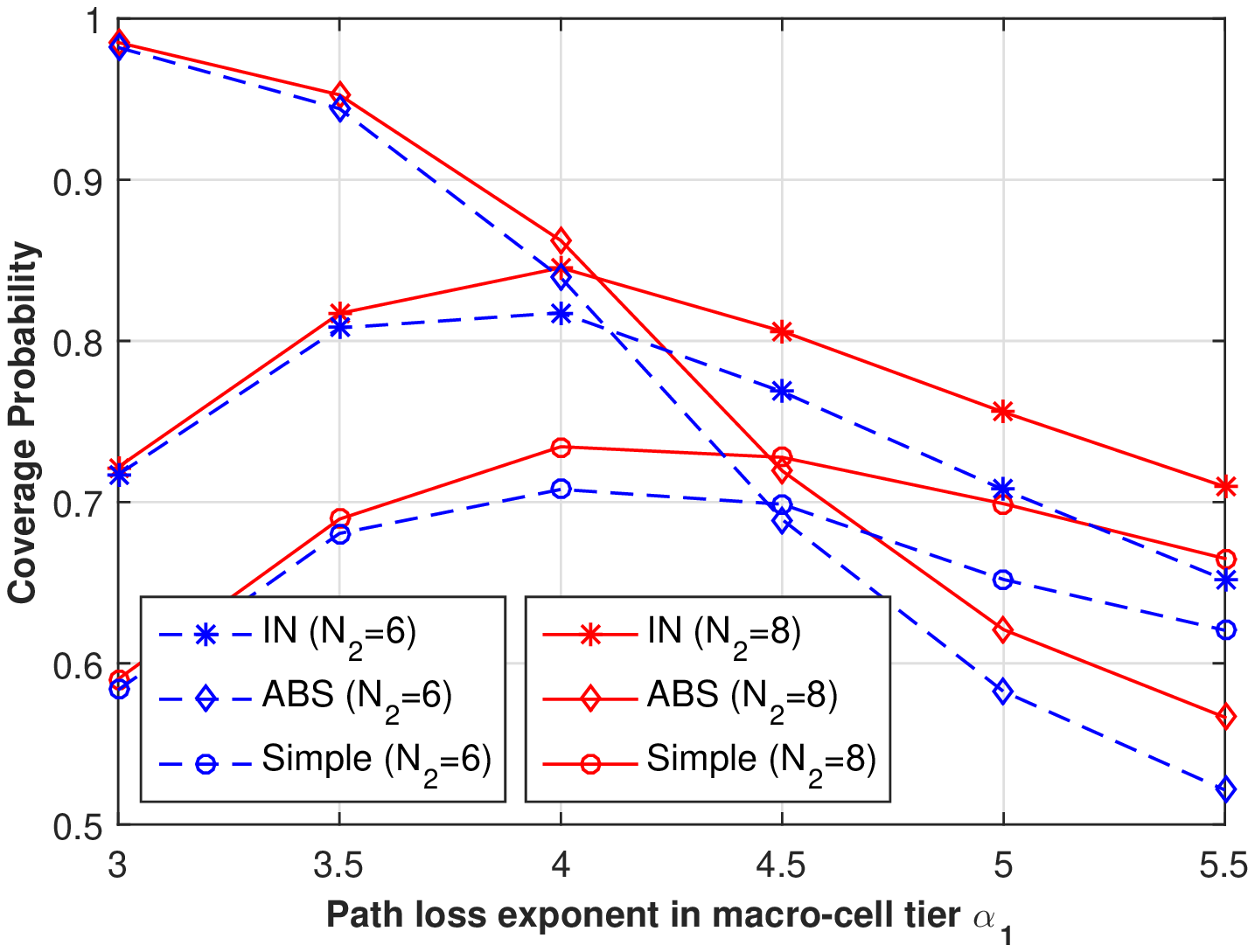}
\label{fig:compare_alpha1_68}
}
\caption{\small{Coverage probability  versus path loss exponent in macro-cell tier.  $N_{1}=16$, $\alpha_2=4$, $T_1=T_2=6$, $\frac{P_{1}}{P_{2}}=15$ dB, $\lambda_{1}=0.0005$ nodes/m$^{2}$, and $\lambda_{2}=0.001$ nodes/m$^{2}$.}}\label{fig:compare_alpha1}
\vspace{-5mm}
\end{figure}



\begin{figure}[t] \centering
\includegraphics[width=0.4\columnwidth]{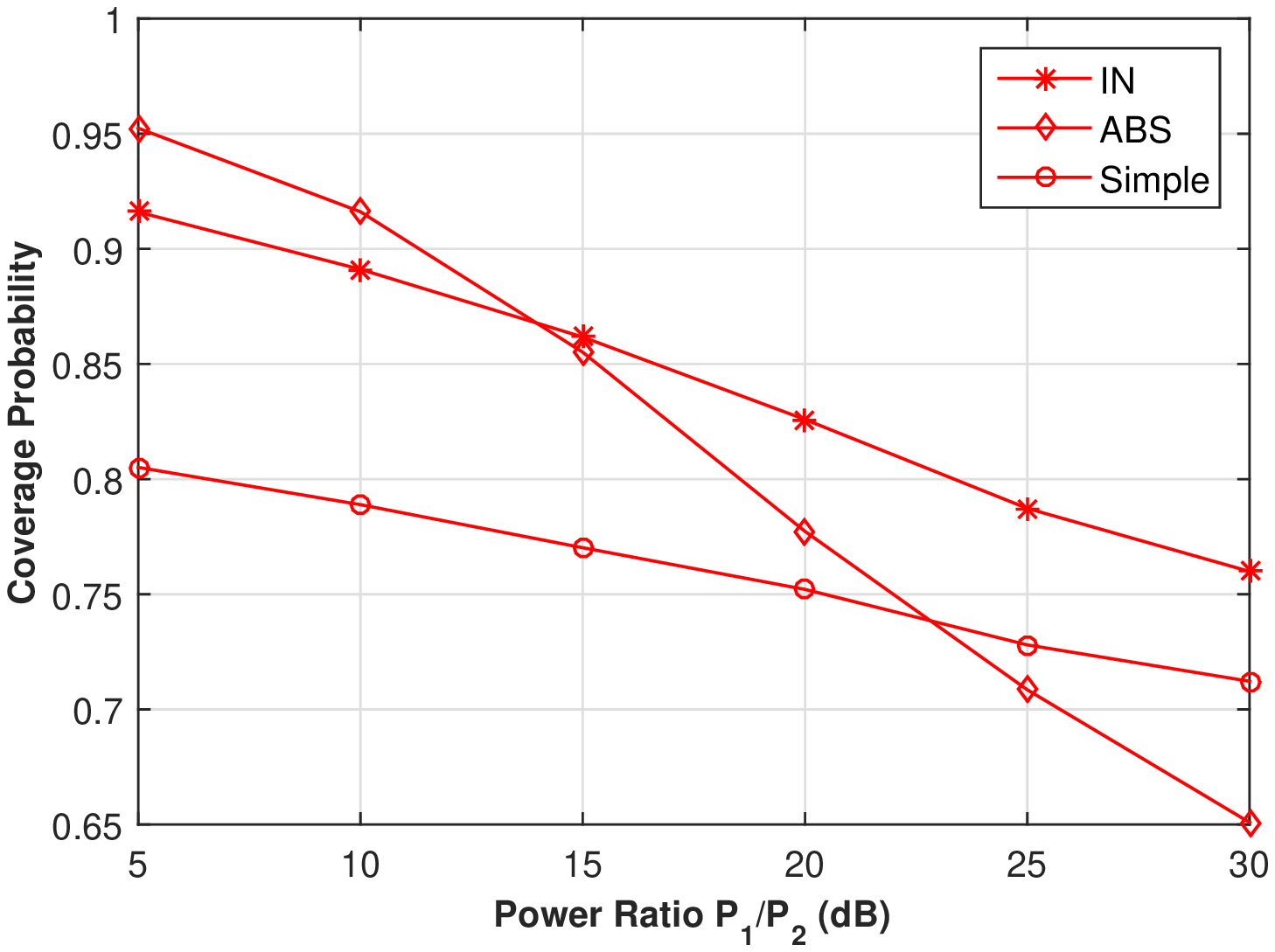}
\caption{\small{Coverage probability  versus  power ratio.  $N_{1}=16$, $N_2=8$, $T_1=T_2=6$, $\alpha_1=\alpha_2=4.5$, $\lambda_{1}=0.0005$ nodes/m$^{2}$, $\lambda_2=0.002$ nodes/m$^{2}$ and $\beta=10$ dB.}}\label{fig:CPvsP_ratio}
\vspace{-5mm}
\end{figure}

\section{Conclusions}
In this paper, we proposed a user-centric IN scheme in downlink two-tier multi-antenna HetNets.  Using tools from stochastic geometry, we first obtained a tractable expression of the coverage probability.
Then, we analyzed the asymptotic  coverage/outage probability  in the  low and high SIR threshold  regimes.
The  analytical results  indicate that the maximum IN DoF and the IN thresholds affect the asymptotic coverage/outage probability in dramatically different ways.
Moreover, we characterized the optimal maximum IN DoF which optimizes the  coverage/outage  probability. The optimization results reveal that the IN scheme can linearly improve the  outage probability  in the low SIR threshold regime, but cannot improve the  coverage probability in the high SIR threshold regime. Finally, numerical results showed that  the user-centric IN scheme can achieve good gains  in coverage/outage probability  over existing schemes.

\appendix

\subsection{Proof of Lemma \ref{lem:num_req_pmf}}\label{proof:pmf_Lbar}
According to Slivnyak's theorem \cite{haenggi09}, we focus on a macro-BS located at the origin, referred to as macro-BS $0$. Note that both scheduled macro and pico users may send IN requests to macro-BS $0$. We first characterize the probability that a scheduled macro-user sends an IN request to macro-BS $0$. Denote $R_{1i}$ as the distance between macro-BS $0$ and a randomly selected (according to the uniform distribution) scheduled macro-user, referred to  as scheduled macro-user $i$.
Assume that the scheduled macro-users form a homogeneous PPP with density $\lambda_{1}$. Conditioned on $R_{1i}=r$, scheduled macro-user $i$ sends an IN request to macro-BS $0$ with probability 
\small{\begin{align}\label{eq:prob_macro_send}
p_{1i,R_{1i}}(r,T_{1})&={\rm Pr}\left(T_{1}^{-\frac{1}{\alpha_{1}}}r<Y_{1}<r\right)=\int_{T_{1}^{-\frac{1}{\alpha_{1}}}r}^{r}f_{Y_{1}}(y){\rm d}y
\end{align}}
\normalsize{
where $f_{Y_{1}}(y)$ is the  p.d.f. of $Y_{1}$ given by (\ref{eq:pdfY1}) \cite[Lemma 4]{SinghTWC13}. Then, the scheduled macro-user density at distance $r$ away from macro-BS $0$ is $p_{1i,R_{1i}}(r,T_{1})\lambda_{1}$. This indicates that the scheduled macro-users at distance $r$ away from macro-BS $0$ which send IN requests to macro-BS $0$ form an inhomogeneous PPP with density $p_{1i,R_{1i}}(r,T_{1})\lambda_{1}$.}
Next, we characterize the probability that a scheduled pico-user sends an IN request to macro-BS $0$. Denote $R_{2i}$ as the distance between macro-BS $0$ and a randomly selected (according to the uniform distribution) scheduled pico-user, referred to as scheduled pico-user $i$. Similarly, we assume that the scheduled pico-users form a homogeneous PPP with density $\lambda_{2}$, and it is independent of the PPP formed by the scheduled macro-users. Then, we can show that the scheduled pico-users at distance $r$ away from macro-BS $0$ which send IN requests to macro-BS $0$ form an inhomogeneous PPP with density $p_{2i,R_{2i}}(r,T_{2})\lambda_{2}$, where
\small{\begin{align}\label{eq:prob_pico_send}
p_{2i,R_{2i}}(r,T_{2})&={\rm Pr}\left(\left(\frac{P_{2}}{P_{1}T_{2}}\right)^{\frac{1}{\alpha_{2}}}r^{\frac{\alpha_{1}}{\alpha_{2}}}<Y_{2}<\left(\frac{P_{2}}{P_{1}}\right)^{\frac{1}{\alpha_{2}}}r^{\frac{\alpha_{1}}{\alpha_{2}}}\right)=\int_{\left(\frac{P_{2}}{P_{1}T_{2}}\right)^{\frac{1}{\alpha_{2}}}r^{\frac{\alpha_{1}}{\alpha_{2}}}}^{\left(\frac{P_{2}}{P_{1}}\right)^{\frac{1}{\alpha_{2}}}r^{\frac{\alpha_{1}}{\alpha_{2}}}}f_{Y_{2}}(y){\rm d}y.
\end{align}}
\normalsize{Note that $f_{Y_{2}}(y)$ is the p.d.f. of $Y_{2}$ given by (\ref{eq:pdfY2}) \cite[Lemma 4]{SinghTWC13}}. 

By the superposition property of PPPs \cite{haenggi09}, the scheduled macro-users and the scheduled pico-users at distance $r$ away from macro-BS $0$ which send IN requests to macro-BS $0$, i.e., the potential IN users of macro-BS $0$, still form an inhomogeneous PPP with density $p_{1i,R_{1i}}(r,T_{1})\lambda_{1}+p_{2i,R_{2i}}(r,T_{2})\lambda_{2}$. Therefore, the number of the potential IN users of macro-BS $0$ is Poisson distributed with  parameter (mean) $\bar{L}(T_1,T_2)=2\pi\int_{0}^{\infty}r \left(p_{1i,R_{1i}}(r,T_{1})\lambda_{1}+p_{2i,R_{2i}}(r,T_{2})\lambda_{2}\right){\rm d}r=\bar{L}_{1}(T_1)+\bar{L}_{2}(T_2)$.


\subsection{Proof of Lemma \ref{lem:INprob}}\label{proof:lem_INprob}

Let   $p_{c}\left(U,T_{1},T_{2}, K\right)$ denote the probability that an arbitrary potential IN macro-BS of $u_0$  selects $u_0$  for IN when it has  $K$ potential IN users besides  $u_0$. If  $K+1\leq U$,  $p_{c}\left(U,T_{1},T_{2}, K\right)=1$;  if  $K+1>U$, $p_{c}\left(U,T_{1},T_{2}, K\right)=\frac{U}{K+1}$ , as the selection is according to the uniform distribution. Thus, for given K, we have $p_{c}\left(U,T_{1},T_{2}, K\right)=\min\left\{\frac{U}{K+1},1\right\}$. Averaging over $K$, we have $p_{c}\left(U,T_{1},T_{2}\right)={\rm E}\left[\min\left\{\frac{U}{K+1},1\right\}\right]$.
As shown in \cite{li14IC}, each scheduled user will send the IN request based on its own distances to each of its potential IN macro-BSs and its serving BS, which are independent of the other scheduled users. Thus, given that $u_0$ has sent the request to the potential IN macro-BS, $K$ follows the same distribution as  $K_0$. 
Therefore, we have
\small{\begin{align}\label{eq:INprob_cal}
&p_{c}\left(U,T_{1},T_{2}\right)={\rm E}\left[\min\left\{\frac{U}{K_{0}+1},1\right\}\right]=\sum_{k=0}^{U-1}{\rm Pr}\left(K_{0}=k\right)+\sum_{k=U}^{\infty}\frac{U}{k+1}{\rm Pr}\left(K_{0}=k\right)\;.
\end{align}}\normalsize{
Substituting (\ref{eq:K_pmf}) into (\ref{eq:INprob_cal}), we have the final result.}

\subsection{Proof of Theorem \ref{thm:overall_CP}}\label{proof:thm1}
Let $R_{j,1C}$ and $R_{j,1O}$ denote the minimum and maximum possible distances between $u_{0}\in\mathcal{U}_{j}$ and its nearest and furthest macro-interferers (among $u_{0}$'s potential IN macro-BSs which do not select $u_{0}$ for IN), respectively. Let $R_{j,2}$ denote the minimum possible distance between $u_{0}\in \mathcal{U}_{j}$ and its nearest pico-interferer. The relationships between $R_{j,1C}$, $R_{j,1O}$, $R_{j,2}$, and $Y_{j}$, respectively, are shown in Table \ref{tab:para_B2larger}. Based on (\ref{eq:SIRj0}) and conditioned on $Y_{j}=y$, we have
\small{\begin{align}
&{\rm Pr}\left({\rm SIR}_{j,0}>\beta|u_{0}\in\mathcal{U}_{j},Y_{j}=y\right)\notag\\
=&{\rm Pr}\left(\left|\mathbf{h}_{j,00}^{\dagger}\mathbf{f}_{j,0}\right|^{2}>\beta y^{\alpha_{j}}\left(\frac{P_{1}}{P_{j}}I_{j,1C}+\frac{P_{1}}{P_{j}}I_{j,1O}+\frac{P_{2}}{P_{j}}I_{j,2}\right)\right)\notag\\
\eqla& {\rm E}_{I_{j,1C},I_{j,1O},I_{j,2}}\bigg[\exp\left(-\beta y^{\alpha_{j}}\left(\frac{P_{1}}{P_{j}}I_{j,1C}+\frac{P_{1}}{P_{j}}I_{j,1O}+\frac{P_{2}}{P_{j}}I_{j,2}\right)\right)\sum_{n=0}^{M_{j}-1}\frac{\left(\beta y^{\alpha_{j}}\right)^{n}}{n!}\big(\frac{P_{1}}{P_{j}}I_{j,1C}+\frac{P_{1}}{P_{j}}I_{j,1O} +\frac{P_{2}}{P_{j}}I_{j,2}\big)^{n}\bigg] \notag\\
 \eqlb & {\rm E}_{I_{j,1C},I_{j,1O},I_{j,2}}\bigg[\exp\left(-\beta y^{\alpha_{j}}\left(\frac{P_{1}}{P_{j}}I_{j,1C}+\frac{P_{1}}{P_{j}}I_{j,1O}+\frac{P_{2}}{P_{j}}I_{j,2}\right)\right)\sum_{n=0}^{M_{j}-1}\frac{\left(\beta y^{\alpha_{j}}\right)^{n}}{n!}\sum_{\left(n_{a}\right)_{a=1}^{3}\in\mathcal{N}_{n}}\binom{n}{n_{1},n_{2},n_{3}}\notag\\
&\hspace{2.5cm} \times\left(\frac{P_{1}}{P_{j}}I_{j,1C}\right)^{n_{1}}\left(\frac{P_{1}}{P_{j}}I_{j,1O}\right)^{n_{2}}\left(\frac{P_{2}}{P_{j}}I_{j,2}\right)^{n_{3}}\bigg]\notag\\
=&\sum_{n=0}^{M_{j}-1}\frac{\left(\beta y^{\alpha_{j}}\right)^{n}}{n!}\sum_{\left(n_{a}\right)_{a=1}^{3}\in\mathcal{N}_{n}}\binom{n}{n_{1},n_{2},n_{3}}\left(\frac{P_{1}}{P_{j}}\right)^{n_{1}+n_{2}}\left(\frac{P_{2}}{P_{j}}\right)^{n_{3}} {\rm E}_{I_{j,1C}}\left[I_{j,1C}^{n_{1}}\exp\left(-\beta y^{\alpha_{j}}\frac{P_{1}}{P_{j}}I_{j,1C}\right)\right]\notag\\
&\hspace{0mm}\times{\rm E}_{I_{j,1O}}\left[I_{j,1O}^{n_{2}}\exp\left(-\beta y^{\alpha_{j}}\frac{P_{1}}{P_{j}}I_{j,1O}\right)\right] {\rm E}_{I_{j,2}}\left[I_{j,2}^{n_{3}}\exp\left(-\beta y^{\alpha_{j}}\frac{P_{2}}{P_{j}}I_{j,2}\right)\right]\notag\\
=&\sum_{n=0}^{M_{j}-1}\frac{\left(-\beta y^{\alpha_{j}}\right)^{n}}{n!}\sum_{\left(n_{a}\right)_{a=1}^{3}\in\mathcal{N}_{n}}\binom{n}{n_{1},n_{2},n_{3}}\left(\frac{P_{1}}{P_{j}}\right)^{n_{1}+n_{2}}\left(\frac{P_{2}}{P_{j}}\right)^{n_{3}} \mathcal{L}^{(n_{1})}_{I_{j,1C}}\left(s\right)|_{s=\beta y^{\alpha_{j}}\frac{P_{1}}{P_{j}}} \mathcal{L}^{(n_{2})}_{I_{j,1O}}\left(s\right)|_{s=\beta y^{\alpha_{j}}\frac{P_{1}}{P_{j}}}\notag\\
&\times  \mathcal{L}_{I_{j,2}}^{(n_{3})}(s)|_{s=\beta y^{\alpha_{j}}\frac{P_{2}}{P_{j}}}
\end{align}}\normalsize{
where  (a) is due to $\left|\mathbf{h}_{j,00}^{\dagger}\mathbf{f}_{j,0}\right|^{2}\dis{\rm Gamma}\left(M_{j},1\right)$, (b) is due to Multinomial Theorem, and  $\mathcal{L}^{(n)}_{I}(s,r)\triangleq {\rm E}_I[(-I)^n\exp(-sI)]$  denotes the $n$th-order derivative of the Laplace transform  of random variable $I$, i.e., $\mathcal{L}_{I}(s)\triangleq {\rm E}_I[\exp(-sI)]$.}

\begin{table}[t]
\caption{Parameter values}\label{tab:para_B2larger}
\begin{center}
\vspace{-6mm}
\begin{scriptsize}
\begin{tabular}{|c|c|c|c|c|c|}
\hline
$j$&$R_{j,1C}$&$R_{j,1O}$&$R_{j,2}$\\
\hline
$1$ &$Y_{1}$&$T_{1}^{\frac{1}{\alpha_{1}}}Y_{1}$ &$\left(\frac{P_{2}}{P_{1}}\right)^{\frac{1}{\alpha_{2}}}Y_{1}^{\frac{\alpha_{1}}{\alpha_{2}}}$\\
\hline
$2$&$\left(\frac{P_{1}}{P_{2}}\right)^{\frac{1}{\alpha_{1}}}Y_{2}^{\frac{\alpha_{2}}{\alpha_{1}}}$ & $\left(\frac{P_{1}}{P_{2}}T_{2}\right)^{\frac{1}{\alpha_{1}}}Y_{2}^{\frac{\alpha_{2}}{\alpha_{1}}}$&$Y_{2}$\\
\hline
\end{tabular}
\end{scriptsize}
\end{center}
\vspace{-6mm}
\end{table}

Now, we calculate  $\mathcal{L}_{I}(s)$ and $\mathcal{L}^{(n)}_{I}(s)$, respectively.  First, $\mathcal{L}_{I_{j,1C}}(s)$ can be calculated as follows:
\small{\begin{align}\label{eq:LT_IC}
&\mathcal{L}_{I_{j,1C}}(s)={\rm E}_{\Phi_{j,1C},\{\mathbf{g}_{1,\ell}\}}\left[\exp\left(-s\sum_{\ell\in\Phi_{j,1C}}D_{1,\ell0}^{-\alpha_{1}}\mathbf{g}_{1,\ell}\right)\right]
\eqlc{\rm E}_{\Phi_{j,1C}}\left[\prod_{\ell\in\Phi_{j,1C}}{\rm E}_{\{\mathbf{g}_{1,\ell}\}}\left[\exp\left(-sD_{1,\ell0}^{-\alpha_{1}}\mathbf{g}_{1,\ell}\right)\right]\right]\notag
\end{align}
\begin{align}
&\eqld {\rm E}_{\Phi_{j,1C}}\left[\prod_{\ell\in\Phi_{j,1C}}\frac{1}{1+sD_{1,\ell0}^{-\alpha_{1}}}\right]\eqle\exp\left(-2\pi p_{\bar{c}}\left(U,T_{1},T_{2}\right)\lambda_{1}\int_{r_{j,1C}}^{r_{j,1O}}\left(1-\frac{1}{1+\frac{s}{r^{\alpha_{1}}}}\right)r{\rm d}r\right) \triangleq \mathcal{L}_{I_{j,1C}}(U,s,r_{j,1C},r_{j,1O}) 
\end{align}}\normalsize{
where $\mathbf{g}_{1,\ell} \triangleq \left|\mathbf{h}_{1,\ell0}^{\dagger}\mathbf{f}_{1,\ell}\right|^{2}$, (c) is obtained by noting that $\mathbf{g}_{1,\ell}$ ($\ell\in\Phi_{j,1C}$) are mutually independent, (d) is due to $\left|\mathbf{h}_{1,\ell0}^{\dagger}\mathbf{f}_{1,\ell}\right|^{2}\dis{\rm Gamma}(1,1)$ (i.e., ${\rm Exp}(1)$),  and (e) is obtained by using the probability generating functional of a PPP \cite{haenggi09}. Further, by first letting $s^{-\frac{1}{\alpha_{1}}}r=t$  (i.e., $\frac{s}{r^{\alpha_1}}=t^{-\alpha_1}$) and then $\frac{1}{1+t^{-\alpha_{1}}}=w$, we  have $\int_{r_{j,1C}}^{r_{j,1O}}\left(1-\frac{1}{1+\frac{s}{r^{\alpha_{1}}}}\right)r{\rm d}r=\frac{s^{\frac{2}{\alpha_1}}}{\alpha_1}\int_{1/(1+sr_{j,1C}^{-\alpha_1})}^{1/(1+sr_{j,1O}^{-\alpha_1})}w^{\frac{2}{\alpha_1}-1}(1-w)^{-\frac{2}{\alpha_1}}{\rm d}w$. By the definition of $B'(a,b,z)$, we can obtain}
\small{\begin{align}
&\mathcal{L}_{I_{j,1C}}(U,s,r_{j,1C},r_{j,1O})\notag\\
&=\exp\Bigg(-\left(B^{'}\left(\frac{2}{\alpha_{1}},1-\frac{2}{\alpha_{1}},\frac{1}{1+sr_{j,1C}^{-\alpha_{1}}}\right)-B^{'}\left(\frac{2}{\alpha_{1}},1-\frac{2}{\alpha_{1}},\frac{1}{1+sr_{j,1O}^{-\alpha_{1}}}\right)\right)\frac{2\pi}{\alpha_{1}}p_{\bar{c}}\left(U,T_{1},T_{2}\right)\lambda_{1}s^{\frac{2}{\alpha_{1}}}\Bigg)\label{eq:LT_1in}.
\end{align}}

Next, based on (\ref{eq:LT_IC}) and utilizing Fa${\rm \grave{a}}$ di Bruno's formula \cite{johnson02}, $\mathcal{L}_{I_{j,1C}}^{(n_{1})}(s)$ can be calculated as follows:
\small{\begin{align}
\mathcal{L}_{I_{j,1C}}^{(n_{1})}(s)=& \mathcal{L}_{I_{j,1C}}(s)\sum_{(m_{a})_{a=1}^{n_1}\in\mathcal{M}_{n_{1}}}\frac{n_{1}!}{\prod_{a=1}^{n_{1}}m_{a}!}\prod_{a=1}^{n_{1}}\left(\frac{2\pi p_{\bar{c}}\left(U,T_{1},T_{2}\right)\lambda_{1}}{a!}\int_{r_{j,1C}}^{r_{j,1O}}\frac{{\rm d}^{a}}{{\rm d}s^{a}}\left(\frac{1}{1+\frac{s}{r^{\alpha_{1}}}}\right)r{\rm d}r\right)^{m_{a}}\notag\\
 \eqlf &(-1)^{n_{1}}\mathcal{L}_{I_{j,1C}}(U,s,r_{j,1C},r_{j,1O})\sum_{(m_{a})_{a=1}^{n_1}\in\mathcal{M}_{n_{1}}}\frac{n_{1}!}{\prod_{a=1}^{n_{1}}m_{a}!}\prod_{a=1}^{n_{1}}\left(2\pi p_{\bar{c}}\left(U,T_{1},T_{2}\right)\lambda_{1}\int_{r_{j,1C}}^{r_{j,1O}}\frac{r^{1-a\alpha_{1}}}{\left(1+\frac{s}{r^{\alpha_{1}}}\right)^{a+1}}{\rm d}r\right)^{m_{a}}\notag\\
&\hspace{-3mm} \triangleq \mathcal{L}_{I_{j,1C}}^{(n_{1})}(U,s,r_{j,1C},r_{j,1O}) 
\end{align}}\normalsize{
 where (f) is due to $\frac{{\rm d}^{a}}{{\rm d}s^{a}}\left(\frac{1}{1+\frac{s}{r^{\alpha_{1}}}}\right)=  \frac{(-1)^a (a!)r^{-a\alpha_1}}{\left(1+\frac{s}{r^{\alpha_1}}\right)^{a+1}}$ and $\prod_{a=1}^{n_1}(-1)^{am_a}=(-1)^{\sum_{a=1}^{n_1}am_a}=(-1)^{n_1}$.    Similarly, by first letting $s^{-\frac{1}{\alpha_{1}}}r=t$  and then $\frac{1}{1+t^{-\alpha_{1}}}=w$, we have $\int_{r_{j,1C}}^{r_{j,1O}}\frac{r^{1-a\alpha_{1}}}{\left(1+\frac{s}{r^{\alpha_{1}}}\right)^{a+1}}{\rm d}r=\frac{s^{\frac{2}{\alpha_1}-a}}{\alpha_1}\int_{1/(1+sr_{j,1C}^{-\alpha_1})}^{1/(1+sr_{j,1O}^{-\alpha_1})}$ $w^{\frac{2}{\alpha_1}}(1-w)^{-\frac{2}{\alpha_1}+a-1}{\rm d}w$. Thus, we can calculate $\mathcal{L}_{I_{j,1C}}^{(n_{1})}(U,s,r_{j,1C},r_{j,1O})$.  Let $\tilde{\mathcal{L}}_{I_{j,1C}}^{(n_{1})}(s)\triangleq \mathcal{L}_{I_{j,1C}}^{(n_{1})}(s)/(-\frac{1}{s})^{n_{1}}$ $=\mathcal{L}_{I_{j,1C}}^{(n_{1})}(U,s,r_{j,1C},r_{j,1O})/(-\frac{1}{s})^{n_{1}}\triangleq \tilde{\mathcal{L}}_{I_{j,1C}}^{(n_{1})}(U,s,r_{j,1C},r_{j,1O})$.  Similarly, we can calculate  $\mathcal{L}_{I_{j,1O}}\left(s\right)$, $\mathcal{L}^{(n_{2})}_{I_{j,1O}}\left(s\right)$, $\mathcal{L}_{I_{j,2}}(s)$ and $\mathcal{L}_{I_{j,2}}^{(n_{3})}(s)$. Finally, removing the conditions on $Y_{j}=y$ and after some algebraic manipulations, we can obtain the final result.}

\subsection{Proof of Theorem \ref{them:CPj_lowbeta}}\label{proof:thm2}
Conditioned on $Y_{j}=y$, we have
\small{\begin{align}\label{eq:OPj_condi}
&1-{\rm Pr}\left({\rm SIR}_{j,0}>\beta|u_{0}\in\mathcal{U}_{j},Y_{j}=y\right)\notag\\
=&{\rm Pr}\left(\left|\mathbf{h}_{j,00}^{\dagger}\mathbf{f}_{j,0}\right|^{2}\le\beta y^{\alpha_{j}}\left(\frac{P_{1}}{P_{j}}I_{j,1C}+\frac{P_{1}}{P_{j}}I_{j,1O}+\frac{P_{2}}{P_{j}}I_{j,2}\right)\right)\notag\\
\eqla&\exp\left(-\beta y^{\alpha_{j}}\left(\frac{P_{1}}{P_{j}}I_{j,1C}+\frac{P_{1}}{P_{j}}I_{j,1O}+\frac{P_{2}}{P_{j}}I_{j,2}\right)\right)\sum_{n=M_{j}}^{\infty}\frac{\left(\beta y^{\alpha_{j}}\right)^{n}}{n!}\left(\frac{P_{1}}{P_{j}}I_{j,1C}+\frac{P_{1}}{P_{j}}I_{j,1O}+\frac{P_{2}}{P_{j}}I_{j,2}\right)^{n}\notag\\
\eqlb&\sum_{n=M_{j}}^{\infty}\frac{\left(-\beta y^{\alpha_{j}}\right)^{n}}{n!}\sum_{\left(n_{a}\right)_{a=1}^{3}\in\mathcal{N}_{n}}\binom{n}{n_{1},n_{2},n_{3}}\left(\frac{P_{1}}{P_{j}}\right)^{n_{1}+n_{2}}\left(\frac{P_2}{P_j}\right)^{n_3}  \mathcal{L}^{(n_{1})}_{I_{j,1C}}\left(s\right)|_{s=\beta y^{\alpha_{j}}\frac{P_{1}}{P_{j}}} \mathcal{L}^{(n_{2})}_{I_{j,1O}}\left(s\right)|_{s=\beta y^{\alpha_{j}}\frac{P_{1}}{P_{j}}}\notag\\
&\times  \mathcal{L}_{I_{j,2}}^{(n_{3})}(s)|_{s=\beta y^{\alpha_{j}}\frac{P_{2}}{P_{j}}} \nonumber
\end{align}
\begin{align}
 \eqlc & \sum_{n=M_{j}}^{\infty}\underbrace{\frac{1}{n!}\sum_{(n_{a})_{a=1}^{3}\in\mathcal{N}_{n}}\binom{n}{n_{1},n_{2},n_{3}}\mathcal{\tilde{L}}^{(n_{1})}_{I_{j,1C}}\left(s\right)|_{s=\beta y^{\alpha_{j}}\frac{P_{1}}{P_{j}}}\mathcal{\tilde{L}}^{(n_{2})}_{I_{j,1O}}\left(s\right)|_{s=\beta y^{\alpha_{h}}\frac{P_{1}}{P_{j}}}\mathcal{\tilde{L}}^{(n_{3})}_{I_{j,2}}\left(s\right)|_{s=\beta y^{\alpha_{j}}\frac{P_{2}}{P_{j}}}}_{\triangleq \mathcal{T}_{j,Y_{j}}\left(n,y,U,T_{1},T_{2},\beta\right)} 
\end{align}}\normalsize{
 where (a) is due to $\left|\mathbf{h}_{j,00}^{\dagger}\mathbf{f}_{j,0}\right|^{2}\dis{\rm Gamma}\left(M_{j},1\right)$, (b) is due to Multinomial Theorem, and (c) is due to similar  calculations in Appendix \ref{proof:thm1}. Removing the condition on $Y_{j}=y$, we have}
\small{\begin{align}\label{eq:OP_j_sumtoinf}
&1-\mathcal{S}_{j}\left(U,T_{1},T_{2},\beta\right)=\int_{0}^{\infty}\sum_{n=M_{j}}^{\infty}\mathcal{T}_{j,Y_{j}}\left(n,y,U,T_{1},T_{2},\beta\right)f_{Y_{j}}(y){\rm d}y.
\end{align}}
\normalsize{Now, we calculate $\lim_{\beta\to0}\int_{0}^{\infty}\sum_{n=M_{j}}^{\infty}\mathcal{T}_{j,Y_{j}}\left(n,y,U,T_{1},T_{2},\beta\right)f_{Y_{j}}(y){\rm d}y$, i.e., the asymptotic outage probability when $\beta\to0$.}
We note that $B^{'}(a,b,z)=\frac{(1-z)^{b}}{b}+o\left((1-z)^{b}\right)$ as $z\to1$.
Then, we have
\small{\begin{align}
B^{'}\left(\frac{2}{\alpha},1-\frac{2}{\alpha},\frac{1}{1+c\beta}\right)&= \frac{\left(c\beta\right)^{1-\frac{2}{\alpha}}}{1-\frac{2}{\alpha}}+o\left(\beta^{1-\frac{2}{\alpha}}\right)\;,\label{eq:betainc1}\\
B^{'}\left(1+\frac{2}{\alpha},a-\frac{2}{\alpha},\frac{1}{1+c\beta}\right)&= \frac{(c\beta)^{a-\frac{2}{\alpha}}}{a-\frac{2}{\alpha}}+o\left(\beta^{a-\frac{2}{\alpha}}\right)\;,\label{eq:betainc2}
\end{align}} \normalsize{where $c\in\mathbb{R}^{+}$. Based on these two asymptotic expressions,
we can obtain\footnote{$f(x)=o\left(g(x)\right)$ means $\lim_{x\to0}\frac{f(x)}{g(x)}=0$.}}
\small{\begin{align}
&\mathcal{\tilde{L}}^{(n_{1})}_{I_{j,1C}} \left(s\right)=\beta^{n_{1}}\sum_{(m_{a})_{a=1}^{n_{1}}\in\mathcal{M}_{n_{1}}}\frac{n_{1}!}{\prod_{a=1}^{n_{1}}m_{a}!}\prod_{a=1}^{n_{1}}\Bigg(\frac{\frac{2\pi}{\alpha_{1}}p_{\bar{c}}\left(U,T_{1},T_{2}\right)\lambda_{1}}{a-\frac{2}{\alpha_{1}}}\left(1-\left(\frac{1}{T_{j}}\right)^{a-\frac{2}{\alpha_{1}}}\right) \left(\frac{P_{1}y^{\alpha_{j}}}{P_{j}}\right)^{\frac{2}{\alpha_{1}}} \Bigg)^{m_{a}}+o\left(\beta^{n_{1}}\right)\;,\label{eq:LT_1C_lowbeta}\\
&\mathcal{\tilde{L}}^{(n_{2})}_{I_{j,1O}} \left(s\right)=\beta^{n_{2}}\sum_{(p_{a})_{a=1}^{n_{2}}\in\mathcal{M}_{n_{2}}}\frac{n_{2}!}{\prod_{a=1}^{n_{2}}p_{a}!} \prod_{a=1}^{n_{2}}\left(\frac{\frac{2\pi}{\alpha_{1}}\lambda_{1}}{a-\frac{2}{\alpha_{1}}}\left(\frac{P_{1}}{P_{j}}\right)^{\frac{2}{\alpha_{1}}}y^{\frac{2\alpha_{j}}{\alpha_{1}}}\left(\frac{1}{T_{j}}\right)^{a-\frac{2}{\alpha_{1}}}\right)^{p_{a}}+o\left(\beta^{n_{2}}\right)\;,\label{eq:LT_1O_lowbeta}\\
&\mathcal{\tilde{L}}^{(n_{3})}_{I_{j,2}} \left(s\right)=\beta^{n_{3}}\sum_{(q_{a})_{a=1}^{n_{3}}\in\mathcal{M}_{n_{3}}}\frac{n_{3}!}{\prod_{a=1}^{n_{2}}q_{a}!} \prod_{a=1}^{n_{3}}\left(\frac{\frac{2\pi}{\alpha_{2}}\lambda_{2}}{a-\frac{2}{\alpha_{2}}}\left(\frac{P_{2}}{P_{j}}\right)^{\frac{2}{\alpha_{2}}}y^{\frac{2\alpha_{j}}{\alpha_{2}}}\right)^{q_{a}}+o\left(\beta^{n_{3}}\right)\;.\label{eq:LT_2_lowbeta}
\end{align}}\normalsize{
Moreover, utilizing dominated convergence theorem, we can show that}
\small{
\begin{align}\nonumber
&\lim_{\beta\to0}\int_{0}^{\infty}\sum_{n=M_{j}}^{\infty}\mathcal{T}_{j,Y_{j}}\left(n,y,U,T_{1},T_{2},\beta\right)f_{Y_{j}}(y){\rm d}y=\int_{0}^{\infty}\sum_{n=M_{j}}^{\infty}\lim_{\beta\to0}\mathcal{T}_{j,Y_{j}}\left(n,y,U,T_{1},T_{2},\beta\right)f_{Y_{j}}(y){\rm d}y\;.
\end{align}}\normalsize{
Hence, substituting (\ref{eq:LT_1C_lowbeta}), (\ref{eq:LT_1O_lowbeta}) and (\ref{eq:LT_2_lowbeta}) into (\ref{eq:OP_j_sumtoinf}), and after some algebraic manipulations, we obtain Results 1), 2) and 3) in \emph{Theorem \ref{them:CPj_lowbeta}}. To complete the proof, we now show that $b_{2}\left(U,T_{1},T_{2}\right)$ decreases with $U$. This can be proved by noting that i) $b_{2}\left(U,T_{1},T_{2}\right)$ is an increasing function of $p_{\bar{c}}\left(U,T_{1},T_{2}\right)$, and ii) $p_{\bar{c}}\left(U,T_{1},T_{2}\right)$ decreases with $U$ (which can be easily shown using (\ref{eq:INprob_cal})).}

\subsection{Proof of Lemma \ref{lem:opt_U}}\label{proof:lem_optU}

First, we characterize the maximum order gain. When $U\in\{0,1,\ldots,N_{1}-N_{2}\}$, we have $N_{1}-U\ge N_{2}$, implying $\min\{N_1-U,N_2\}=N_2$. When $U\in\{ N_{1}-N_{2}+1,\ldots,N_{1}-1\}$, we have $N_{1}-U< N_{2}$, implying $\min\{N_1-U,N_2\}=N_{1}-U<N_{2}$. 
Thus, we can show that the maximum order gain is $\max_{U\in \{0,1,\cdots, N_1-1\}}\min\{N_1-U,N_2\}=N_2$,  achieved at any $U\in\{0,1,\ldots,N_{1}-N_{2}\}$.    Next, we compare the coefficients of $\beta^{N_{2}}$ achieved at  different $U\in\{0,1,\ldots,N_{1}-N_{2}\}$. We consider two cases. i) When $U<N_{1}-N_{2}$, as $b_{2}\left(U,T_{1},T_{2}\right)$ decreases with $U$, the coefficients satisfy $ \mathcal A_2 b_{2}\left(N_{1}-N_{2}-1,T_{1},T_{2}\right)<\mathcal A_2 b_{2}\left(N_{1}-N_{2}-2,T_{1},T_{2}\right)<\ldots< \mathcal A_2 b_{2}\left(0,T_{1},T_{2}\right)$. ii) When $U=N_{1}-N_{2}$, the coefficient of $\beta^{N_{2}}$ is $ \mathcal A_1 b_{1}\left(N_{1}-N_{2},T_{1},T_{2}\right)+\mathcal A_2 b_{2}\left(N_{1}-N_{2},T_{1},T_{2}\right)$.
Therefore, we can complete the proof.

\subsection{Proof of Theorem \ref{them:CP_highbeta_bound}}\label{proof:CP_highbeta_uneql}

\subsubsection{Upper Bound}

Let $\mathcal{S}_{j,Y_{j}}\left(y,\beta,U,T_{1},T_{2}\right)\define {\rm Pr}\left({\rm SIR}_{j,0}>\beta|u_{0}\in\mathcal{U}_{j},Y_{j}=y\right)$ denote the conditional SIR coverage probability. Then,  from Theorem \ref{thm:overall_CP} (Appendix C), $\mathcal{S}_{j}\left(\beta,U,T_{1},T_{2}\right)$ can be written as
\small{\begin{align}\label{eq:CPj_Sj_pdfYj}
\mathcal{S}_{j}\left(\beta,U,T_{1},T_{2}\right)=&\int_{0}^{\infty}\mathcal{S}_{j,Y_{j}}\left(y,\beta,U,T_{1},T_{2}\right)f_{Y_{j}}(y){\rm d}y
\end{align}}
\normalsize{
where $\mathcal{S}_{j,Y_{j}}\left(y,\beta,U,T_{1},T_{2}\right)=\mu_{j}\left(\beta,U,T_{1},T_{2}\right)g_{j}\left(y,\beta,U,T_{1},T_{2}\right)$ and $f_{Y_{j}}(y)$ is the p.d.f. of $Y_{j}$ given in \emph{Lemma \ref{lem:num_req_pmf}}. Here, $g_{j}\left(y,\beta,U,T_{1},T_{2}\right)=\exp\left(-c_{1}(\beta)\beta^{\frac{2}{\alpha_{1}}}y^{\frac{2\alpha_{j}}{\alpha_{1}}}- c_{2}(\beta)\beta^{\frac{2}{\alpha_{2}}}y^{\frac{2\alpha_{j}}{\alpha_{2}}}\right)$ with $c_{1}(\beta)=\left(p_{\bar{c}}\left(U,T_{1},T_{2}\right)\left(B^{'}\left(\frac{2}{\alpha_{1}},1-\frac{2}{\alpha_{1}},\frac{1}{1+\beta}\right)-B^{'}\left(\frac{2}{\alpha_{1}},1-\frac{2}{\alpha_{1}},\frac{1}{1+\frac{\beta}{T_{j}}}\right)\right)+B^{'}\left(\frac{2}{\alpha_{1}},1-\frac{2}{\alpha_{1}},\frac{1}{1+\frac{\beta}{T_{j}}}\right)\right)$\\$\times\frac{2\pi\lambda_{1}}{\alpha_{1}}\left(\frac{P_{1}}{P_{j}}\right)^{\frac{2}{\alpha_{1}}}$ and $ c_{2}(\beta)=\frac{2\pi\lambda_{2}}{\alpha_{2}}\left(\frac{P_{2}}{P_{j}}\right)^{\frac{2}{\alpha_{2}}}B^{'}\left(\frac{2}{\alpha_{2}},1-\frac{2}{\alpha_{2}},\frac{1}{1+\beta}\right)$, and}
\small{\begin{align}\label{eq:CPj_orig}
\mu_{j}\left(\beta,U,T_{1},T_{2}\right)&=\sum_{n=0}^{M_{j}-1}\sum_{(n_{a})_{a=1}^{3}\in\mathcal{N}_{n}}\sum_{(m_{a})_{a=1}^{n_{1}}\in\mathcal{M}_{n_{1}}}\sum_{(p_{a})_{a=1}^{n_{2}}\in\mathcal{M}_{n_{2}}}\sum_{(q_{a})_{a=1}^{n_{3}}\in\mathcal{M}_{n_{3}}}c(\beta) y^{\frac{2\alpha_{j}}{\alpha_{1}}\left(\sum_{a=1}^{n_{1}}m_{a}+\sum_{a=1}^{n_{2}}p_{a}\right)+\frac{2\alpha_{j}}{\alpha_{2}}\sum_{a=1}^{n_{3}}q_{a}}\;,
\end{align}}\normalsize
with
\small{\begin{align}
c(\beta)=&\frac{1}{n!}\binom{n}{n_{1},n_{2},n_{3}}\frac{n_{1}!}{\prod_{a=1}^{n_{1}}m_{a}!}\frac{n_{2}!}{\prod_{a=1}^{n_{2}}p_{a}!}\frac{n_{3}!}{\prod_{a=1}^{n_{3}}q_{a}!}\prod_{a=1}^{n_{3}}\left(\frac{2\pi\lambda_{2}}{\alpha_{2}}\left(\frac{\beta P_{2}}{P_{j}}\right)^{\frac{2}{\alpha_{2}}}B^{'}\left(1+\frac{2}{\alpha_{2}},a-\frac{2}{\alpha_{2}},\frac{1}{1+\beta}\right)\right)^{q_{a}}\notag\\
&\hspace{0cm}\times\prod_{a=1}^{n_{1}}\left(\frac{2\pi\lambda_{1}}{\alpha_{1}}p_{\bar{c}}\left(U,T_{1},T_{2}\right)\left(\frac{\beta P_{1}}{P_{j}}\right)^{\frac{2}{\alpha_{1}}}\right)^{m_{a}}\prod_{a=1}^{n_{2}}\left(\frac{2\pi\lambda_{1}}{\alpha_{1}}\left(\frac{\beta P_{1}}{P_{j}}\right)^{\frac{2}{\alpha_{1}}}B^{'}\left(1+\frac{2}{\alpha_{1}},a-\frac{2}{\alpha_{1}},\frac{1}{1+\frac{\beta}{T_{j}}}\right)\right)^{p_{a}}\notag\\
&\hspace{0cm}\times\prod_{a=1}^{n_{1}}\left(B^{'}\left(1+\frac{2}{\alpha_{1}},a-\frac{2}{\alpha_{1}},\frac{1}{1+\beta}\right)-B^{'}\left(1+\frac{2}{\alpha_{1}},a-\frac{2}{\alpha_{1}},\frac{1}{1+\frac{\beta}{T_{j}}}\right)\right)^{m_{a}}.\nonumber
\end{align}}\normalsize
Let $\tilde{f}_{Y_{j}}(y)=\frac{2\pi\lambda_{j}}{\mathcal{A}_{j}}y\exp\left(-\pi \lambda_j y^{2}\right)$. Let $\mathcal{\tilde{S}}_{j,Y_{j}}\left(y,\beta,U,T_{1},T_{2}\right)\triangleq\mu_{j}\left(\beta,U,T_{1},T_{2}\right)\tilde{g}_{j}\left(y,\beta,U,T_{1},T_{2}\right)$ with $\tilde{g}_{j}\left(y,\beta,U,T_{1},T_{2}\right)=\exp\left(-c_{j}(\beta)\beta^{\frac{2}{\alpha_{j}}}y^{2}\right)$. Then, we have
\small{\begin{align}\label{eq:CPj_ub_1infty}
&\int_{0}^{\infty}\mathcal{S}_{j,Y_{j}}\left(y,\beta,U,T_{1},T_{2}\right)f_{Y_{j}}(y){\rm d}y
\stackrel{(a)}{<}\int_{0}^{\infty}\mathcal{\tilde{S}}_{j,Y_{j}}\left(y,\beta,U,T_{1},T_{2}\right)\tilde{f}_{Y_{j}}(y){\rm d}y\notag
\end{align}
\begin{align}
&\eqlb \frac{\pi\lambda_{j}}{\mathcal{A}_{j}}\sum_{n=0}^{M_{j}-1}\sum_{(n_{a})_{a=1}^{3}\in\mathcal{N}_{n}}\sum_{(m_{a})_{a=1}^{n_{1}}\in\mathcal{M}_{n_{1}}}\sum_{(p_{a})_{a=1}^{n_{2}}\in\mathcal{M}_{n_{2}}}\sum_{(q_{a})_{a=1}^{n_{3}}\in\mathcal{M}_{n_{3}}}
\Gamma\left(\frac{\alpha_{j}}{\alpha_{1}}\left(\sum_{a=1}^{n_{1}}m_{a}+\sum_{a=1}^{n_{2}}p_{a}\right)+\frac{\alpha_{j}}{\alpha_{2}}\sum_{a=1}^{n_{3}}q_{a}+1\right)c(\beta)\notag\\
&\hspace{5mm}\times\left(c_{j}(\beta)\beta^{\frac{2}{\alpha_{j}}}+\pi\lambda_j \right)^{-\frac{\alpha_{j}}{\alpha_{1}}\left(\sum_{a=1}^{n_{1}}m_{a}+\sum_{a=1}^{n_{2}}p_{a}\right)-\frac{\alpha_{j}}{\alpha_{2}}\sum_{a=1}^{n_{3}}q_{a}-1}\;,
\end{align}}
\normalsize{where (a) is due to $g_j\left(y,\beta,U,T_{1},T_{2}\right)<\tilde g_j\left(y,\beta,U,T_{1},T_{2}\right)$ and $f_{Y_{j}}(y)<\tilde f_{Y_{j}}(y)$,  and (b) is due to $\int_0^{\infty}u^a\exp(-bu){\rm d}u=b^{-a-1}\Gamma(a+1)$. To calculate the order of $\beta$ for \eqref{eq:CPj_ub_1infty} as $\beta\to \infty$, we first calculate the orders of $\beta$ for $c(\beta)$ and $c_j(\beta)$ as $\beta \to \infty$. We note that $B'(a,b,z)=B(a,b)-\frac{z^a}{a}+o(z^a)$, as $z\to 0$.  Then, we have
\small{\begin{align}
&B'\left(1+\frac{2}{\alpha_1},a-\frac{2}{\alpha_1},\frac{1}{1+\beta}\right)=B\left(1+\frac{2}{\alpha_1},a-\frac{2}{\alpha_1}\right)-\frac{1}{1+\frac{2}{\alpha_1}}\left(\frac{1}{\beta}\right)^{1+\frac{2}{\alpha_1}}+o\left(\left(\frac{1}{\beta}\right)^{1+\frac{2}{\alpha_1}}\right),\label{eq:highSIR_appro_a_1}\\
&B'\left(1+\frac{2}{\alpha_1},a-\frac{2}{\alpha_1},\frac{1}{1+\beta/T_j}\right)=B\left(1+\frac{2}{\alpha_1},a-\frac{2}{\alpha_1}\right)-\frac{1}{1+\frac{2}{\alpha_1}}\left(\frac{T_j}{\beta}\right)^{1+\frac{2}{\alpha_1}}+o\left(\left(\frac{1}{\beta}\right)^{1+\frac{2}{\alpha_1}}\right),\label{eq:highSIR_appro_a_t}\\
&B'\left(\frac{2}{\alpha_1},1-\frac{2}{\alpha_1},\frac{1}{1+\beta}\right)=B\left(\frac{2}{\alpha_1},1-\frac{2}{\alpha_1}\right)-\frac{\alpha_1}{2}\left(\frac{1}{\beta}\right)^{\frac{2}{\alpha_1}}+o\left(\left(\frac{1}{\beta}\right)^{\frac{2}{\alpha_1}}\right),\label{eq:highSIR_appro_1_1}\\
&B'\left(\frac{2}{\alpha_1},1-\frac{2}{\alpha_1},\frac{1}{1+\beta/T_j}\right)=B\left(\frac{2}{\alpha_1},1-\frac{2}{\alpha_1}\right)-\frac{\alpha_1}{2}\left(\frac{T_j}{\beta}\right)^{\frac{2}{\alpha_1}}+o\left(\left(\frac{1}{\beta}\right)^{\frac{2}{\alpha_1}}\right).\label{eq:highSIR_appro_1_t}
\end{align}}
\normalsize 
By \eqref{eq:highSIR_appro_a_1}-\eqref{eq:highSIR_appro_1_t}, as $\beta\to \infty$, we have
\small{\begin{align}
c(\beta)&=\frac{1}{n!}\binom{n}{n_{1},n_{2},n_{3}}\frac{n_{1}!}{\prod_{a=1}^{n_{1}}m_{a}!}\frac{n_{2}!}{\prod_{a=1}^{n_{2}}p_{a}!}\frac{n_{3}!}{\prod_{a=1}^{n_{3}}q_{a}!}\notag\\
&\hspace{0cm}\times\prod_{a=1}^{n_{1}}\left(\frac{2\pi\lambda_{1}}{\alpha_{1}}p_{\bar{c}}\left(U,T_{1},T_{2}\right)\left(\frac{ P_{1}}{P_{j}}\right)^{\frac{2}{\alpha_{1}}}\beta^{\frac{2}{\alpha_1}}\left(\frac{T_j^{1+\frac{2}{\alpha_1}}-1}{1+\frac{2}{\alpha_1}}\beta^{-1-\frac{2}{\alpha_1}}+o\left(\left(\frac{1}{\beta}\right)^{1+\frac{2}{\alpha_1}}\right)\right)\right)^{m_{a}}\notag\\
&\hspace{0cm}\times\prod_{a=1}^{n_{2}}\left(\frac{2\pi\lambda_{1}}{\alpha_{1}}\left(\frac{ P_{1}}{P_{j}}\right)^{\frac{2}{\alpha_{1}}}\beta^{\frac{2}{\alpha_1}}\left(B\left(1+\frac{2}{\alpha_{1}},a-\frac{2}{\alpha_{1}}\right)-\frac{T_j^{1+\frac{2}{\alpha_1}}}{1+\frac{2}{\alpha_1}}\beta^{-1-\frac{2}{\alpha_1}}+o\left(\left(\frac{1}{\beta}\right)^{1+\frac{2}{\alpha_1}}\right)\right)\right)^{p_{a}}\notag\\
&\hspace{0cm}\times\prod_{a=1}^{n_{3}}\left(\frac{2\pi\lambda_{2}}{\alpha_{2}}\left(\frac{ P_{2}}{P_{j}}\right)^{\frac{2}{\alpha_{2}}}\beta^{\frac{2}{\alpha_2}}\left(B\left(1+\frac{2}{\alpha_{2}},a-\frac{2}{\alpha_{2}}\right)-\frac{1}{1+\frac{2}{\alpha_2}}\beta^{-1-\frac{2}{\alpha_2}}+o\left(\left(\frac{1}{\beta}\right)^{1+\frac{2}{\alpha_2}}\right)\right)\right)^{q_{a}},\label{eq:order_c}
\end{align}
\begin{align}
c_{1}(\beta)&=\frac{2\pi\lambda_{1}}{\alpha_{1}}\left(\frac{P_{1}}{P_{j}}\right)^{\frac{2}{\alpha_{1}}}\left(\frac{\alpha_1}{2}p_{\bar{c}}\left(U,T_1,T_2\right)\left(T_j^{\frac{2}{\alpha_1}}-1\right)\beta^{-\frac{2}{\alpha_1}}
+B\left(\frac{2}{\alpha_{1}},1-\frac{2}{\alpha_{1}}\right)-\frac{\alpha_1}{2}\left(\frac{T_j}{\beta}\right)^{\frac{2}{\alpha_1}}+o\left(\left(\frac{1}{\beta}\right)^{\frac{2}{\alpha_1}}\right)\right),\label{eq:order_c_1}\\
c_{2}(\beta)&=\frac{2\pi\lambda_{2}}{\alpha_{2}}\left(\frac{P_{2}}{P_{j}}\right)^{\frac{2}{\alpha_{2}}}\left(B\left(\frac{2}{\alpha_2},1-\frac{2}{\alpha_2}\right)-\frac{\alpha_2}{2}\left(\frac{1}{\beta}\right)^{\frac{2}{\alpha_2}}+o\left(\left(\frac{1}{\beta}\right)^{\frac{2}{\alpha_2}}\right)\right).\label{eq:order_c_2}
\end{align}} 
\normalsize
We can obtain the order of $\beta$  for each term corresponding to a choice for $n$, $(n_a)_{a=1}^3\in\mathcal N_n$ and $(m_a)_{a=1}^{n_1}\in\mathcal M_{n_1}$ in  \eqref{eq:CPj_ub_1infty} as $\beta\to\infty$: $\beta^{-\left(1+\frac{2}{\alpha_{1}}\right)\sum_{a=1}^{n_{1}}m_{a}-\frac{2}{\alpha_{j}}}$, which can be maximized when $n_{1}=0$. Hence, we obtain the order of the upper bound: $\beta^{-\frac{2}{\alpha_{j}}}$. Moreover,  based on \eqref{eq:CPj_ub_1infty}, \eqref{eq:order_c}-\eqref{eq:order_c_2} and after some algebraic manipulation, we obtain the expressions of $\eta_1(U,T_{1},T_{2})$ and $\eta_2$.}

\subsubsection{Lower Bound}
First, we note that $\mathcal{S}_{j}\left(\beta,U,T_{1},T_{2}\right)$ can be rewritten as
\small{\begin{align}\label{eq:CPj_Sj_pdfYj_lb}
&\mathcal{S}_{j}\left(\beta,U,T_{1},T_{2}\right)=\int_{0}^{1}\mathcal{S}_{j,Y_{j}}\left(y,\beta,U,T_{1},T_{2}\right)f_{Y_{j}}(y){\rm d}y+\int_{1}^{\infty}\mathcal{S}_{j,Y_{j}}\left(y,\beta,U,T_{1},T_{2}\right)f_{Y_{j}}(y){\rm d}y\notag\\
&\int_{0}^{1}\mathcal{S}_{j,Y_{j}}\left(y,\beta,U,T_{1},T_{2}\right)f_{Y_{j}}(y){\rm d}y
 \stackrel{(c)}{>}\int_{0}^{1}\hat{\mathcal{S}}_{j,Y_{j}}\left(y,\beta,U,T_{1},T_{2}\right)\hat{f}_{Y_{j}}(y){\rm d}y\notag\\
 \eqld &\frac{\pi\lambda_{j}}{\mathcal{A}_{j}}\frac{\alpha_{\max}}{\alpha_{j}}\sum_{n=0}^{M_{j}-1}\sum_{(n_{a})_{a=1}^{3}\in\mathcal{N}_{n}}\sum_{(m_{a})_{a=1}^{n_{1}}\in\mathcal{M}_{n_{1}}}\sum_{(p_{a})_{a=1}^{n_{2}}\in\mathcal{M}_{n_{2}}}\sum_{(q_{a})_{a=1}^{n_{3}}\in\mathcal{M}_{n_{3}}}c(\beta)\notag\\
&\times\left(a_{1}+a_{2}+ c_{1}(\beta)\beta^{\frac{2}{\alpha_{1}}}+ c_{2}(\beta) \beta^{\frac{2}{\alpha_{2}}}\right)^{-\frac{\alpha_{\max}}{\alpha_{1}}\left(\sum_{a=1}^{n_{1}}m_{a}+\sum_{a=1}^{n_{2}}p_{a}\right)-\frac{\alpha_{\max}}{\alpha_{2}}\sum_{a=1}^{n_{3}}q_{a}-\frac{\alpha_{\max}}{\alpha_{j}}}\notag\\
&\times \gamma\left(\frac{\alpha_{\max}}{\alpha_{1}}\left(\sum_{a=1}^{n_{1}}m_{a}+\sum_{a=1}^{n_{2}}p_{a}\right)+\frac{\alpha_{\max}}{\alpha_{2}}\sum_{a=1}^{n_{3}}q_{a}+\frac{\alpha_{\max}}{\alpha_{j}},a_{1}+a_{2}+ c_{1}(\beta)\beta^{\frac{2}{\alpha_{1}}}+c_{2}(\beta) \beta^{\frac{2}{\alpha_{2}}}\right)
\end{align}}
\normalsize{where $\hat{\mathcal{S}}_{j,Y_{j}}\left(y,\beta,U,T_{1},T_{2}\right)=\mu_{j}(\beta,U,T_{1},T_{2})\hat{g}_{j}(y,\beta,U,T_{1},T_{2})$ with $\hat{g}_{j}(y,\beta,U,T_{1},T_{2})$\\$=\exp\left(-c_{1}\beta^{\frac{2}{\alpha_{1}}}y^{\frac{2\alpha_{j}}{\alpha_{\max}}}-c_{2}\beta^{\frac{2}{\alpha_{2}}}y^{\frac{2\alpha_{j}}{\alpha_{\max}}}\right)$, $\hat{f}_{Y_{j}}(y)=\frac{2\pi\lambda_{j}}{\mathcal{A}_{j}}y\exp\left(-a_{1}y^{\frac{2\alpha_{j}}{\alpha_{\max}}}-a_{2}y^{\frac{2\alpha_{j}}{\alpha_{\max}}}\right)$,  (c) is due to $g_j\left(y,\beta,U,T_{1},T_{2}\right)\geq\hat{g}_j\left(y,\beta,U,T_{1},T_{2}\right)$ and $f_{Y_{j}}(y)\geq\hat f_{Y_{j}}(y)$ when $y\in [0,1]$,  and (d) is due to $\int_{0}^{1}u^a\exp^{-bu}{\rm d}u=b^{-a-1}\gamma(a+1,b)$.  Here, $\gamma(a,z)\triangleq \int_0^z u^{a-1}e^{-u}{\rm d}u$ denotes the lower incomplete gamma function. Similar to the method in calculating the order of the upper bound, when $\beta\to\infty$, we can obtain the order of  $\beta$ for each term corresponding to  a choice for $n$, $(n_a)_{a=1}^3\in\mathcal N_n$, $(m_a)_{a=1}^{n_1}\in\mathcal M_{n_1}$, $(p_a)_{a=1}^{n_2}\in\mathcal M_{n_2}$ and $(q_a)_{a=1}^{n_3}\in\mathcal M_{n_3}$ in \eqref{eq:CPj_Sj_pdfYj_lb} as $$\beta^{-\sum_{a=1}^{n_{1}}m_{a}+\frac{2}{\alpha_{1}}\sum_{a=1}^{n_{2}}p_{a}+\frac{2}{\alpha_{2}}\sum_{a=1}^{n_{3}}q_{a}-\frac{2\alpha_{\max}}{\alpha_{\min}\alpha_{1}}\left(\sum_{a=1}^{n_{1}}m_{a}+\sum_{a=1}^{n_{2}}p_{a}\right)-\frac{2\alpha_{\max}}{\alpha_{\min}\alpha_{2}}\sum_{a=1}^{n_{3}}q_{a}-\frac{2}{\alpha_{\min}}\frac{\alpha_{\max}}{\alpha_{j}}},$$ which can be maximized when $n_{1}=n_{2}=n_{3}=0$, i.e., $n=0$. Hence, we obtain the order of the lower bound as $\beta^{-\frac{2}{\alpha_{j}}\frac{\alpha_{\max}}{\alpha_{\min}}}$. Moreover, based on (\ref{eq:betainc1}) and after some algebraic manipulation, we obtain the expression of $\xi_{j}$.}

\subsection{Proof of Theorem \ref{them:CP_highbeta_eql}}\label{proof:largebeta_eqlalpha}

When $\alpha_{1}=\alpha_{2}=\alpha$, from \eqref{eq:CPj_Sj_pdfYj}, we have 
\small{\begin{align}
&{\rm Pr}\left({\rm SIR}_{j,0}>\beta\right)\notag\\
&=\frac{2\pi\lambda_j}{\mathcal A_j}\sum_{n=0}^{M_{j}-1}\sum_{(n_{a})_{a=1}^{3}\in\mathcal{N}_{n}}\sum_{(m_{a})_{a=1}^{n_{1}}\in\mathcal{M}_{n_{1}}}\sum_{(p_{a})_{a=1}^{n_{2}}\in\mathcal{M}_{n_{2}}}\sum_{(q_{a})_{a=1}^{n_{3}}\in\mathcal{M}_{n_{3}}}c(\beta)\notag\\
&\hspace{5mm}\times\int_0^{\infty}y^{2\left(\sum_{a=1}^{n_1}m_a+\sum_{a=1}^{n_2}p_a+\sum_{a=1}^{n_3}q_a\right)+1}
\exp\left(-\left(c_1\beta^{\frac{2}{\alpha}}+c_2\beta^{\frac{2}{\alpha}}+a_1+a_2\right)y^2\right){\rm d}y\notag\\
&=\frac{\pi\lambda_j}{\mathcal A_j}\sum_{n=0}^{M_{j}-1}\sum_{(n_{a})_{a=1}^{3}\in\mathcal{N}_{n}}\sum_{(m_{a})_{a=1}^{n_{1}}\in\mathcal{M}_{n_{1}}}\sum_{(p_{a})_{a=1}^{n_{2}}\in\mathcal{M}_{n_{2}}}\sum_{(q_{a})_{a=1}^{n_{3}}\in\mathcal{M}_{n_{3}}}\Gamma\left(\sum_{a=1}^{n_1}m_a+\sum_{a=1}^{n_2}p_a+\sum_{a=1}^{n_3}q_a+1\right)c(\beta)\notag\\
&\hspace{5mm}\times\left(c_1(\beta)\beta^{\frac{2}{\alpha}}+c_2(\beta)\beta^{\frac{2}{\alpha}}+a_1+a_2\right)^{-\sum_{a=1}^{n_1}m_a-\sum_{a=1}^{n_2}p_a-\sum_{a=1}^{n_3}q_a-1}\label{eq:highSIR_equal_alpha}
\end{align}}\normalsize
 When $\beta\to\infty$, based on \eqref{eq:order_c}-\eqref{eq:order_c_2}, we can obtain the order of $\beta$ for each term corresponding to  a choice for $n$, $(n_a)_{a=1}^3\in\mathcal N_n$ and $(m_a)_{a=1}^{n_1}\in\mathcal M_{n_1}$ in \eqref{eq:highSIR_equal_alpha} as: $\beta^{-\left(1+\frac{2}{\alpha}\right)\sum_{a=1}^{n_{1}}m_{a}-\frac{2}{\alpha}}$, which can be maximized when $n_{1}=0$. Hence, the order is $\beta^{-\frac{2}{\alpha}}$. Moreover, after some algebraic manipulation, we can obtain the expressions of $c_1\left(U,T_1,T_2\right)$ and $c_2\left(T_1,T_2\right)$.

\subsection{Proof of Lemma \ref{lem:opt_U}}\label{proof:lem_optU-high}
We solve the optimization problem for $\alpha_1=\alpha_2$. When $\alpha_1\neq\alpha_2$, the optimization problem can be solved in a similar way and is omitted due to page limit. First, we rewrite  $c_{1}\left(U,T_{1},T_{2}\right)$ in \eqref{eq:c1} as $c_{1}\left(U,T_{1},T_{2}\right)=\frac{\pi\lambda_{1}}{\mathcal{A}_{1}}\sum_{u=0}^{U}{\rm Pr}\left(u_{{\rm IN},0}=u\right)f(u)$, where $f(u)$ denotes the expression after ${\rm Pr}\left(u_{{\rm IN},0}=u\right)$ in \eqref{eq:c1}. It can be easily verified  that $f(u)$ is a decreasing function of $u$. By {\em Lemma \ref{lem:pmf_uIN0}}, we have $c_{1}\left(U,T_{1},T_{2}\right)=\frac{\pi\lambda_{1}}{\mathcal{A}_{1}}\left(\sum_{u=0}^{U}{\rm Pr}\left(K_0=u\right)f(u)+\sum_{k=U+1}^{\infty}{\rm Pr}\left(K_0=k\right)f(U)\right)$.
Thus, we have $c_{1}\left(U+1,T_{1},T_{2}\right)-c_{1}\left(U,T_{1},T_{2}\right)=\frac{\pi\lambda_{1}}{\mathcal{A}_{1}}\left(f(U+1)-f(U)\right)\sum_{k=U+1}^{\infty}{\rm Pr}\left(K_0=k\right)<0$.
Therefore, we can show $U^{*}(\beta,T_{1},T_{2})=0$.


\begin{scriptsize}

\end{scriptsize}

\end{document}